\documentclass[prd,11pt,tightenlines,nofootinbib,superscriptaddress]{revtex4}
\usepackage{amsfonts,amssymb,amsthm,bbm,amsmath,eufrak}
\usepackage{verbatim}
\usepackage{graphicx}
\usepackage{subfig}
\usepackage{tikz}

\usepackage{hyperref}
\hypersetup{
    colorlinks,
    citecolor=black,
    filecolor=black,
    linkcolor=black,
    urlcolor=black
}

\newcommand{\C}{{\mathbb C}}
\newcommand{\N}{{\mathbb N}}

\newcommand{\one}{\mathbbm{1}}

\newcommand{\cH}{{\mathcal H}}

\newcommand{\id}{\mathbbm{1}}

\newcommand{\be}{\begin{equation}}
\newcommand{\ee}{\end{equation}}
\newcommand{\beq}{\begin{eqnarray}}
\newcommand{\eeq}{\end{eqnarray}}
\newcommand{\bea}{\begin{eqnarray}}
\newcommand{\eea}{\end{eqnarray}}
\newcommand{\nn}{\nonumber}

\newcommand{\bra}{\langle}
\newcommand{\ket}{\rangle}

\newcommand{\tr}{{\mathrm Tr}}

\newcommand{\rd}{\mathrm{d}}

\newcommand{\bpm}{\begin{pmatrix}}
\newcommand{\epm}{\end{pmatrix}}
\newcommand{\bvm}{\begin{vmatrix}}
\newcommand{\evm}{\end{vmatrix}}

\def\nn{\nonumber}

\newtheorem{theorem}{Theorem}[section]
\newtheorem{lemma}[theorem]{Lemma}
\begin{document}

\title{Pachner moves in a 4d Riemannian holomorphic Spin Foam model}

\author{{\bf Andrzej Banburski}\email{abanburski@perimeterinstitute.ca}}
\affiliation{ Perimeter Institute for Theoretical Physics,
Waterloo, Ontario, Canada.}
\affiliation{Department of Physics, University of Waterloo, Waterloo, Ontario, Canada}

\author{{\bf Lin-Qing Chen}\email{lchen@perimeterinsititute.ca}}
\affiliation{ Perimeter Institute for Theoretical Physics,
Waterloo, Ontario, Canada.}
\affiliation{Department of Physics, University of Waterloo, Waterloo, Ontario, Canada}

\author{{\bf Laurent Freidel}\email{lfreidel@perimeterinsititute.ca}}
\affiliation{ Perimeter Institute for Theoretical Physics,
Waterloo, Ontario, Canada.}

\author{{\bf Jeff Hnybida}\email{jhnybida@perimeterinsititute.ca}}
\affiliation{ Institute for Mathematics, Astrophysics, and Particle Physics
Faculty of Science, Radboud University, Nijmegen, Netherlands}

\begin{abstract}
In this work we study a  Spin Foam model for 4d Riemannian gravity, and propose a new way of imposing the simplicity constraints  that 
uses the recently developed holomorphic representation.
Using the power of the holomorphic integration techniques, 
and with the introduction of two new tools: the homogeneity map and the loop identity,  for the first time we give the analytic expressions for the behaviour of the Spin Foam amplitudes under 4-dimensional Pachner moves.
 It turns out that this behaviour is controlled by an insertion of nonlocal mixing operators.
 In the case of the  5--1 move, the expression governing the change of the amplitude can be interpreted as a vertex renormalisation equation.
 We find  a natural truncation scheme that allows us to get an invariance up to an overall factor for the 4--2 and 5--1 moves, but not for the 3--3 move. 
 The study of the divergences shows that there is a range of parameter space for which  the 4--2 move is finite while the 5--1 move diverges.
 This opens up the possibility  to recover diffeomorphism invariance in the   continuum limit of Spin Foam models for 4D Quantum Gravity. 
\end{abstract}

\maketitle
\tableofcontents
%%%%%%%%%%%%%%%%%%%%%%%%%%%%%%%%%%%%%%%%%%%%%%%%%%%%%%%%%%%%%%%%%%%%%%%%%%%%%%%%%%%%%%%

%%%%%%%%%%%%%%%%%%%%%%%%%%%%%%%%%%%%%%%%%%%%%%%%%%%
\section{Introduction}
%%%%%%%%%%%%%%%%%%%%%%%%%%%%%%%%%%%%%%%%%%%%%%%%%%%

Spin Foam models attempt to rigorously define a path integral for transition amplitudes in Quantum Gravity \cite{Rovelli:2011eq, Perez:2012wv}.  These models are well defined thanks to  a discrete regularization of spacetime.  The dynamics of this discrete structure is still not understood and is currently under intense investigation \cite{Mamone:2009pw, Riello:2013bzw,Dittrich:2014mxa,Bahr:2014qza}.  In this paper we study the Spin Foam  dynamics for one particular model with the goal of determining its continuum limit. 
We also develop along the way new strategies, tools and techniques that can be applied to a larger class of models.

Spin Foam models for gravity are usually constructed in analogy with a formulation of General Relativity, due to Plebanski \cite{Plebanski:1977zz}, as a constrained topological gauge theory known as BF theory.  Each Spin Foam model is a proposal for a discretized version of these constraints, known as simplicity constraints.  The most popular proposals are due to EPRL and FK \cite{Engle:2007wy, Engle:2007uq, Engle:2007qf, Freidel:2007py}.  In this paper we will focus on a Spin(4) Spin Foam model which is closely related to a model proposed by Dupuis and Livine \cite{Dupuis:2011fz}, for 4d Riemannian Quantum Gravity.  We choose this model because it is formulated in a holomorphic representation which is particularly suitable for performing analytical calculations.

The DL model arose from recent work in rewriting spin foams in a coherent state (or holomorphic) representation beginning with \cite{coh1, Livine:2007ya} and leading to many new tools and interesting results \cite{Freidel:2009nu, Freidel:2009ck, Freidel:2010tt, Borja:2010rc, Livine:2011gp, Dupuis:2012vp, Dupuis:2010iq, Dupuis:2011dh, Dupuis:2011wy}.  In the holomorphic representation complicated integrals over SU(2) group elements can be rewritten as simple integrals over the complex plane. This allows for exact evaluations of complicated spin network functions \cite{Freidel:2012ji}, and gives hope to study the dynamics of Spin Foam amplitudes analytically, and eventually numerically.

The natural path towards finding a continuum limit of a discrete theory involves studying coarse graining and applying renormalization methods. Note that already in flat spacetime lattice QCD \cite{mm}, this is non-trivial, as one needs to study the critical behaviour of the model. In order to obtain locally covariant continuum theory, it would seem that the usual global scale transformations might be not appropriate. Some early ideas \cite{Markopoulou:2002ja, Oeckl:2002ia, Oeckl:2004yf} in Spin Foam models have instead focused on
defining coarse graining via local scale transformations. A notion of refinement scale can be provided by embedding finer triangulations into coarser ones, while requiring so-called \emph{cylindrical consistency} \cite{Ashtekar:1994mh, Ashtekar:1994wa, Baratin:2010nn, Dittrich:2012jq, Bahr:2014qza}. Not much work has been done in this direction however, as the dynamics of Spin Foam models have not been understood beyond triangulations built out of more than few basic building blocks \cite{Mamone:2009pw, Riello:2013bzw}. Recently, a more global approach with the use of Tensor Network Renormalization scheme \cite{levin, Singh:2012np, Gu:2009dr} has been used to numerically study dimensionally reduced analogue Spin Foam models - so-called spin nets \cite{Dittrich:2011zh, Dittrich:2011av, Bahr:2012qj, Dittrich:2013uqe, Dittrich:2013aia, Dittrich:2013voa}, but the ideas have yet to be applied to full 4d models. Another approach under investigation is renormalization of Group Field Theories (GFT) \cite{Livine:2005tp, Freidel:2009hd, BenGeloun:2011rc, Rivasseau:2011hm, BenGeloun:2012pu, Geloun:2012bz, Carrozza:2012uv, Samary:2012bw, Carrozza:2013wda, Geloun:2013saa, Carrozza:2014rba}, which generate Spin Foam amplitudes. However, the renormalizable GFTs studied so far have not been of relevance to 4d gravity. It thus still seems crucial to understand what are the dynamics of Spin Foams for configurations that can be iteratively coarse grained.

The most basic local coarse graining move on a simplicial decomposition of a manifold is a specific type of so-called Pachner moves \cite{Pachner}. Pachner moves are local changes of triangulation that allow to go from some triangulation of a manifold to any other triangulation in a finite number of steps. In 4 dimensions there are three different Pachner moves: The 3--3 move, 4--2 move and  5--1 move (and their inverses). An $n$--$(2+d-n)$  Pachner move changes a triangulation composed out of $n$ d-simplices to one with $(2+d-n)$  d-simplices. Only the $n$--1 Pachner moves are pure coarse graining moves. The action of classical 4d Regge Calculus is known to be invariant under 5--1 and the 4--2 moves \cite{Dittrich:2011vz}, but it has been a long standing open problem to make any statement about invariance under these moves for non-topological Spin Foam models. Only the naive degree of divergence of the 5--1 move has been estimated for the EPRL model \cite{Perini:2008pd}. This question is not even obvious in linearized gravity, as the partition function of the quantum linearized Regge Calculus has recently been found to be not invariant, as it picks up a nonlocal measure factor \cite{Dittrich:2014rha}.

In this paper, we calculate the 4-dimensional Pachner moves for the first time in a Spin Foam model with simplicity constraints describing 4d Riemannian Quantum Gravity. We find that the model considered is not invariant under the 5--1 Pachner move, as the configuration of five 4-simplices reduces to a single 4-simplex with an insertion of a nonlocal operator inside. Similar behaviour occurs also for the 4--2 move. We conjecture that this is also the case for the other  Spin Foam models of 4d gravity studied in the literature and discuss the possible meaning of this operator and the necessity for truncation in defining coarse graining. We then find a natural  truncation scheme that allows us to make both the 4--2 and 5--1 moves invariant up to a weight depending on the boundary data. The 3--3 move is not invariant, unless very special symmetric boundary data are considered, as expected for a model of 4D Quantum Gravity.

The plan of the paper is as follows: in section \ref{sec:spinfoams} we review the construction of discretized BF theory and quantize it. We then present the diagrammatics and discuss their geometrical meaning. Next, we go on to reviewing the holomorphic representation, which simplifies the construction of SU(2) invariants. This allows us to write the partition function of BF theory in the holomorphic representation. We then review the holomorphic simplicity constraints. We show that these can be imposed in two different ways. The usual way is to impose them on the boundary spin network function, resulting in the Dupuis-Livine model \cite{Dupuis:2011fz, Dupuis:2011dh}, which is very similar to the EPRL-FK model. The alternative way is to impose the constraints onto a Spin(4) projector. Surprisingly, the two models have the same semi-classical limit for one 4-simplex \cite{spinfoamfactory} . In section \ref{sec:3dPachner} we review the calculation of Pachner moves in the Ponzano-Regge model for 3d Quantum Gravity. We start with defining the notion of Pachner moves. We then discuss the gauge fixing procedure for the internal rotational SU(2) symmetry. Next, we derive a crucial identity (which we refer to as the loop identity) that allows for the calculation of Pachner moves. We finish the section with calculating the 3--2 and 4--1 Pachner moves and discussing the fate of diffeomorphism symmetry. In section \ref{sec:4dPachner} we calculate the Pachner moves in the 4d Spin Foam model constructed from the constrained projectors. We obtain the loop identity for the constrained model and find that it exhibits extra mixing of strands in graphs, compared to the topological case. We then apply the loop identity to  the 3--3, 4--2 and the 5--1 moves, and find them to be not invariant. In section \ref{sec:coarse} we discuss the necessity for coarse graining and show that there is a natural truncation scheme, which makes the 4--2 and 5--1 moves invariant up to a factor, while keeping the 3--3 move not invariant. We then go on to calculate and analyze the degrees of divergence of the 4--2 and 5--1 moves, and find that for a range of parameters, the latter can be divergent, while the former can be convergent. We conclude by discussing the implications of these results and the connection to recovering diffeomorphism invariance in the continuum limit of Spin Foam models.

%%%%%%%%%%%%%%%%%%%%%%%%%%%%%%%%%%%%%%%%%%%%%%%%%%%
\section{Holomorphic Spin Foam Models}\label{sec:spinfoams}
%%%%%%%%%%%%%%%%%%%%%%%%%%%%%%%%%%%%%%%%%%%%%%%%%%%
In this section, we will start from reviewing the discretized path integral formalism of BF theory and introducing cable diagrams notation. We will then review the holomorphic representation and the simplicity constraints. After giving a brief review of the holomorphic Spin Foam model, which was proposed in \cite{Dupuis:2011fz, Dupuis:2011dh}, we will introduce an alternative model through a different imposition of these constraints. The section will end with introducing the homogeneity map, the key tool which makes the computation in this paper possible.

%%%%%%%%%%%%%%%%%%%%%%%%%%%%%%%%%%%%%%%%%%%%%%%%%%%
\subsection{BF theory and Cable Diagrams}
%%%%%%%%%%%%%%%%%%%%%%%%%%%%%%%%%%%%%%%%%%%%%%%%%%%
Spin Foam models originated from the insight that classical gravity can be described as a topological field theory (BF theory) with a simplicity constraint. As BF theory only has topological degrees of freedom, it can be easily quantized.  Spin Foam models are then a path integral quantization for BF theory, with simplicity constraints imposed at the quantum level.  In this section, we  introduce the discretized path integral formalism of BF theory, and also give an intuitive graphic notation.

Let $\Delta$ be a simplicial complex homeomorphic to a $d$-dimensional manifold $\mathcal{M}$ and let $\Delta^\ast$ be its dual 2-complex. The partition function of SU(2) BF theory is defined in terms of the edges $e$ and faces $f$ of $\Delta^\ast$ by 
\begin{equation}
\mathcal{Z}_{BF} (\mathcal{M}) =\int \prod _{e\in \Delta^*} dg_e \prod_{f\in \Delta^*} \delta(g_{e_1}...g_{e_n}) .
\end{equation}
The $\delta$ functions for each face can be expanded in representations $j_f$ using the Peter-Weyl theorem as $\delta(g_{e_1}...g_{e_n}) = \sum_{j_f} (2j_f+1) \text{tr}_{j_f} \left(g_{1} \cdots g_{e_n} \right)$, where the trace is over the representation $j_f$.
Inserting the resolution of identity on the representation space $V^{j_{f_1}} \otimes \cdots \otimes V^{j_{f_d}}$ between each group element in the trace, we can write
%\begin{equation} \label{eqn_BF_intertwiners}
%\mathcal{Z}_{BF} (\mathcal{M})=\sum_{j_f} \prod_{f \in \Delta^\ast} (2j_f +1) \prod_{e \in \Delta^\ast}  P^{j_{f_1},...,j_{f_d}},
%\end{equation}
where $P^{j_{f_1},...,j_{f_d}}$ is the projector onto the SU(2) invariant subspace of $V^{j_{f_1}} \otimes \cdots \otimes V^{j_{f_d}}$ given by group averaging the tensor product of irreducible representations $\rho^j : V^j \rightarrow V^j$
\be \label{eqn_Haar_projector}
  P^{j_{f_1},...,j_{f_d}} = \int \rd g_e \: \rho^{j_{f_1}}(g_e) \otimes \cdots \otimes \rho^{j_{f_d}}(g_e) ,
\ee
where the basis labels of the representations $\rho^{j_{f}}$ have been supressed.  The projector is the unique map $P^{j_1,...,j_d} : \rho^{j_1} \otimes ....\otimes \rho^{j_d} \rightarrow \text{Inv}_\text{SU(2)}[\rho^{j_1} \otimes ....\otimes \rho^{j_d} ]$, and is often called the Haar projector.  For more details, see for example \cite{Perez:2012wv}.

Cable diagrams are an  intuitive and useful graphic notation, for  the computations of Spin Foam partition functions (a good review of these techniques  is given in \cite{Oeckl:2002ia}). Here it is used to represent the structure of partition function on the dual 2-complex $\Delta^\ast$. Cable diagrams are basically composed by strands passing through boxes: a strand denotes a representation of a symmetry group living on the edge e of $\Delta^\ast$, and a box denotes  the group averaging of a set of representations in the projector.
\be \label{cable}
\rho^{j} = 
\begin{tikzpicture}[baseline=0,scale=0.45]
  \draw (-2,0) -- (2,0); 
  \node at (0,0.5) {$j$}; 
\end{tikzpicture}
\qquad \text{and} \qquad
P^{j_{1},...,j_{d}} = 
\begin{tikzpicture}[baseline=0,scale=0.4]
  \node at (-2,2) {$j_1$}; \draw (-3,1.5) -- (-1,1.5);
  \node at (-2,1) {$j_2$}; \draw (-3,0.5) -- (-1,0.5);
  \node at (-2,0) {$j_3$}; \draw (-3,-0.5) -- (-1,-0.5);
  \node at (-2,-1) {$j_4$}; \draw (-3,-1.5) -- (-1,-1.5);
  \draw (-1,-2) --(-1,2) -- (1,2) -- (1,-2) -- (-1,-2);
  \draw (1,1.5) -- (3,1.5); \node at (2,2) {$j_1$}; 
  \draw (1,0.5) -- (3,0.5); \node at (2,1) {$j_2$}; 
  \draw (1,-0.5) -- (3,-0.5); \node at (2,0) {$j_3$}; 
  \draw (1,-1.5) -- (3,-1.5); \node at (2,-1) {$j_4$}; 
\end{tikzpicture}
\ \ \ (d=4)
\ee

%Graphically we represent each face of $\Delta^\ast$ by a line labelled by the representation $j_f$.  is denoted by a box as

Strands form closed loops, which correspond to the faces in the dual 2-complex $\Delta^\ast$. Fig.\ref{dual_tetrahedron} and Fig.\ref{dual_cable_tetrahedron} give an example in 3 dimensions: a 3-simplex, its dual 2-complex and the corresponding cable diagram.

\begin{figure} [h]
\centering
\parbox{6cm}{
\includegraphics[width=0.3\textwidth]{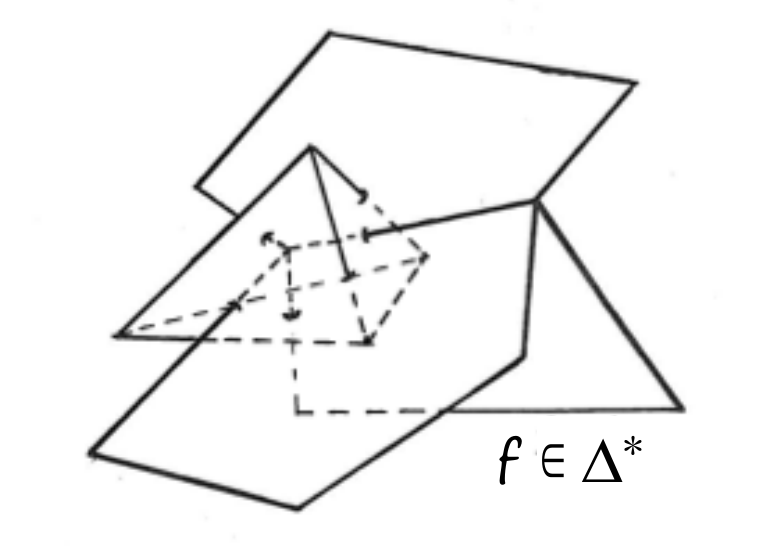}
\caption{A tetrahedron and its dual 2-complex}
\label{dual_tetrahedron}}
\qquad
\begin{minipage}{5 cm}
\includegraphics[width=0.9\textwidth]{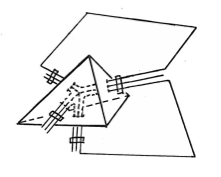}
\caption{A tetrahedron and its cable diagram}
\label{dual_cable_tetrahedron}
\end{minipage}
\end{figure}

The projector $P^{j_1,..,j_d}$ can be expressed as a sum over a basis of invariant tensors called intertwiners, as is usually done in Spin Foam models:
\be
  P^{j_1,..,j_d} = \sum_{\iota} \|j_i, \iota \ket \bra j_i, \iota\|
  = \sum_{\iota} \:
  \begin{tikzpicture}[baseline=0,scale=0.45]
  \node at (-3.5,1.5) {$j_1$}; \draw (-3,1.5) -- (-1,0);
  \node at (-3.5,0.5) {$j_2$}; \draw (-3,0.5) -- (-1,0);
  \node at (-3.5,-0.5) {$j_3$}; \draw (-3,-0.5) -- (-1,0);
  \node at (-3.5,-1.5) {$j_4$}; \draw (-3,-1.5) -- (-1,0);
  \draw (1,0) -- (3,1.5); \node at (3.5,1.5) {$j_1$}; 
  \draw (1,0) -- (3,0.5); \node at (3.5,0.5) {$j_2$}; 
  \draw (1,0) -- (3,-0.5); \node at (3.5,-0.5) {$j_3$}; 
  \draw (1,0) -- (3,-1.5); \node at (3.5,-1.5) {$j_4$}; 
  \node at (-0.5,0) {$\iota$}; \node at (0.5,0) {$\iota$};
\end{tikzpicture}
\ee
where $\iota$ labels a basis of normalized intertwiners. We see that the projector factorizes on the edges, while the intertwiners contract at the vertices of $\Delta^\ast$ expressing the partition function in terms of so called vertex amplitudes. For example, if $\mathcal{M}$ is 4 dimensional, the BF partition function can be writen as
\be\label{eqn_BF_intertwiners}
\mathcal{Z}_{BF} (\mathcal{M})=\sum_{j_f} \prod_{f \in \Delta^\ast} (2j_f +1) \sum_{\iota_e} \prod_{v \in \Delta^\ast} 
\begin{tikzpicture}[baseline=-1,scale=1.3]
  \draw (0,1) -- (0.951,0.309) -- (0.588,-0.809) -- (-0.588,-0.809) -- (-0.951,0.309) -- (0,1) -- (0.588,-0.809) -- (-0.951,0.309) -- (0.951,0.309) -- (-0.588,-0.809) -- (0,1);
  \node at (0,1.15) {$\iota_{e_1}$};
  \node at (1.151,0.409) {$\iota_{e_2}$};
  \node at (0.788,-1.00) {$\iota_{e_3}$};
  \node at (-0.738,-0.970) {$\iota_{e_4}$};
  \node at (-1.151,0.409) {$\iota_{e_5}$};
\end{tikzpicture} .
\ee
This is perhaps the most familiar form of BF theory.  In 3d the vertex amplitudes of the Ponzano-Regge model are 6j symbols with no intertwiner labels since rank three intertwiners are one dimensional.  In 4d the vertex amplitudes of the Oorguri model are 15j symbols labelled by 10 spins and 5 intertwiner labels, which are also spins. In this work, we will be most interested in the coherent state representation which can be defined in any dimension.  The Livine-Speziale coherent intertwiners form an over complete basis and are labelled by a set of $d$ spinors $\{z_i\}_{i=1}^{d}$.  The vertex amplitude therefore depends on $d(d+1)$ spinors.

%\be
%\mathcal{Z}_{BF} (\mathcal{M})=\sum_{j_f} \prod_{f \in \Delta^\ast} (2j_f +1) %\sum_{S_e} \prod_{v \in \Delta^\ast} \{15j_v(j_f,S_e)\}
%\ee
%{\color{blue} I might try to draw the 15j into the equation using the package Tikz.}

%The second reason for introducing intertwiners is that they have a nice geometrical interpretation.  In the coherent state representation... {\color{blue} Explain the imposition of simplicity constraints by identifying geometrical properties of intertwiners}

%\begin{figure} [h]
%\centering
%\parbox{5cm}{
%\includegraphics[width=0.3\textwidth]{4simplices.png}
%\caption{A 4-simplex,  Need to be redrawed}
%\label{fig:2figsA}}
%\qquad
%\begin{minipage}{5cm}
%\includegraphics[width=0.9\textwidth]{4simplex.pdf}
%\caption{Cable graph for a 4-simplex.}
%\label{fig:2figsB}
%\end{minipage}
%\end{figure}

%%%%%%%%%%%%%%%%%%%%%%%%%%%%%%%%%%%%%%%%%%%%%%%%%%%
\subsection{The Holomorphic Representation and Diagramatics}
%%%%%%%%%%%%%%%%%%%%%%%%%%%%%%%%%%%%%%%%%%%%%%%%%%%
We choose to use a spinor representation of SU(2) in the Bargmann-Fock space $L^2_{hol}(\C^2,d\mu)$  of holomorphic polynomials of a spinor \cite{Bargmann, Schwinger, Livine:2011gp}.  One of the features of this representation that will facilitate our calculations is that the Hermitian inner product is Gaussian:
\be \label{barg_in_prod}
  \bra f | g \ket = \int_{\C^2} \overline{f(z)} g(z) \rd\mu(z),
\ee
where $\rd\mu(z) = \pi^{-2} e^{-\bra z | z \ket}  \rd^{4}z$ and $\rd^{4}z$ is the Lebesgue measure on $\C^2$.

Given $z \in \C^2$ we denote its conjugate by $\check{z}$.  We use a bra-ket notation for $z$ and square brackets $\check{z}$ as in
\be
  |z\ket = \bpm z_0 \\ z_1 \epm, \qquad |z] = \bpm -\overline{z}_1 \\ \overline{z}_0 \epm .
\ee
That is $|\check{z}\ket = |z]$.  Notice that while $\bra z|$ is anti-holomorphic, $[z|$ is holomorphic and orthogonal to $|z\ket$, i.e. $[z|z\ket = 0$. This non-standard notation for spinors will turn out to be useful, as we will always work with contractions of spinors, without the need for writing out the indices. Our notation is related to the usual one as follows: $z_A = | z \ket$, $\bar{z}_{A'} = \bra z |$, and the spinor invariants are $[z|w\ket = z_{A'}w_A\epsilon^{A' A}$ and $\bra z|w\ket = \bar{z}_{A'}w_A\delta^{A' A}$.
The bracket $[z|w\ket$ associated with the $\epsilon$ tensor is skew-symmetric, holomorphic and SL$(2,\mathbb{C})$ invariant.
The bracket $\bra z|w\ket$ associated with the identity  tensor is hermitian, and only SU$(2)$ invariant.

Let us now study the identity on the Bargmann-Fock space $L^2_{hol}(\C^2,d\mu)$. The delta distribution on this space is given by $\delta_{w}(z) = e^{\bra z | w \ket}$, since for any holomorphic function $\int \rd \mu(z) f(z) e^{\bra z | w \ket} = f(w)$.  Let us use a line to represent the  delta  graphically by
\be
e^{\bra z|w\ket} = 
\begin{tikzpicture}[baseline=0,scale=0.45]
  \node at (-3,0) {$\bra z |$}; \draw (-2,0) -- (2,0); 
  \node at (3,0) {$|w\ket$}; 
\end{tikzpicture} \qquad \text{and} \qquad
%\ee
%and the quantity
%\be
\frac{\bra z|w\ket^{2j}}{(2j)!} = 
\begin{tikzpicture}[baseline=0,scale=0.45]
  \node at (-3,0) {$\bra z |$}; \draw (-2,0) -- (2,0); 
  \node at (3,0) {$|w\ket$}; 
  \node at (0,1) {$j$};
\end{tikzpicture}.
\ee
Therefore the Gaussian integral $\int \rd \mu(w) e^{\bra z|w\ket + \bra w | z'\ket} = e^{\bra z|z'\ket}$ implies the contraction
\be \label{eqn_contraction}
\int \rd \mu(w)
\begin{tikzpicture}[baseline=0,scale=0.45]
  \node at (-3,0) {$\bra z|$}; \draw (-2,0) -- (2,0); 
  \node at (3,0) {$|w\ket$}; 
  \node at (0,1) {$j$};
  \node at (4,0) {$\bra w |$}; \draw (5,0) -- (9,0); 
  \node at (10,0) {$|z'\ket$}; 
  \node at (7,1) {$j'$};
\end{tikzpicture}
= \delta_{j,j'}
\begin{tikzpicture}[baseline=0,scale=0.45]
  \node at (-3,0) {$\bra z|$}; \draw (-2,0) -- (2,0); 
  \node at (3,0) {$|z'\ket$}; 
  \node at (0,1) {$j$};
\end{tikzpicture}.
\ee

For a function of four spinors (with obvious generalization to $n$ spinors) we can thus define the \emph{trivial projector}, which we will denote as
\be \label{trivial_proj}
  \one(z_i;w_i) = e^{\sum_{i=1}^{4} [z_i|w_i\ket} = 
  \begin{tikzpicture}[baseline=0,scale=0.45]
  \node at (-4,1.5) {$[z_1|$}; \draw (-3,1.5) -- (3,1.5);
  \node at (-4,0.5) {$[z_2|$}; \draw (-3,0.5) -- (3,0.5);
  \node at (-4,-0.5) {$[z_3|$}; \draw (-3,-0.5) -- (3,-0.5);
  \node at (-4,-1.5) {$[z_4|$}; \draw (-3,-1.5) -- (3,-1.5);
  \node at (4,1.5) {$|w_1\ket$}; 
  \node at (4,0.5) {$|w_2\ket$}; 
  \node at (4,-0.5) {$|w_3\ket$}; 
  \node at (4,-1.5) {$|w_4\ket$}; 
\end{tikzpicture}.
\ee

Next, we will study how SU(2) acts on the elements of the Bargmann-Fock space. For a generic holomorphic function $f\in L^2_{hol}(\C^2,d\mu)$, the group action is given by
\be
  g \cdot f(z) = f(g^{-1}z).
\ee
The group SU(2) acts irreducibly on the subspaces of holomorphic polynomials homogeneous of degree $2j$. Holomorphic polynomials with different  degrees of homogeneity are orthognal with each other. Indeed,
$L^2_{hol}(\C^2,d\mu)=\bigoplus_{j\in \N /2} V^j$ and
an orthonormal basis of $V^j$ is given by 
\be
  e^{j}_{m}(z) \equiv \frac{z_{0}^{j+m} z_{1}^{j-m}}{\sqrt{(j+m)!(j-m)!}}
\ee
and it is of dimension $2j+1$.

In the study of gauge-invariant Spin Foam models, we will be interested in the SU(2) invariant functions on $n$ spinors 
\be
f(g z_1, g z_2, ... , g z_n) = f(z_1, z_2, ... , z_n)\,, \quad \forall g \in {\rm SU(2)}.
\ee 
We will denote the invariant elements of $L^2(\C^2,\rd\mu)^{\otimes n}$ to be in $\cH_n$, which is the Hilbert space of $n$-valent intertwiners:
\be
  \cH_n = \bigoplus_{j_i} \cH_{j_1,...,j_n}\equiv \bigoplus_{j_i} \text{Inv}_{\text{SU}(2)}\left[V^{j_1} \otimes \cdots \otimes V^{j_n} \right] .
\ee
One way to construct an element of $\cH_n$ is to average a function of $n$ spinors over the group using the Haar measure.  In this way we can construct a projector $P :L^2(\C^2,\rd\mu)^{\otimes n} \rightarrow \cH_n$ which is called the Haar projector as
\be
  P(f)(w_i) = \int \prod_i \rd\mu(z_i) P(\check{z}_i;w_i) f(z_1, z_2, ... , z_n) = \int_{\text{SU(2)}} \rd g f(g w_1, g w_2, ... , g w_n)
\ee
where the kernel is given by\footnote{For a review of Guassian integration techniques see Appendix \ref{app_gauss}.}
\be
  P(z_i;w_i) = \int_{\text{SU}(2)} \rd g \,e^{\sum_i [z_i|g|w_i\ket} = 
\begin{tikzpicture}[baseline=0,scale=0.45]
  \node at (-4,1.5) {$[z_1|$}; \draw (-3,1.5) -- (-1,1.5);
  \node at (-4,0.5) {$[z_2|$}; \draw (-3,0.5) -- (-1,0.5);
  \node at (-4,-0.5) {$[z_3|$}; \draw (-3,-0.5) -- (-1,-0.5);
  \node at (-4,-1.5) {$[z_4|$}; \draw (-3,-1.5) -- (-1,-1.5);
  \draw (-1,-2) --(-1,2) -- (1,2) -- (1,-2) -- (-1,-2);
  \draw (1,1.5) -- (3,1.5); \node at (4,1.5) {$|w_1\ket$}; 
  \draw (1,0.5) -- (3,0.5); \node at (4,0.5) {$|w_2\ket$}; 
  \draw (1,-0.5) -- (3,-0.5); \node at (4,-0.5) {$|w_3\ket$}; 
  \draw (1,-1.5) -- (3,-1.5); \node at (4,-1.5) {$|w_4\ket$}; 
\end{tikzpicture},
\label{proj}
\ee
 where we use a box to represent group averaging with respect to the Haar measure over SU(2). Hence the projector onto the invariant subspace is simply the group average of $\one(z_i;w_i)$. From the above, we see that a contraction of two spinors on the same strand but belonging to two different projectors is obtained by setting $z^1_i = \check{w}^2_i$. This implies that the kernel of the  projector satisfies the projection property
\be\label{eq:PPisP}
 \int \prod_i \rd \mu(w_i) P(z_i;w_i) P(\check{w}_i;z'_i) = P(z_i;z'_i) .
\ee
 We will also refer from now on to the kernel $P(z_i;w_i)$ as a projector for convenience.   As shown in \cite{Freidel:2010tt, Freidel:2012ji}, we can perform the integration over $g$  in Eq. (\ref{proj}) explicitly, which gives a power series in the holomorphic spinor invariants:
\be  \label{proj2}
  P(z_i;w_i) = \sum_{[k]} \frac{1}{(J+1)!} \prod_{i<j} \frac{([z_i|z_j\ket[w_i|w_j\ket)^{k_{ij}}}{k_{ij}!} ,
\ee
where the sum is over a set of $n(n-1)/2$ non-negative integers $[k]\equiv (k_{ij})_{i\neq j = 1,\cdots, n}$  with $1\leq i<j\leq n$ and $k_{ij}=k_{ji}$. A short proof of this  statement is  given in the Appendix \ref{app_thm_proj} for the reader's convenience. Thus a basis of n-valent intertwiners  is given by
\be\label{C}
  ( z_{i} | k_{ij} \ket  \equiv \prod_{i<j}\frac{ [z_{i}|z_{j}\ket^{k_{ij}}}{k_{ij}!}.
\ee
The non-negative integers $(k_{ij})_{i\neq j = 1,\cdots, n}$ are satisfying the $n$ homogeneity conditions 
\be\label{kj}
\sum_{j\neq i} k_{ij} =2j_{i}.
\ee

The sum of spins at the vertex is defined by $J = \sum_i j_i = \sum_{i<j} k_{ij}$ and is required to be a positive integer.  We also see from Eq. (\ref{proj2})  that the identity on ${\cal H}_{j_{i}}$ is resolved as follows
\be \label{eqn_res_id_disc}
  \one_{{\cal H}_{j_{i}}} = \sum_{[k]\in K_{j}} \frac{| k_{ij} \ket \bra k_{ij} | }{||k_{ij}||^{2}}, \qquad \|k_{ij}\|^{2} =\frac{ (J+1)!}{\prod_{i<j}k_{ij}!}.
\ee
with the set $K_j$ defined by integers $k_{ij}$ satisfying Eq.(\ref{kj}). For more details on these intertwiners and the coherent states defined by them, see \cite{Freidel:2013fia} where this basis was introduced for the first time.

Before we go on to the discussion of simplicity constraints, let us notice that using a multinomial expansion Eq.(\ref{proj2}) can be writen in terms of total spin:
\be \label{eqn_res_id_UN}
  P(z_i;w_i) = \sum_{J=0}^\infty \frac{\left( \sum_{i<j} [z_i|z_j\ket[w_i|w_j\ket\right)^J}{J!(J+1)!},
\ee
which will turn out to be a quite useful expression for the projector for computation purposes. Note that this is an expansion in U(N) coherent intertwiners of total area $J$.

%%%%%%%%%%%%%%%%%%%%%%%%%%%%%%%%%%%%%%%%%%%%%%%%%%%
%%%%%%%%%%%%%%%%%%%%%%%%%%%%%%%%%%%%%%%%%%%%%%%%%%%
\subsection{Holomorphic Simplicity Constraints}
%%%%%%%%%%%%%%%%%%%%%%%%%%%%%%%%%%%%%%%%%%%%%%%%%%%
%%%%%%%%%%%%%%%%%%%%%%%%%%%%%%%%%%%%%%%%%%%%%%%%%%%
%Holomorphic simplicity constraints for spinorial Spin Foam models were first  introduced in \cite{Dupuis:2011fz} for Riemannian gravity. Here we present these constraints, but refer the reader to \cite{Dupuis:2011dh} for their derivation.  We will also be implimenting the constraints in a slightly different manner that intended there, as we will explain below.

Holomorphic simplicity constraints for spinorial Spin Foam models were first  introduced in \cite{Dupuis:2011fz} for Riemannian gravity. Here we give a short summary, but refer the reader to the original paper for their full derivation.

For the Riemannian 4d Spin Foam models, we use the gauge group Spin(4) = SU(2)$_L$ $\times$ SU(2)$_R$, which is the double cover of SO(4).   The holomorphic simplicity constraints are isomorphisms between the two representation spaces of SU(2): for any two edges $i,j$ that are a part of  the same vertex $a$, they are defined by
\begin{equation}
[z^a_{iL} | z^a_{jL} \ket = \rho^2 [z^a_{iR} | z^a_{jR} \ket,
\label{hsc}
\end{equation}
where $\rho$ is a function of the Immirzi parameter $\gamma$ given by
\begin{equation}
\rho^2 = \left\{ 
  \begin{array}{ll}
 (1-\gamma)/(1+\gamma), & \quad  |\gamma| < 1\\
 (\gamma-1)/(1+\gamma), & \quad  |\gamma| > 1
  \end{array} \right .
\end{equation}
The holomorphic simplicity constraints Eq.(\ref{hsc}) essentially tell us that there exists a unique group element $g_a \in \mathrm{SL}(2,\C)$ for each vertex $a$, such that
\begin{equation}\label{eq:simplicityonz}
\forall i, \ \ \  g_{a} |z_{iL}^a \ket = \rho\  |z_{iR}^a \ket.
\end{equation}

A general element of $\mathrm{SL}(2,\C)$ can be decomposed into the product of an hermitian matrix times an element of $\mathrm{SU}(2)$, so that 
$g_a = h_a u_a$ with $h^\dagger_a = h_a$.
It is only when $h_a=\one$ that the holomorphic simplicity constraints imply the usual geometrical simplicity constraints.
In the FK formulation of the spin foam model which is only partially holomorphic this is implied since the norm of the spinors is fixed. 
The fully  holomorphic formulation of DL therefore relaxes at the quantum level the simplicity constraints. Fortunately, one one can check following \cite{Freidel:2013fia} that in the semi-classical limit of Holomorphic amplitudes the Gauss constraints due to the gauge invariance of the amplitude can be realized in the form 
\be
\sum_i |z_{iL}^a \ket\bra z_{iL}^a | = A_L \one, \quad  \sum_i |z_{iR}^a \ket\bra z_{iR}^a | = A_R \one.
\ee
This imposes that in the classical limit $h_a =\one$ and the geometrical simplicity constraints  $u_{a} |z_{iL}^a \ket = \rho\  |z_{iR}^a \ket$ with $u_a \in \mathrm{SU}(2)$ are satisfied.

Let us recall the geometrical meaning of these:
Each spinor defines a three vector $\vec V (z) \in \bold{R}^3$ through the equation
\begin{equation}
|z \ket \bra z| =\frac{1}{2} \left(   1 \!\! 1  \bra z | z \ket  + \vec V(z)  \cdot \vec\sigma \right),\qquad
|z ] [ z| =\frac{1}{2} \left( 1 \!\! 1   [ z | z ]   - \vec V(z)  \cdot \vec\sigma \right)
\end{equation}
where $\vec\sigma$ is the vector made by Pauli matrices. Thus around a vertex in a spin-network, each link  dual to a triangle in the simplicial manifold, is associated with two 3-vectors $\vec V_L(z)$ and $\vec V_R(z)$ given by the left and right spinors.  Classically, they correspond to the selfdual $b_+$ and anti-selfdual $b_-$ components of the $B$ field respectively :
\begin{equation}
V^i_L(z) = b_+^i := B^{0i} +\frac{1}{2} \epsilon^i_{kl} B^{kl}, \ \ \ \  V^i_R (z) =b_-^i :=  -B^{0i} +\frac{1}{2} \epsilon^i_{kl} B^{kl} .
\end{equation}
Note here that the time norm is chosen to be $N_I=(1,0,0,0)^T$. For the Hodge dual of the B field, we find $(\ast b)_+ = b_+ = \vec V _L(z)$, and $(\ast b)_- = - b_- = - \vec V_R (z) $.

For the vectors  $\vec V_L(z)$ and $\vec V_R(z)$  defined by the spinors of the two copies of SU(2)  this means that the holomorphic simplicity constraints imply
\begin{equation}
 g_{a} \triangleright  \vec V_L(z_{i}^a)  = \rho^2  \vec V_R(z_{i}^a),\ \ \ \ \ \forall i \in a
\end{equation}
which leads to the constraint that the norm of the selfdual and anti-selfdual components of the bivector $(g_{a}, \one) \triangleright  (B+ \gamma \ast B)\ $ have to be equal to each other:
\begin{equation}
|(1+\gamma) g_{a} \triangleright b_+| = |(1-\gamma)  \triangleright b_-|  .
\end{equation}
Thus $B$ and $\ast B$ are simple bivectors, and for the spin network vertex  $a$, there exists a common time norm to all the bivectors:
\begin{equation}\label{eq:4dnormal}
\mathcal{N}_a = (g_{a}, \one)^{-1} \triangleright (1,0,0,0) .
\end{equation}
The existence of this shared time norm implies the linear simplicity constraints introduced by the EPRL and FK models \cite{Freidel:2007py, Engle:2007uq, Engle:2007qf, Engle:2007wy}.

It is interesting to note that $g_a$ can be expressed purely in terms of spinors as 
\begin{equation}
 g_{a} =\frac{ |z_{iR}^a \ket  \bra z^a_{iL}|  + | z^a_{iR} ] [z^a_{iL}|}{\sqrt{\bra z^a_{iL}|z^a_{iL} \ket  \bra z^a_{iR}|z^a_{iR} \ket }}, \ \ \ \forall i\in a .
\label{g}
\end{equation}
It is easy to check that this satisfies Eq. (\ref{eq:simplicityonz}). Note here that $g_a$ is a unique group element for all strands belonging to the same vertex.

%%%%%%%%%%%%%%%%%%%%%%%%%%%%%%%%%%%%%%%%%%%%%%%%%%%
\subsection{Imposing constraints}
%%%%%%%%%%%%%%%%%%%%%%%%%%%%%%%%%%%%%%%%%%%%%%%%%%%
We will now impose the holomorphic simplicity constraints on the Spin(4) BF theory in order to obtain a model of 4d Riemannian Quantum Gravity. There are two natural ways of imposing these constraints - either on the boundary spin network defined by contraction of coherent states \cite{Dupuis:2011fz}, or on the whole projector (\ref{proj}). We will first summarize the usual approach, which we will refer to as the DL prescription. Then we will introduce an alternative model in which the constraint is imposed on the whole projector. It is very surprising that the alternative model actually has the same asymptotic behavior \cite{spinfoamfactory} as the DL prescription and EPRL/FK model (with $|\gamma|<1$)\cite{Conrady:2008mk, Barrett:2009gg, Han:2011rf}, i.e. the amplitude is weighted by a cosine of the Regge \cite{Regge:1961px}  action. It leads however to a much simpler calculation when we evaluate the Pachner moves than the DL case. Even though there is no technical obstacle to use the DL prescription, we will study it in a subsequent article, and focus on the constrained projector model in this paper.
%%%%%%%%%%%%%%%%%%%%%%%%%%%%%%%%%%%%%%%%%%%%%%%%%%%
\subsubsection{DL prescription  \label{sec:DLprojector}}
%%%%%%%%%%%%%%%%%%%%%%%%%%%%%%%%%%%%%%%%%%%%%%%%%%%
In \cite{Dupuis:2011fz, Dupuis:2011dh} Dupuis and Livine  introduced a Spin Foam model similar to the EPRL/FK models, but written in terms of spinorial coherent states with the holomorphic simplicity constraints. Since BF amplitudes can be seen as evaluations of spin network functions, the simplicity constraints in this model are imposed in the usual way -- on the boundary spin network given by the amplitude. The amplitude for a single 4-simplex $\sigma$ is given by a product of contraction of coherent states for left and right sectors, with the simplicity constraints imposed on the boundary spinors as follows
\be
\mathcal{A}_\sigma(\{z_\Delta^\tau\}) = \int \left[ \rd g^L_\tau\right]^5\left[ \rd g^R_\tau\right]^5 e^{\sum_{\Delta\in\sigma} \rho^2 [z_\Delta^{s(\Delta)}|g_{s(\Delta)}^{L\ \ \ -1}g_{t(\Delta)}^L|z_\Delta^{t(\Delta)}\ket + [z_\Delta^{s(\Delta)}|g_{s(\Delta)}^{R\ \ \ -1}g_{t(\Delta)}^R|z_\Delta^{t(\Delta)}\ket }
\ee
where $\Delta$ label different triangles/strands and $\tau , s(\Delta), t(\Delta)$ label tetrahedra/projectors. Graphically this is presented in Fig. \ref{fig:DLamplitude}. 
\begin{figure}[h]
	\centering
		\includegraphics[width=0.4\textwidth]{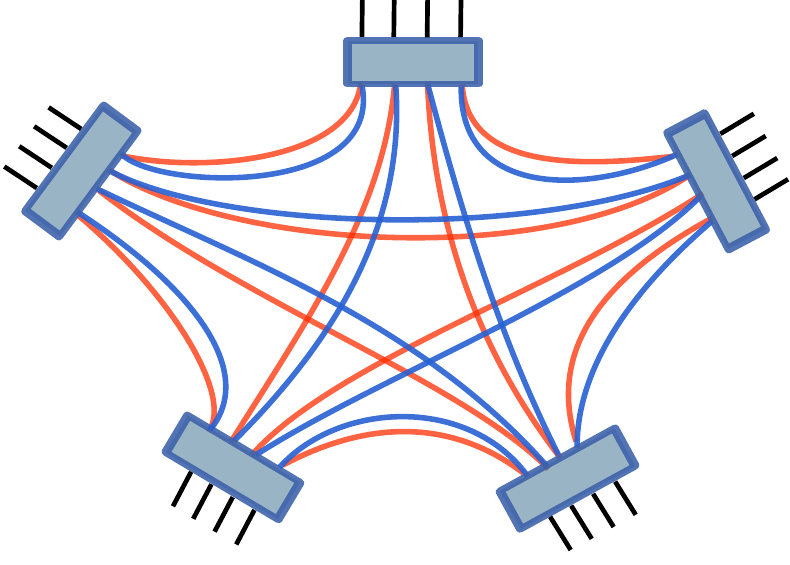}
	\caption{Graph for the 4-simplex amplitude in the DL model. The contractions inside correspond to two copies of BF 20j symbols, constrained on the boundary.}
		\label{fig:DLamplitude}
\end{figure}
This amplitude corresponds to two copies of 20j symbols from BF theory constrained by $[z_L^{s(\Delta)}|z_L^{t(\Delta)}\ket = \rho^2 [z_R^{s(\Delta)}|z_R^{t(\Delta)}\ket$ on the boundary.

%%%%%%%%%%%%%%%%%%%%%%%%%%%%%%%%%%%%%%%%%%%%%%%%%%
\subsubsection{Constrained projector}\label{sec:constrainedprojector}
%%%%%%%%%%%%%%%%%%%%%%%%%%%%%%%%%%%%%%%%%%%%%%%%%%
Since spin foam amplitudes for BF theory are constructed by gluing together projectors  (\ref{eqn_res_id_UN}) into graphs corresponding to 4d quantum geometries, we find it natural to instead impose the constraints on the arguments of the projectors themselves.  
Let us consider the Spin(4) projector obtained by taking a product of two SU(2) projectors
\be
  P(z_i;w_i) P(z'_i;w'_i) = \sum_{J} \frac{\left( \sum_{i<j} [z_i|z_j\ket[w_i|w_j\ket\right)^J}{J!(J+1)!} \sum_{J'} \frac{\left( \sum_{i<j} [z'_i|z'_j\ket[w'_i|w'_j\ket\right)^{J'}}{J'!(J'+1)!},
\ee
where we use a prime to distinguish the left and right 
SU(2) sectors.  We will now impose the holomorphic simplicity constraints on both incoming and outgoing strands in the Spin(4) projector
$$[z'_i|z'_j\ket = \rho^2 [z_i|z_j\ket \qquad [w'_i|w'_j\ket = \rho^2 [w_i|w_j\ket .$$ 
This will make the two products of spinors in the two projectors proportional to each other, with the proportionality constant being $\rho^4$. Note that the imposition of simplicity constraints on all of the spinors also forces the measure of integration on $\C^2$ to change to
\be
\rd\mu_\rho(z) := \frac{(1+\rho^2)^2}{\pi^{2}}e^{-(1+\rho^2)\bra z|z\ket} \rd^2 z .
\ee
 The factor of $(1+\rho^2)^2$ is added for normalization. It insures that 
\be
\int \rd\mu_\rho(z) =1.
\ee  
Moreover this choice of normalization  is confirmed by the study of asymptotics of both this and the DL model \cite{spinfoamfactory}.
It is exactly this choice that insures that both models have the same semi-classical limit.
We are now ready to define a new \emph{constrained propagator} $P_\rho$ by applying the simplicity constraints on the Spin(4) projector 
\be
  P_\rho (z_i;w_i) \equiv P(z_i;w_i) P(\rho z_i;\rho w_i) = \sum_{J} \sum_{J'} \frac{\rho^{4J'}}{J!(J+1)!J'!(J'+1)!} \left( \sum_{i<j} [z_i|z_j\ket[w_i|w_j\ket\right)^{J+J'} .
\ee
The two sums over integers $J$ and $J'$ are independent, so we can simplify this expression for the constrained propagator into a single sum by letting $J+J' \rightarrow J$. This allows us  to arrive at a more compact form of the constrained propagator, given by
\be
  P_\rho(z_i;w_i)  = \sum_{J} F_\rho(J) \frac{\left( \sum_{i<j} [z_i|z_j\ket[w_i|w_j\ket\right)^{J}}{J!(J+1)!} \label{eq:PH},
\ee
where we have recognized that the numerical factor in front of the spinors is actually the power series expansion of the hypergeometric function
\be
F_\rho(J):=  {}_2F_1(-J-1,-J;2;\rho^4) = \sum_{J'=0}^{J} \frac{J!(J+1)! \rho^{4J'}}{(J-J')!(J-J'+1)!J'!(J'+1)!}.
\ee

Notice that the constrained Spin(4) propagator is just an SU(2) projector with non-trivial weights (greater than 1) for each term that depend on the Barbero-Immirzi parameter. In general, this hypergeometric function is a complicated function of $\rho$, but let us look at two interesting limiting cases. For $\rho =0$, which corresponds to $\gamma \rightarrow 1$, we have
\be
  {}_2F_1(-J-1,-J;2;0)=1,
\ee
so we end up with pure SU(2) BF theory. This is obvious, as $\rho = 0$ forces all the left spinors to be 0. Another limit often considered is $\rho = 1$, which in this construction surprisingly corresponds to both of the limits $\gamma \rightarrow 0$ and $\gamma \rightarrow\infty$. In this limit we get also a relatively simple expression
\be
  {}_2F_1(-J-1,-J;2;1) = \frac{(2J+2)!}{(J+2)!(J+1)!}.
\ee
This limit does not have an obvious interpretation apart from its simplicity.

We can now define the partition function for the Spin Foam model made up from these constrained propagators. Since the propagators are just BF projectors with non-trivial weight, we can write the partition function on a 2-complex $\Delta^\ast$ as
\be
\mathcal{Z}^{\Delta^\ast}_{G}=\sum_{j_f} \prod_{f \in \Delta^\ast}\mathcal{A}_f(j_f) \int \left\{\prod_{all}\rd\mu_\rho(z)\rd\mu_\rho(w)\right\}\sum_{k^e_{ff'}\in K_j}\prod_e  P_\rho^{k^e_{ij}}(z^e_i;w^e_i),
\ee
where $\mathcal{A}_f(j_f)$ is a face weight, the set $K_j$ was defined previously in Eq. (\ref{kj}) to be the set of integers $k_{ij}$ satisfying $\sum_{i\neq j}k_{ij}=2j_i$ and contraction of spinors according to the 2-complex $\Delta^\ast$ on different edges is implied. The constrained propagator at fixed spins is given by
\be
P_\rho^{k^e_{ij}}(z^e_i;w^e_i) :=\frac{ F_\rho(J_e)}{(J_e+1)!}\prod_{i<j}\frac{([z^e_i|z^e_j\ket[w^e_i|w^e_j\ket)^{k^e_{ij}}}{k^e_{ij}!} .
\ee

Each constrained propagator comes with an orientation, with spinors $z$ incoming into the box and spinors $w$ outgoing in this paper's convention. A change of this edge orientation results in overall minus sign for the amplitude. Additionally we also put an orientation on each strand, which dictates how spinors on different propagators are contracted. An example is shown in Fig. \ref{fig:contraction}. When we glue 4-simplices, we have two propagators contracted on the dual edge along which they are glued. 

\begin{figure}[h]
	\centering
		\includegraphics[width=0.5\textwidth]{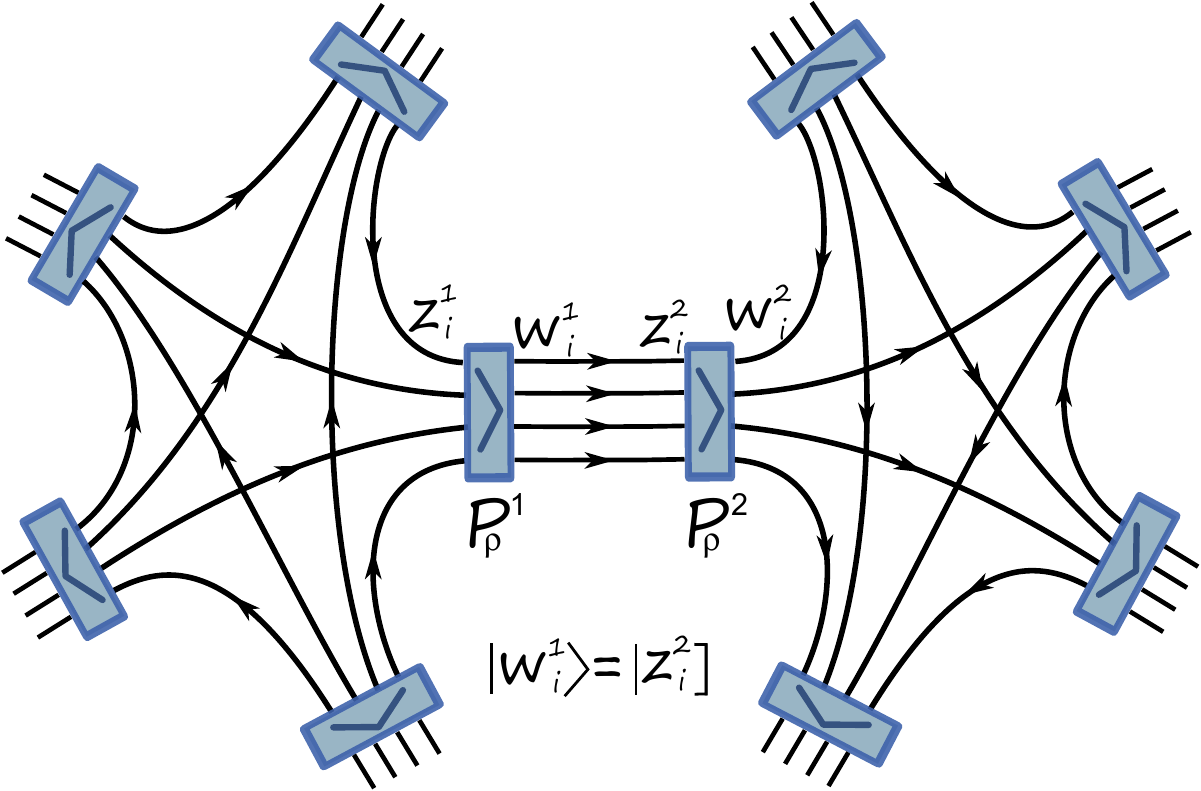}
	\caption{Graph for the amplitude of contraction of two 4-simplices. Propagators $P_\rho^1$ and $P_\rho^2$ belong to the same edge but two different 4-simplices. The spinors belonging on the same strand but belonging to different propagators are contracted according to the strand orientation. For example, spinors $w^1_i = \check{z}^2_i$.}
		\label{fig:contraction}
\end{figure}

It is interesting to note here that, unlike in the usual Spin Foam models, this definition in terms of propagators does not necessarily constrain the partition function to be a product of vertex amplitudes, thus allowing for more general non-geometrical structures.

%%%%%%%%%%%%%%%%%%%%%%%%%%%%%%%%%%%%%%%%%%%%%%%%%%%
\subsection{The Homogeneity Map}
%%%%%%%%%%%%%%%%%%%%%%%%%%%%%%%%%%%%%%%%%%%%%%%%%%%
In this section we will introduce a very useful tool that will allow us to make calculations of Pachner moves more tractable. Notice that the propagator $P_\rho$ is a polynomial obtained from  products of monomials $ [z_i|z_j\ket[w_i|w_j\ket^{k_{ij}}$ which possess a degree of homogeneity of $4 k_{ij}$. 
These products of monomials of different homogeneity degrees are orthogonal in the Bargmann-Fock space. The homogeneity property allows us to always separate out and track  terms of given homogeneity in a power series expansion. This means that we can perform transformations term by term in the series expansion of the propagator and independently integrate each term.

Let us hence define a more general propagator $G_\tau$ that can be exponentiated
\be
  G_\tau(z_i;w_i)  = \sum_{J} \tau^J \frac{\left( \sum_{i<j} [z_i|z_j\ket[w_i|w_j\ket\right)^{J}}{J!} = e^{\tau \sum_{i<j}[z_i|z_j\ket [w_i|w_j\ket  }
\ee
and denote it graphically by
\be 
G_\tau(z_i;w_i) = 
\begin{tikzpicture}[baseline=0,scale=0.45]
  \node at (-4,1.5) {$[z_1|$}; \draw (-3,1.5) -- (-1,1.5);
  \node at (-4,0.5) {$[z_2|$}; \draw (-3,0.5) -- (-1,0.5);
  \node at (-4,-0.5) {$[z_3|$}; \draw (-3,-0.5) -- (-1,-0.5);
  \node at (-4,-1.5) {$[z_4|$}; \draw (-3,-1.5) -- (-1,-1.5);
  \draw (-1,-2) --(-1,2) -- (1,2) -- (1,-2) -- (-1,-2); \node at (0,0) {$\tau$};
  \draw (1,1.5) -- (3,1.5); \node at (4,1.5) {$|w_1\ket$}; 
  \draw (1,0.5) -- (3,0.5); \node at (4,0.5) {$|w_2\ket$}; 
  \draw (1,-0.5) -- (3,-0.5); \node at (4,-0.5) {$|w_3\ket$}; 
  \draw (1,-1.5) -- (3,-1.5); \node at (4,-1.5) {$|w_4\ket$}; 
\end{tikzpicture}.
\ee
We can see that $\tau^J$ tracks homogeneity of the polynomial in spinors. If we transform each of these $\tau^J$ into a function of $J$, the integrals of the polynomials stay the same.
In this way we can perform complicated calculations with $G_\tau$ and in the end we can use the following map, defined by a functional $H_f$ mapping $G_\tau$ to the desired function $f$:
\begin{align}
H_\rho: G_\tau \rightarrow P_\rho &  \qquad \text{with} \qquad H_\rho:  \tau^J \rightarrow \frac{F_\rho(J)}{(J+1)!} \qquad \text{(Simplicity Constraints)} \label{eqn_SC_map} \\
H_P: G_\tau \rightarrow P & \qquad \text{with} \qquad	H_P: \tau^J \rightarrow \frac{1}{(J+1)!}  \qquad \text{(BF Theory)} \label{eqn_BF_map}
\end{align}
in order to recover the propagators of the BF theory  or the one of the gravity model with simplicity constraints imposed. Not that $P_0=P$ so the BF model is included in our more general description.  We are of course not limited to only these choices and could in principle study a much wider class of spin foam models built by non-trivial propagators.

By considering how BF projectors compose in Eq. (\ref{eq:PPisP}), it is quite easy to find the homogeneity map for composing the propagator $P_\rho$ $n$ times: $P_\rho \circ \cdots \circ P_\rho$. To do this, we just realize that if one  reintroduces back the factor of $1/(J+1)!$ into the definition of $G_\tau$, it then defines just the BF projector $P$ with the spinors $z$ rescaled to $\sqrt{\tau}z$. The homogeneity map for the composition is therefore given by\footnote{Note that the factor of $(1+\rho^2)^{2J}$ in (\ref{eqn_SC_map}) comes from the fact that the measure has changed under the simplicity constraints to $\rd\mu_\rho(z) = (1+\rho^2)^2\pi^{-2}e^{-(1+\rho^2)\bra z|z\ket}$.  Hence every contraction produces a factor $(1+\rho^2)^{-2j}$ where $j$ is the representation of the line.  There is one such contraction for each $j$ where $J = \sum j$ for each $\tau$.}
\be
  \tau^J \rightarrow \frac{F_\rho(J)^n}{(J+1)!(1+\rho^2)^{2(n-1)J}} \qquad \text{for} \qquad  G_\tau \rightarrow P_\rho^n \qquad \text{($n$ Propagators)} .
\ee

For the purpose of calculating Pachner moves, we will need to consider contracted loops of spinors. In BF theory, such a loop should correspond to an SU(2) delta function. Using the spinorial language however, we get
\be
\raisebox{-7.5mm}{\includegraphics[keepaspectratio = true, scale = 0.5] {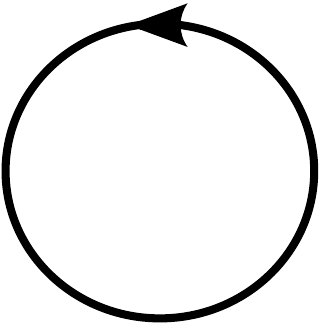}} \ \  = \ \ \int \rd \mu(z) e^{\bra z|z\ket} = \sum_j \int \rd \mu (z) \frac{\bra z|z\ket^{2j}}{(2j)!}  =\sum_j  \chi^j(\id) = \sum_j (2j+1),
\ee
whereas a delta function is $\delta_{SU(2)}(\id) = \sum_j (2j+1)^2$. One way of going around this is to change measure of integration for this loop to $\rd\tilde{\mu}(z) = (\bra z|z\ket -1)\rd\mu (z)$, as was suggested in \cite{Dupuis:2011dh}. This provides the additional factor of $(2j+1)$. An alternative way is to follow in the spirit of the homogeneity map and introduce a $\tau '$ that tracks the homogeneity in this loop. For clarity, we add a symbol for this face weight into the graph:
\be
\raisebox{-8mm}{\includegraphics[keepaspectratio = true, scale = 0.5] {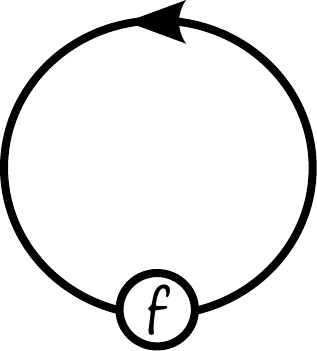}} \ \  = \ \ \int \rd \mu(z) e^{\tau '\bra z|z\ket} = \sum_j \tau '^{2j} (2j+1).
\ee
The replacement of $\tau '^{2j}\rightarrow (2j+1)$ now defines a homogeneity map for a BF loop. Of course, we now do not have to restrict ourselves to this simple face weight and can choose an arbitrary function of spin.

The homogeneity map we have developed in this section will be very useful in computing the 4-dimensional Pachner moves. In later sections we will define additional homogeneity maps as we go on, to simplify the calculations.

%%%%%%%%%%%%%%%%%%%%%%%%%%%%%%%%%%%%%%%%%%%%%%%%%%%
\section{Pachner moves in 3d quantum gravity}\label{sec:3dPachner}
%%%%%%%%%%%%%%%%%%%%%%%%%%%%%%%%%%%%%%%%%%%%%%%%%%%
In this section we review the notion of Pachner moves and their calculation in 3d SU(2) BF theory, to set up the stage for comparison to the 4-dimensional models. A crucial tool that allows the calculation is a so called \emph{loop identity}, which we will derive in the holomorphic representation.

%%%%%%%%%%%%%%%%%%%%%%%%%%%%%%%%%%%%%%%%%%%%%%%%%%%
\subsection{Definition of Pachner Moves}
%%%%%%%%%%%%%%%%%%%%%%%%%%%%%%%%%%%%%%%%%%%%%%%%%%%

To show that a theory defined on a triangulated manifold is topologically invariant, we need a way to relate different triangulations. This is provided by the Pachner moves, which are local replacements of a set of connected simplices by another set of connected simplices.

\begin{theorem}
Any simplicial piecewise linear manifold $\mathcal{M}$ can be transformed into any other simplicial piecewise linear manifold $\mathcal{M}'$ homeomorphic to $\mathcal{M}$ by a finite sequence of Pachner moves.

For proof, see \cite{Pachner}.
\end{theorem}

Pachner moves are constructed by adding (or removing) vertices, edges, triangles etc. to (from) the existing triangulation. They can be also obtained  in d dimensions by gluing a (d+1)-simplex onto the d-dimensional triangulation. There are several Pachner moves in each dimension and they correspond to changing a configuration of $n$ basic building blocks (d-simplices) into a configuration of $m$ building blocks - we call them $n$--$m$ Pachner moves. In two dimensions we hence have the moves 2--2, 1--3 moves and their reverse.
\begin{figure}[h]
	\centering
		\includegraphics[width=0.7\textwidth]{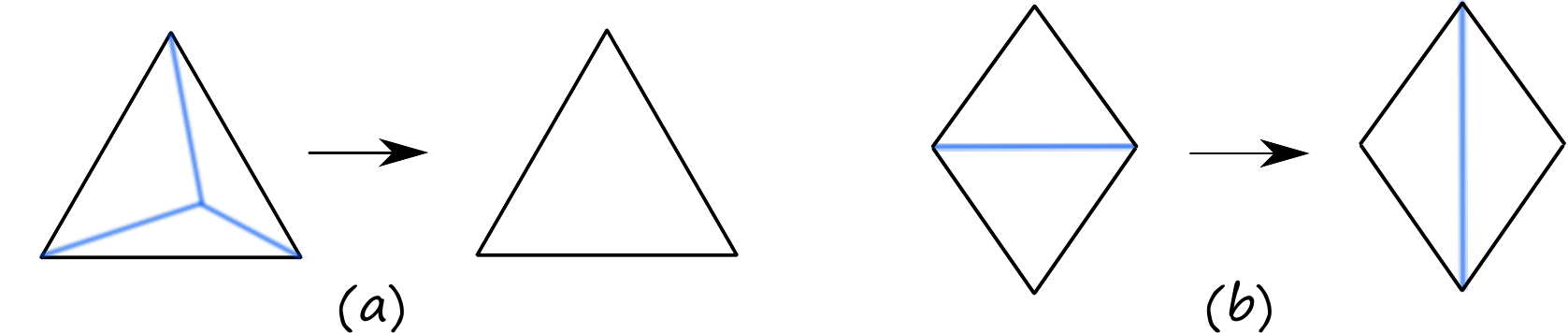}
	\caption{Two dimensional Pachner moves: a) 3--1  move, in which three triangles are merged into one by removing a vertex inside; b) 2--2 move, in which two triangles exchange the edge, along which they are glued.}
		\label{fig:2dmoves}
\end{figure}
 The 2--2 move corresponds to changing the edge along which two triangles are glued, while the 1--3 move corresponds to adding a vertex inside a triangle and connecting it to the other vertices by three edges, ariving in a configuration with three triangles. Fig. \ref{fig:2dmoves} shows the inverse.  In three dimensions we have 3--2, 4--1 moves and their reverse, see Fig. \ref{fig:3dmoves}.
\begin{figure}[h]
	\centering
		\includegraphics[width=1 \textwidth]{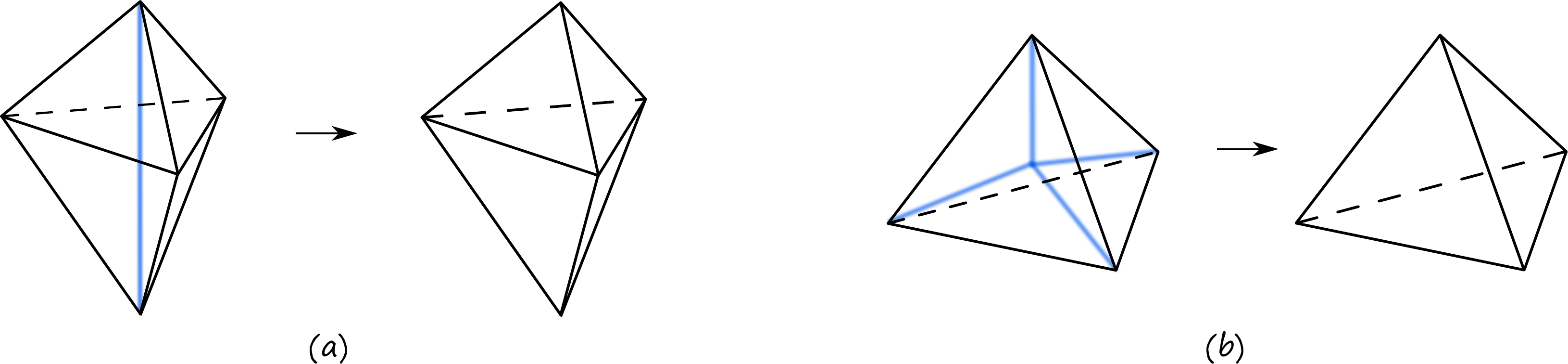}
	\caption{Three dimensional Pachner moves: a) 3--2  move, in which three tetrahedra are changed into two tetrahedra by removing a common edge; b) 4--1 move, in which four tetrahedra are combined into one by removing a common vertex.}
		\label{fig:3dmoves}
\end{figure}
 The 3--2 move corresponds to removing an edge, along which three tetrahedra were glued and changing it into a configuration of two tetrahedra. The 4--1 move is combining four tetrahedra into one tetrahedron through removing a common vertex.

%%%%%%%%%%%%%%%%%%%%%%%%%%%%%%%%%%%%%%%%%%%%%%%%%%%
\subsection{Fixing the gauge}\label{sec:gf}
%%%%%%%%%%%%%%%%%%%%%%%%%%%%%%%%%%%%%%%%%%%%%%%%%%%

3-dimensional gravity, as described by BF theory, has two important gauge symmetries -- internal rotational ``Lorentz'' $SU(2)$ gauge symmetry and the \emph{translational} symmetry \cite{Freidel:2002dw} following  from the Bianchi identity $\rd_\omega F = 0$. The action can be explicitly written as 
\be
S\left[e,\omega\right] = \int_\mathcal{M} \tr \left[e\wedge \left(\rd\omega + \omega\wedge\omega\right)\right].
\ee
The local Lorentz gauge transformations $\delta^L$  and the translational symmetry transformations $\delta^T$ are given by
\be
\begin{split}
\delta_X^L\omega &= \rd_\omega X, \ \ \ \ \delta_X^L e = \left[e,X\right] \\
\delta_X^T\omega &= 0, \ \ \ \ \ \ \ \  \delta_X^T e = \rd_\omega X,
\end{split}
\ee
where the parameter of transformations is an $\mathfrak{su}(2)$ Lie algebra element $X$. There is also obviously the diffeomorphism symmetry generated by a vector field $\xi_\mu$, but one can show \cite{Freidel:2002dw} that on-shell we have
\be
\delta^D_\xi = \delta^L_{\iota_\xi \omega}+\delta^T_{\iota_\xi e},
\ee
i.e. the three symmetries are related.

Let us now understand how to fix the ``Lorentz'' gauge on a spin network. Since the volume of the group SU(2) is finite, the gauge fixing  amounts to only a change of variables along a maximal tree. We  follow \cite{Freidel:2002xb} in defining the gauge fixing procedure.
Consider a graph $\Gamma$ with $E$ edges and $V$ vertices. Each edge is oriented so that it starts at a source vertex $s(e)$ and ends at target $t(e)$. Consider now a spin network function such that
\be
\psi^\Gamma(g_{e_1},\ldots , g_{e_E})=\psi^\Gamma(h^{-1}_{s(e_1)}g_{e_1}h_{t(e_1)},\ldots , h^{-1}_{s(e_E)}g_{e_E}h_{t(e_E)}) .\label{eq:gaugefixing}
\ee
Now choose a maximal tree $T$ in $\Gamma$, i.e. a collection of $V-1$ edges which passes through every vertex, without forming loops. Choose a vertex $A$ to be the root\footnote{One can show that the gauge fixing procedure is independent of this choice.} of the tree T and label $g_{vA}^T$ the product of group elements $g_{e_i}$ along T that connect vertex $v$ and $A$. Next we will use Eq. (\ref{eq:gaugefixing}) with $h_v = g_{vA}^T$, so that $\psi^\Gamma=\psi^\Gamma(G_1^T,\ldots ,G_E^T)$ with $G_e^T = g_{As(e)}^T g_e g_{t(e)A}^T$. 

Now, for any edge $e \in T$, there is a unique path along the tree connecting $A$ and $s(e)$ or $t(e)$. Let us choose this to be $t(e)$, since the other case works in the same way. It follows that $ g_{s(e) A}^T = g_e g_{t(e)A}^T$ and so $G_e^T = \id$ for $e \in T$. Hence the procedure for gauge fixing is to set all the group elements on the maximum tree to $\id$ and change all the other to $g_{e_i} = G_{e_i}^T$. Since $\int_{SU(2)} \rd g = 1$, ending up with empty integrations does not lead to any divergences. In the language of amplitudes written in terms of projectors, this corresponds to replacing the projectors $ P(z_i;w_i) = \int_{\text{SU}(2)} \rd g \,e^{\sum_i [z_i|g|w_i\ket}$ on the maximal tree by the trivial propagators $  \one(z_i;w_i) \equiv e^{\sum_i [z_i|w_i\ket}$. This procedure carries over to the 4-dimensional case trivially, since Spin(4) is just a product $SU(2)\times SU(2)$.

We will postpone the discussion of the translational symmetry to until after we have calculated the 4-1 Pachner move, as we will see that it is directly related to the divergence coming from that calculation. However, in 4-dimensional Spin Foam models the relation between divergences and translational symmetry is unknown -- we will discuss this in Section \ref{sec:coarse}.

%%%%%%%%%%%%%%%%%%%%%%%%%%%%%%%%%%%%%%%%%%%%%%%%%%%
\subsection{The Loop Identity}
%%%%%%%%%%%%%%%%%%%%%%%%%%%%%%%%%%%%%%%%%%%%%%%%%%%

The BF theory partition function is independent of the triangulation $\Delta$.  This can be shown by demonstrating its invariance (up to an overall factor) with respect to a finite set of coarse graining moves, constructed out of Pachner moves.  The Pachner moves can all be derived from one identity which we will call the loop identity. This identity follows from the coherent state representation of the SU(2) delta function
\be
  \delta(g) = \int \rd \tilde{\mu}(z) e^{\bra z | g | z \ket},
\ee
where $\rd \tilde{\mu}(z) = \rd \mu(z)(\bra z | z\ket - 1)$.  Therefore
\bea \label{eqn_loop_id}
  \int \rd \tilde{\mu}(z_n) P(z_1,...,z_n;w_1,...,\check{z}_n) 
  &=& \int \rd g e^{\sum_{i=1}^{n-1} [z_i|g|w_i\ket} \int \rd \tilde{\mu}(z_n) e^{[ z_n | g | z_n ]} \nn \\
  &=& \int \rd g e^{\sum_{i=1}^{n-1} [z_i|g|w_i\ket} \delta( g ) \nn \\ 
  &=& e^{\sum_{i=1}^{n-1} [z_i|w_i\ket} \nn \\
  &=& \one(z_1,...,z_{n-1};w_1,...,w_{n-1}),
\eea
which is represented graphically by
\be\label{eqn_loop_identity}
\raisebox{-17mm}{\includegraphics[keepaspectratio = true, scale = 1] {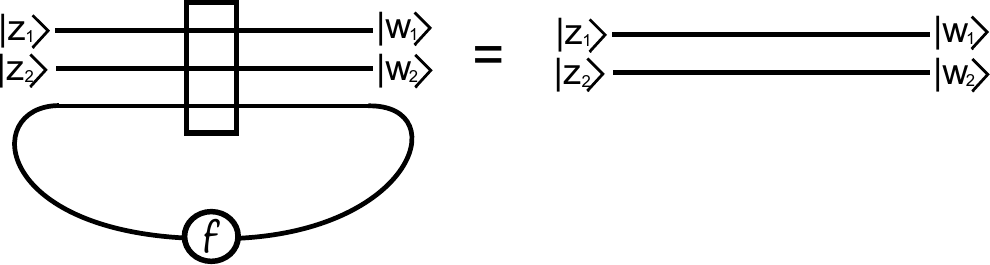}}
\ee

%
%\be \label{eqn_loop_identity}
%\begin{tikzpicture}[baseline=0,scale=0.45]
%  \node at (-4,1.5) {$[z_1|$}; \draw (-3,1.5) -- (-1,1.5);
%  \node at (-4,0.5) {$[z_2|$}; \draw (-3,0.5) -- (-1,0.5);
%  \node at (-4,-0.5) {$[z_3|$}; \draw (-3,-0.5) -- (-1,-0.5);
%  \draw (-3,-1.5) -- (-1,-1.5);
%  \draw (-1,-2) --(-1,2) -- (1,2) -- (1,-2) -- (-1,-2);
%  \draw (1,1.5) -- (3,1.5); \node at (4,1.5) {$|w_1\ket$}; 
%  \draw (1,0.5) -- (3,0.5); \node at (4,0.5) {$|w_2\ket$}; 
%  \draw (1,-0.5) -- (3,-0.5); \node at (4,-0.5) {$|w_3\ket$}; 
%  \draw (1,-1.5) -- (3,-1.5); 
%  \draw (-3,-1.5) to [out=180,in=180] (-3,-3.5);
%  \draw (3,-1.5) to [out=0,in=0] (3,-3.5);
%  \draw (-3,-3.5) -- (3,-3.5);
%\end{tikzpicture}
%=
%  \begin{tikzpicture}[baseline=0,scale=0.45]
%  \node at (-4,1.5) {$[z_1|$}; \draw (-3,1.5) -- (3,1.5);
%  \node at (-4,0.5) {$[z_2|$}; \draw (-3,0.5) -- (3,0.5);
%  \node at (-4,-0.5) {$[z_3|$}; \draw (-3,-0.5) -- (3,-0.5);
%  \node at (4,1.5) {$|w_1\ket$}; 
%  \node at (4,0.5) {$|w_2\ket$}; 
%  \node at (4,-0.5) {$|w_3\ket$}; 
%\end{tikzpicture}
%\ee

Since each closed loop of the BF partition function (\ref{eqn_BF_intertwiners}) has a factor of $2j_f+1$ we will use the convention that two lines are contracted with $\rd \mu(z)$ as in (\ref{eqn_contraction}), however, the contraction of a line with itself, i.e. a loop, is contracted with the measure $\rd \tilde{\mu}(z)$ as in (\ref{eqn_loop_identity}). An alternative way would be to use the homogeneity map to keep track of this face weight.

%%%%%%%%%%%%%%%%%%%%%%%%%%%%%%%%%%%%%%%%%%%%%%%%%%%
\subsection{Alternative method}\label{sec:BFloopagain}
%%%%%%%%%%%%%%%%%%%%%%%%%%%%%%%%%%%%%%%%%%%%%%%%%%%
The expression for the loop identity we have just derived, while compact, does not generalize straightforwardly to the case of 4-dimensional QG models with simplicity constraints (due to the presence of the group integrals). We will thus redo the above calculation with the projector written in terms of only spinors without group integration.

We expect that the loop identity (\ref{eqn_loop_id}) applied to the projector (\ref{eqn_res_id_UN}) implies that
\be \label{eqn_loop_integrated}
 \sum_{j_n}(2j_n+1) \int \rd \mu(z_n)  \frac{\left(\sum_{i<j}[z_i|z_j\ket[w_i|w_j\ket\right)^{J}}{ J!(J+1)!} =  \prod_{i=1}^{n-1} \frac{[z_i|w_i\ket^{2j_i}}{(2j_i)!} ,
\ee
where the integration is performed with $w_{n} = \check{z}_{n}$. Below we will directly show this. Let us perform the integration on the LHS explicitly by using the homogeneity map to keep track of the $1/(J+1)!$ and the face weight $(2j_n+1)$ and then summing over $j_i$. Namely, let us use the homogeneity maps $\tau^J\rightarrow 1/(J+1)!$ and $\tau '^{2j_n}\rightarrow (2j_n+1)$. The result is then
\be \label{eqn_gen_func}
  \int \rd \mu(z_n)  \exp\left({\tau \sum_{i<j<n}[z_i|z_j\ket [w_i|w_j\ket - \tau \tau ' \sum_{i<n} \bra z_n|z_i\ket [w_i|z_n\ket  }\right)
= \frac{e^{\tau \sum_{i<j<n}[z_i|z_j\ket [w_i|w_j\ket}}{\det \left(\one - \tau\tau ' \sum_{i<n}|w_i\ket [z_i|\right)}.
\ee

To continue, we have to be able to evaluate the determinant in the denominator. This is thankfully not too difficult, as the matrix in question is just a $2\times 2$ matrix made up by spinors. Indeed, the following lemma comes in handy

\begin{lemma} 
\label{eqn_det_lemma}
Let $M = \one - \sum_i C_i |A_i\ket[B_i|$ then
$$
  \det M = 1 - \sum_i C_i [B_i|A_i\ket + \sum_{i<j} C_i C_j [A_i|A_j\ket[B_i|B_j\ket .
$$
The proof is given in Appendix \ref{app_det_lemma}.
\end{lemma}

Using this result, we can immediately find the determinant in (\ref{eqn_gen_func}). In our case, all $C_i = \tau \tau '$, hence we get that the loop identity for the homogenized projector $P_\tau$ becomes
\be
  \frac{e^{\sum_{1\leq i<j < n} \tau \tau ' [z_i|z_j\ket[w_i|w_j\ket }}{1- \sum_{i\neq n} \tau \tau ' [z_i|w_i\ket + \sum_{1\leq i<j < n} \tau^2\tau '^2 [z_{i}|z_{j}\ket [w_i|w_j\ket}.
\ee
Now we can expand both the numerator and the denominator in a power series and then use the homogeneity map to restore the $1/(J+1)!$ terms and the face weight. This allows us to get the loop identity for the projector (\ref{eqn_res_id_UN})
\be
\raisebox{-17mm}{\includegraphics[keepaspectratio = true, scale = 1] {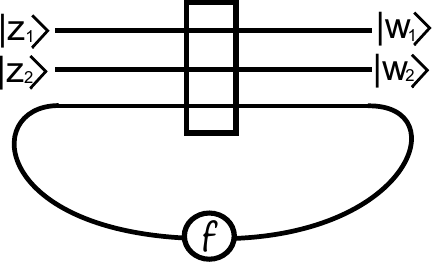}} \ \  = \ \ \sum_{J,J'}C_{BF}(J, J') \left( \sum_{i<n} [z_i|w_i\ket\right)^J\left(\sum_{i<j<n}[z_i|z_j\ket [w_i|w_j\ket\right)^{J'},
\ee
where we have defined the coefficient $C_{BF}(J, J')$ to be given by
\be
C_{BF}(J, J') =\sum_K (-1)^{J'-K}\frac{(J+J'-K)!(J+2J'-2K+1)}{J!(J'-K)!K!(J+2J'-K+1)!}.
\ee
At first glance, this is a worrisome result, as we do not only get the trivial projection (raised to power $J$), but also an unwanted \emph{mixing} term (raised to power $J'$). Notice though, that we have an additional free sum over the variable $K$ in the definition of the coefficient. We can actually explicitly evaluate this sum over $K$ to find the expected result
\be
C_{BF}(J,J') = \frac{\delta_{J', 0}}{J!}.
\ee
Hence only the $J'=0$ term is non-vanishing, so the mixing terms always drop out in BF theory. We thus recover the result (\ref{eqn_loop_integrated}) that we set out to prove. This calculation readily is generalized in the following section to the case with simplicity constraints. The major difference in this case is  the lack of the cancellation of the mixing terms.

%%%%%%%%%%%%%%%%%%%%%%%%%%%%%%%%%%%%%%%%%%%%%%%%%%%
\subsection{Invariance under Pachner moves and symmetry}
%%%%%%%%%%%%%%%%%%%%%%%%%%%%%%%%%%%%%%%%%%%%%%%%%%%

We will now proceed to show the invariance of the 3-dimensional SU(2) BF theory under 3--2 and 4--1 Pachner moves using the language of spinors. In the case of 4-1 move we  find a divergence directly related to the translational symmetry.

\subsubsection{3--2 move}
As can be seen in the Fig. \ref{fig:3dmoves} \! a), the configuration of three tetrahedra in the 3--2 move is glued along one edge. This corresponds to a loop of a single strand in the cable diagram, see Fig.\ref{fig:3-2move}.

\begin{figure}[h]
	\centering
		\includegraphics[width=0.9\textwidth]{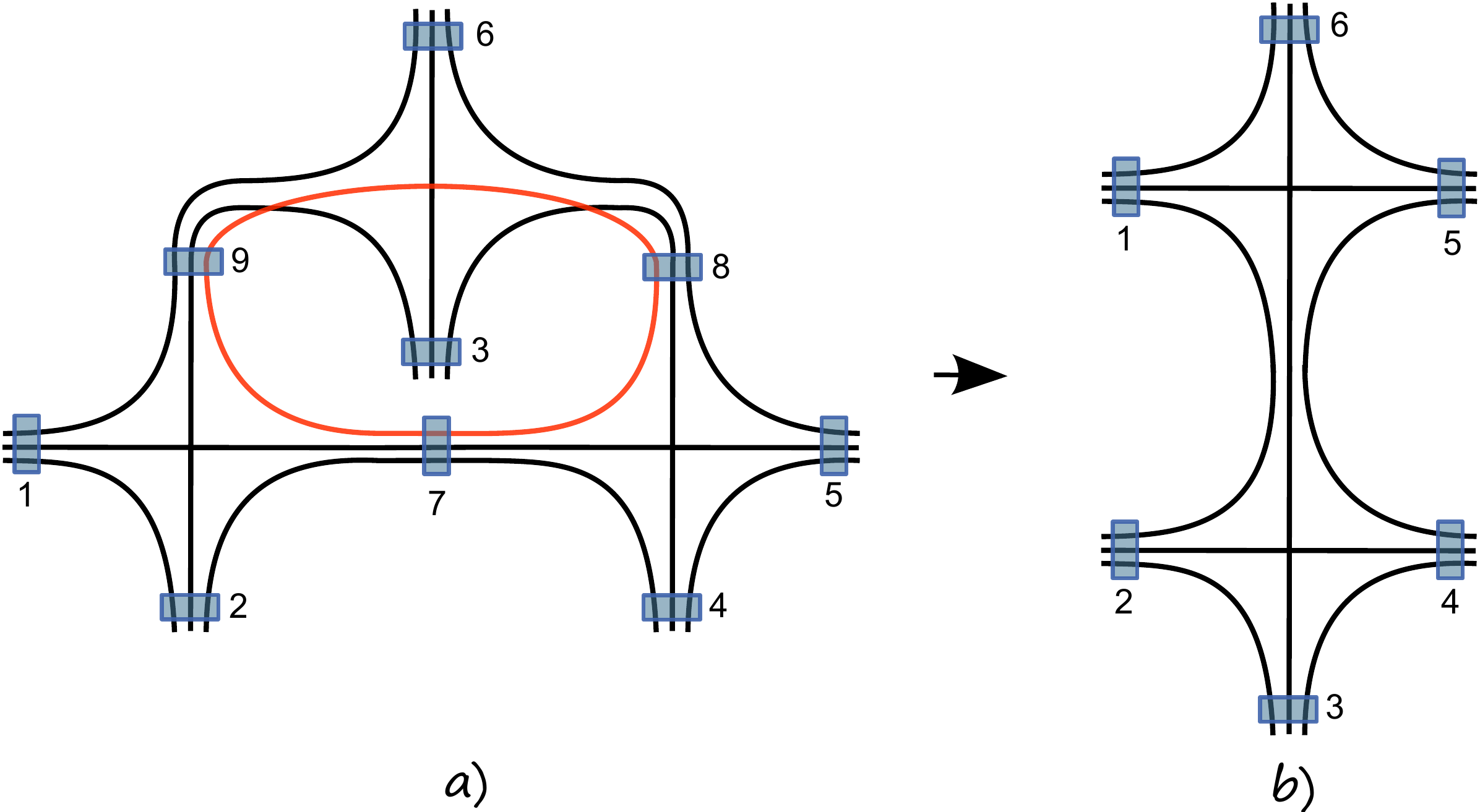}
	\caption{a) Cable diagram for the 3-2 move. The internal loop is colored. b) After gauge-fixing projectors 7 and 9 and performing loop identity on projector 8, the diagram reduces to gluing of two tetrahedral graphs.}
		\label{fig:3-2move}
\end{figure}

 By choosing a maximum tree (with a root at the projector 1) in the diagram, we can gauge fix the projectors number 7 and 9. This allows us to apply the loop identity (\ref{eqn_loop_id}) to integrate out the strand number 10 by performing the group integral in projector number 8. We can identify now that the resulting cable diagram is exactly that of the two tetrahedra glued together, see Fig. \ref{fig:3dmoves} \! b). Hence it is immediate that the SU(2) BF theory is invariant under the 3--2 Pachner move, as the two configurations are gauge equivalent.

\subsubsection{4--1 move}
The configuration of four tetrahedra in the 4--1 move shares in total four edges, which corresponds to four loops in a cable diagram, see Fig. \ref{fig:41move} a).

\begin{figure}[h]
	\centering
		\includegraphics[width=0.9\textwidth]{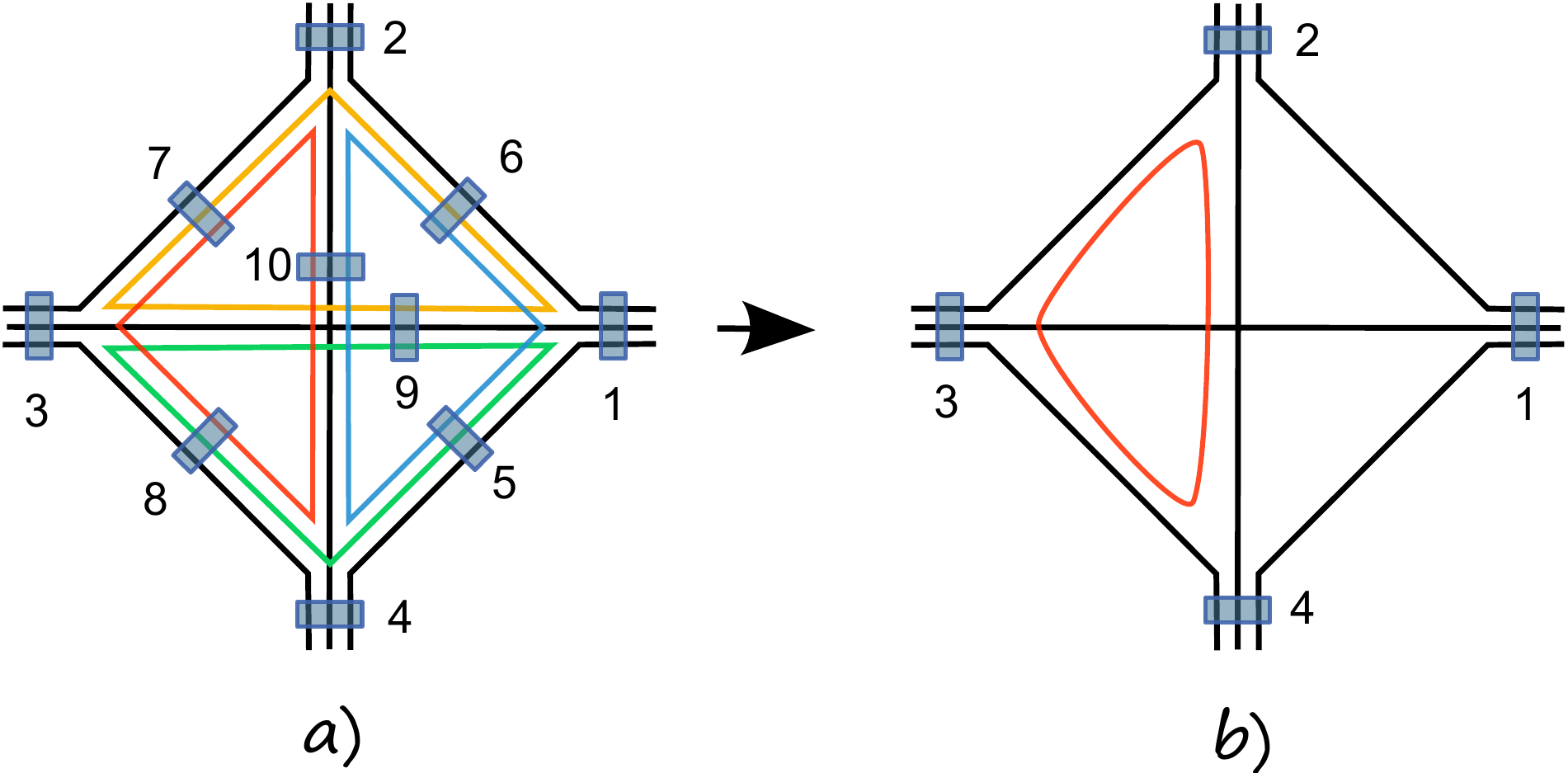}
	\caption{a) Cable diagram for the 4--1 move. The 4 different loops are colored. b) After applying three loop identities we are left with a tetrahedral cable graph with an insertion of one loop.}
		\label{fig:41move}
\end{figure}

We choose a maximum tree with a root at vertex 1, which allows us to gauge fix the projectors number 5, 6 and 9. We can now  apply the loop identity (\ref{eqn_loop_id}) to the projector 10 to remove the blue loop. Similarly we can apply the loop identities to projectors 7 and 8 to remove the yellow and green loops respectively. This leaves us with the last loop and no projectors left inside the graph, as in Fig. \ref{fig:41move} b). This final loop corresponds to the following integral
\be
\begin{split}
\int \rd \tilde{\mu}(z) e^{\bra z|z\ket} &= \sum_j \int \rd \tilde{\mu}(z) \frac{\bra z|z\ket^{2j}}{(2j)!} = \sum_j (2j+1) \int \rd {\mu}(z) \frac{\bra z|z\ket^{2j}}{(2j)!}\\ &=\sum_j (2j+1) \chi^j(\id)=\delta_{SU(2)}(\id).
\end{split}
\ee
Hence, we have shown that the BF partition function is invariant under the 4--1 move up to an overall divergent factor. 
The divergence we obtain  in SU(2) BF theory is exactly a SU(2) delta function $\delta_{SU(2)}(\id) = \sum_j (2j+1)^2$.  In \cite{Freidel:2002dw} it was shown that this is the same as the volume of the $\mathfrak{su}(2)$ Lie algebra. If we put on a cut-off $\Lambda$ on spins, then the divergence scales as $\sum_j (2j+1)^2 \sim \Lambda^3$. Since in 3d spin is proportional to length, we get a divergence that correponds to the translation symmetry of placing the extra vertex inside the tetrahedron. A correct Fadeev-Popov procedure \cite{Freidel:2002dw} divides the amplitude by exactly this divergence, so the Ponzano-Regge model is invariant after gauge fixing under both the 3--2 and 4--1 Pachner moves. This gauge fixing procedure was subsequently refined in \cite{FreidelL2, BarrettN, Matteo1, Matteo2} to lead to a complete  definition of 3 dimensional manifold invariant.  

%%%%%%%%%%%%%%%%%%%%%%%%%%%%%%%%%%%%%%%%%%%%%%%%%%%
\section{Pachner moves in 4d Riemannian holomorphic Spin Foam Model}\label{sec:4dPachner}
%%%%%%%%%%%%%%%%%%%%%%%%%%%%%%%%%%%%%%%%%%%%%%%%%%%
In this section we  obtain the main results of this paper. To calculate the 4d Pachner moves, we follow the strategy used in the BF case and calculate a constrained version of the loop identity, which together with gauge fixing  makes this involved calculation manageable. We find, unsurprisingly, that the models in question are not invariant under Pachner moves. The difference from the BF case is the presence of mixing of strands that exchanges the trivial propagators and delta functions with more complicated operators. We discuss the possible meaning of this mixing of srands as an insertion of an operator.

%%%%%%%%%%%%%%%%%%%%%%%%%%%%%%%%%%%%%%%%%%%%%%%%%%%
\subsection{Toy Loop}
%%%%%%%%%%%%%%%%%%%%%%%%%%%%%%%%%%%%%%%%%%%%%%%%%%%

To capture the essence of the computation without too much complexity, let us start with repeating the calculation of the BF loop identity, but with the constrained propagator $P_\rho(z_i;w_i)$ (\ref{eq:PH}) rather than the SU(2) projector. We will follow the treatment of the loop identity from Section \ref{sec:BFloopagain}.  We will thus find the loop identity for the generating functional 
\be
  G_\tau(z_i;w_i)  = e^{\tau \sum_{i<j}[z_i|z_j\ket [w_i|w_j\ket  } \nonumber
\ee
and at the end of the calculation use the homogeneity map to get the loop identity for $P_\rho(z_i;w_i)$ by changing $\tau^J\rightarrow  F_\rho(J)/(J+1)!$. We also want to be able to insert a face weight, which is a function of the spin we will sum over. This face weight could be \emph{a priori} arbitrary, but for the sake of definiteness, let us choose it to be $(2j+1)^\eta$ with $\eta\in\mathbb{R}$ being a free parameter, which keeps track of divergence properties of the Spin Foam model. The method we use allows us of course to modify the face weight to an arbitrary function of spin. To insert the face weight, we follow the calculation in BF theory and rescale the spinor in the loop by an additional factor of $\tau '$, which will keep track of homogeneity of that specific spinor. At the end of the calculation we can restore the face weight by replacing $ \tau'^{2j} \rightarrow (2j+1)^\eta $  in the series expansion. Let us now calculate the constrained loop identity:
\be \label{eq:constrainedloopidentity}
 \int \rd \mu_\rho (z_4)  e^{\tau \sum_{i<j<4}[z_i|z_j\ket [w_i|w_j\ket - \tau \tau ' \sum_{i<4} \bra z_4|z_i\ket [w_i|z_4\ket  }
= \frac{e^{\tau \sum_{i<j<4}[z_i|z_j\ket [w_i|w_j\ket}}{\det \left(\one - \frac{\tau\tau '}{1+\rho^2} \sum_{i<4}|w_i\ket [z_i|\right)}.
\ee
Unsurprisingly, we get nearly the same result as in Section \ref{sec:BFloopagain}, the difference being the additional factor of $1/(1+\rho^2)$, which arises from the modified integration measure $\rd\mu_\rho(z)$. Of course, the $\tau$ also carries a hypergeometric function of $\rho$. We can again use the lemma \ref{eqn_det_lemma} to evaluate the determinant. We arrive thus at the result
\bea
\int  \rd \mu_\rho (z_4) G_\tau (z_1,\ldots , \tau ' z_4; w_1,\ldots , \check{z_4}) = \frac{e^{\tau \sum_{i<j<4}[z_i|z_j\ket [w_i|w_j\ket}}{1 - \frac{\tau\tau '}{1+\rho^2} \sum_{i<4} [z_i|w_i\ket  + \sum_{i<j<4} \frac{\tau^2\tau '^2}{(1+\rho^2)^2} [z_i|z_j\ket [w_i|w_j\ket} \nonumber\\
= e^{\tau \sum_{i<j<4}[z_i|z_j\ket [w_i|w_j\ket}\! \sum_{N,M}\!\! \frac{(N+M)!}{N!M!}\left(\frac{\tau\tau '}{1+\rho^2} \right)^{N+2M}\!\! \left( \sum_{i<4} [z_i|w_i\ket\right)^{\!\!N} \!\!\left(-\sum_{i<j<4}[z_i|z_j\ket [w_i|w_j\ket\right)^{\!\!M}\!\!.
\label{cli}
\eea
We can now expand the exponential, combine the mixing terms and use the homogeneity map to reintroduce the face weight and the hypergeometric function of $\rho$. We hence find that the constrained loop identity for $P_\rho(z_i;w_i)$  is given by
\be
\raisebox{-17mm}{\includegraphics[keepaspectratio = true, scale = 1] {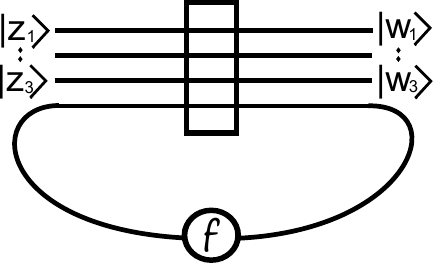}} \ \  = \ \  \sum_{J,J'}C_\rho(J,J')\left( \sum_{i<n} [z_i|w_i\ket\right)^J\left(\sum_{i<j<n}[z_i|z_j\ket [w_i|w_j\ket\right)^{J'},
\ee
with the coefficient $C_\rho(J,J')$ given by
\be
C_\rho(J,J')= \sum_K (-1)^{J'-K}\frac{(J\!+\!J'\!-\!K)!(J\!+\!2J'\!-\!2K\!+\!1)^\eta}{J!(J'\!-\!K)!K!(J\!+\!2J'\!-\!K\!+\!1)!}\frac{F_\rho(J+2J'-2K)}{(1+\rho^2)^{J+2J'-2K}},
\ee
where  $F_\rho(J) =  {}_2F_1(-J-1,-J;2;\rho^4)$. We have hence arrived at an expression very similar to the one in BF theory -- we again got the trivial propagation terms  $\sum_{i<4}  [z_i|w_i\ket$ together with additional mixing terms like $\sum_{i<j<4}[z_i|z_j\ket [w_i|w_j\ket$. Unlike in the BF loop identity however, there is no miraculous cancellation of the $J'\neq 0$ terms, unless we choose $\rho = 0$ and $\eta = 1$, i.e. we reduce this to SU(2) BF theory. Hence the way in which simplicity constraints break the topological symmetry is by introducing additional mixing terms in the loop identity. We can represent this graphically as

\be
\raisebox{-17mm}{\includegraphics[keepaspectratio = true, scale = 1] {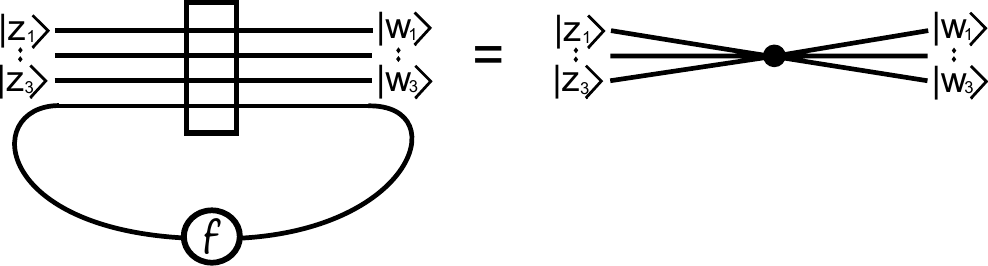}}
\ee

%%%%%%%%%%%%%%%%%%%%%%%%%%%%%%%%%%%%%%%%%%%%%%%%%%%
\subsection{The Constrained Loop Identity}\label{sec:constrainedloolidentity}
%%%%%%%%%%%%%%%%%%%%%%%%%%%%%%%%%%%%%%%%%%%%%%%%%%%
We are now going to see that  the loop identity we need for Pachner moves is somewhat different with the one we considered in the previous section. When we glue together 4-simplices, we need to glue them along their boundaries, necessitating the gluing of two propagators, i.e. we should work with $P_\rho\circ P_\rho$, rather than a single $P_\rho$. 
The reason for this being that in our model the propagator $P_\rho$ is inserted around each vertex and we get the composition of them along an edge. Since $P_\rho$ is not a projector unless $\rho=0$ we have $P_\rho\circ P_\rho\neq P_\rho$ .
Additionally, the loops arising in all the Pachner moves always are composed of three groups of propagatos $P_\rho\circ P_\rho$, rather than the single one we have considered. Fortunately, two of these can be always gauge fixed by a proper choice of a maximal tree, so that we have to consider the loop identity shown in Fig.\ref{fig:loopidentity}. 
\begin{figure}[h]
	\centering
		\includegraphics[width=0.65\textwidth]{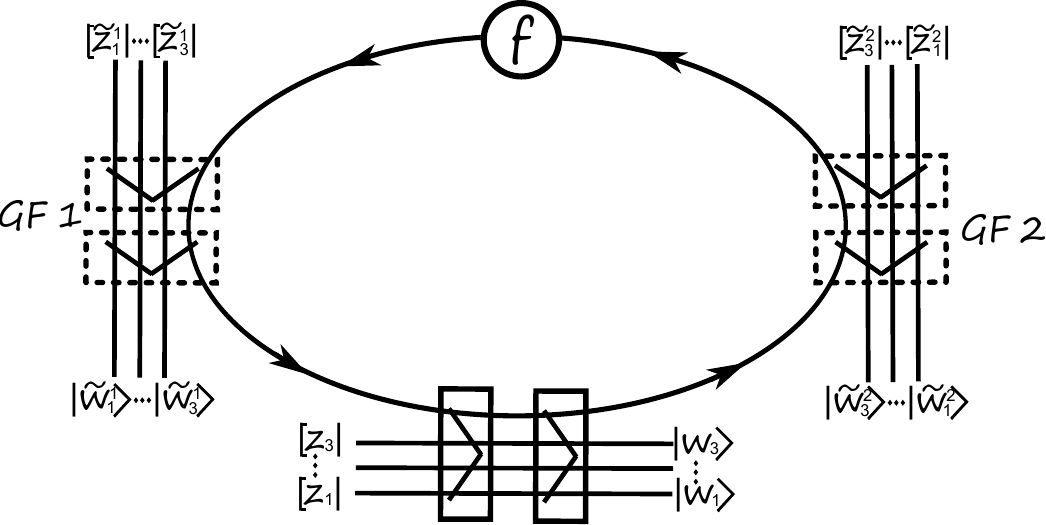}
	\caption{Loop identity with for the constrained projector with two extra gauge fixed projectors}
		\label{fig:loopidentity}
\end{figure}
In BF theory the gauge-fixing reduces the projectors to trivial propagators $\one (z_i;w_i)$, so we did not have to worry about this issue.

We thus have to first find the equivalent of the trivial propagator in the constrained case, i.e. the analog of setting $g = \one$ in (\ref{proj}) to get (\ref{trivial_proj}) but for the propagator (\ref{eq:PH}).  We thus have to restore the group integration. Fortunately, by tracking  homogeneity for each term, we know that
\be
\frac{\left( \sum_{i<j} [z_i|z_j\ket[w_i|w_j\ket\right)^{J}}{J!(J+1)!} = \sum_{\sum j_i = J} \int \rd g  \prod_{i=1}^{4}  \frac{[z_i|g|w_i\ket^{2j_i}}{(2j_i)!}.
\ee
Setting this SU(2) group element to identity and summing over all $J$ allows us to get the partially gauge fixed propagator, which we denote $\one_\rho$
\be
  \one_\rho(\tilde{z}_i;\tilde{w}_i)  = \sum_{J} {}_2F_1\left(-\frac{J}{2}-1,-\frac{J}{2};2;\rho^4\right) \frac{\left( \sum_{i} [\tilde{z}_i|\tilde{w}_i\ket\right)^{J}}{J!}\label{eq:PH0}.
\ee
Note that for the convenience of notation later, we will always add a tilde on the spinors which belong to the partially gauge fixed propagator.  As in the case of the propagator, we find that setting $\rho = 0$ we recover the BF trivial propagator $\one(z_i;w_i)$, as we would expect. We can now use the homogeneity map to define a homogenized trivial propagator $\one_{\tilde{\tau}}$ as
\be
\one_{\tilde{\tau}}=e^{\tilde{\tau} \sum_i [\tilde{z}_i|\tilde{w}_i\ket} \qquad \text{with}\qquad   \tilde{\tau}^J \rightarrow F_\rho(J/2) \qquad \text{for} \qquad  \one_{\tilde{\tau}} \rightarrow \one_\rho .
\ee
We thus have arrived at the expression for the gauge fixed propagators that are necessary for the loop identity. We will have to consider however $P_\rho\circ P_\rho$ and $\one_\rho\circ\one_\rho$, rather than single propagators, as we have mentioned above. We will thus use the following homogeneity maps: for the pair of  gauge-fixed propagators we will have
\be\label{eq:homtrivial}
\one_{\tilde{\tau}}\circ \one_{\tilde{\tau}} = e^{\tilde{\tau} \sum_i [\tilde{z}_i|\tilde{w}_i\ket} \quad \text{with} \quad \tilde{\tau}^J \rightarrow  \frac{F_\rho(J/2)^2}{(1+\rho^2)^{J}} \quad\text{for}\quad \one_{\tilde{\tau}}\circ \one_{\tilde{\tau}}\rightarrow \one_\rho\circ \one_\rho ,
\ee
while for the pair of propagators $P_\rho$ we get
\be\label{eq:homprop}
G_\tau\circ G_\tau = e^{\tau \sum_{i<j}  [z_i|z_j\ket [w_i|w_j\ket} \quad \text{with} \quad \tau^J  \rightarrow \frac{F_\rho(J)^2}{(1+\rho^2)^{2J} (J+1)!} \quad \text{for} \quad G_\tau\circ G_\tau \rightarrow P_\rho\circ P_\rho. 
\ee
With this, we are ready to perform the calculation of this loop identity. The addition of the extra two pairs of gauge-fixed propagators leads to very simple contractions, using our results of spinor Gaussian integrals in the Appendix. Integrating over the three strands inside the loop leads to nearly the same calculation as in the previous section, with the difference being the addition of the trivial propagation in the extra strands connected to the gauge-fixed propagators. Using the homogeneity map, we finally find that the constrained loop identity is given by
\be
\begin{split}
&\raisebox{-12mm}{\includegraphics[keepaspectratio = true, scale = 0.45] {loopidentity.pdf}}\ \ \ \  = \ \ \ \  \sum^\infty_{A,B,J,J'=0}\!\!\!\!   \frac{N\left(A,B, J, J',\rho\right) }{A! B! J! J'!} \times \\
& \times \underbrace{ \left(\sum_{i=1}^3 [\tilde{z}^1_i|\tilde{w}^1_i\ket \right)^A }_{GF 1}  \ \ \underbrace{\left(\sum_{i=1}^3 [\tilde{z}^2_i|\tilde{w}^2_i\ket \right)^B}_{GF2}  \ \ \underbrace{\left(\sum_{i=1}^3 [z_i|w_i\ket \right)^J} _\text{Trival projection} \ \ \underbrace{\left(  \sum_{i<j<4}[z_i|z_j\ket [w_i|w_j\ket \right)^{J'}\! \!}_\text{Mixing terms},
\end{split}
\label{fullloop}
\ee
with the coefficient $N\left(A,B, J, J',\rho\right)$ given by
\be
\begin{split}
N\left(A,B, J, J',\rho\right) \equiv & \sum_{K=0}^{J'}   \frac{ J'!(J\!+K)!(J\!+\!2K\!+\!1)^\eta  }{K! (J'\!-\!K)! (J\!+\!J'\!+\!K\!+\!1)! } \frac{(-1)^{K}}{(1+\rho^2)^{(A+B +12K+7J+2J')} }  \times \\
&\times F_\rho^2\left(J+J'+K\right)  F_\rho^2\left( (A+J)/2+K\right) F_\rho^2\left( (B+J)/2+K\right),\nonumber  
\end{split}
\ee
where we have defined  $ F_\rho(J)  \equiv  {}_2F_1(-J-1,-J;2;\rho^4)$. The variables $|\tilde{z}^1_i\ket, |\tilde{w}^1_i\ket$ appear in the strands attached to the first gauge fixing term,  similarly $|\tilde{z}^2_i\ket, |\tilde{w}^2_i\ket$ appear in the second gauge fixing, while $|z_i\ket, |w_i\ket$ are are labelled for the strands we haven't gauge fixed.
 The face weight coupling constant $\eta$  should be fixed by requirements of divergence, which we will discuss in a later section.  A more detailed calculation of this loop identity can be found in the Appendix.

Even though the expression in Eq.(\ref{fullloop}) has a few layers of summations like a Russian nesting doll and the coefficients look complicated, the physical meaning behind the expression is quite clean -- up to a weight, we get the trivial propagation, like in BF theory, but we also get additional mixing terms for $J' \neq 0$. We will study the properties of this identity in section \ref{sec:looptruncation}. 

For the purpose of calculating the 4-dimensional Pachner moves, it will be convenient to again define an exponentiated expression for this loop identity, which can then be transformed into the proper expression by the homogeneity map. Before using the homogeneity map in Eq. (\ref{fullloop}), we would have an expression purely in terms of $\tau$'s that can be exponentiated. We hence define the exponentiated loop identity to have the following very simple form:
\be\label{eq:exploopid}
L_\tau(z_i, w_i; \tilde{z}^1_i, \tilde{w}^1_i; \tilde{z}^2_i, \tilde{w}^2_i)=\exp\!\!{\left(\!\sum_{i=1}^3 \tilde{\tau}_1[\tilde{z}^1_i|\tilde{w}^1_i\ket  +\tilde{\tau}_2 [\tilde{z}^2_i|\tilde{w}^2_i\ket+\tau_N  [z_i|w_i\ket + \tau_M\!\!\!\! \sum_{i<j<4}[z_i|z_j\ket [w_i|w_j\ket\!\right) }.
\ee
The full loop identity can then be recovered through the following homogeneity map:
\be\label{eq:loophommap}
 \tau_N^J  \tau^{J'}_M  \rightarrow \sum_{K=0}^{J'}\frac{(-1)^{K}(J+K)!J'!}{K!(J'-K)!}(J+2K+1)^\eta\tau^{J'-K}\left(\frac{ \tilde{\tau}_1 \tilde{\tau}_2\tau}{(1+\rho^2)^3}\right)^{J+2K}
\ee
and the $\tilde{\tau}$'s and $\tau$ keep track of the $F_\rho$ factors according to the rules given in Eq. (\ref{eq:homtrivial}) and Eq. (\ref{eq:homprop}).

%%%%%%%%%%%%%%%%%%%%%%%%%%%%%%%%%%%%%%%%%%%%%%%%%%%
\subsection{Computing Pachner Moves with Simplicity Constraints}
%%%%%%%%%%%%%%%%%%%%%%%%%%%%%%%%%%%%%%%%%%%%%%%%%%%
In this section we  compute all the Pachner moves in the 4-d holomorphic Spin Foam model based on the techniques we have developed in the previous sections.  All these moves are based on the configurations of $6$ vertices 
$(ABCDEF)$. In the following we adopt the following   notation:  a simplex $A$ indicates the 4-simplex opposite to the vertex $A$, i.e. it is composed by $[BCDEF]$.  $AE$ indicates the tetrahedron $A\cap E$ composed of the vertices $[BCDF]$, with vertex $A$ and $E$ removed from the triangulation. Triangle $ABD$  indicates  the one composed by $[CEF]$. Also in order to keep track of which vertex is  ``active'', i-e dual to a 4-simplex   and which vertex is ``inactive'', i-e not dual to a 4-simplex, we introduce  a distinction in our notation: an upper case letter $A,B\cdots$ denotes an active vertex, while a lower case letter $c,d\cdots$ denotes an inactive vertex.

%%%%%%%%%%%%%%%%%%%%%%%%%%%%%%%%%%%%%%%%%%%%%%%%%%%
\subsubsection{3--3 move}
%%%%%%%%%%%%%%%%%%%%%%%%%%%%%%%%%%%%%%%%%%%%%%%%%%%
According to these conventions the move 3--3 corresponds to 
\be
ABCdef \to abcDEF .\nonumber
\ee
The 3--3 move is shown as Fig.\ref{fig:33triangulation}. In the first figure  the 4-simplices $ A, B, C$ are sharing the blue triangle. After the move 
the configuration is changed into three  4-simplices $D, E, F$ which share the green triangle. 

\begin{figure}[h]
	\centering
		\includegraphics[width=0.7 \textwidth]{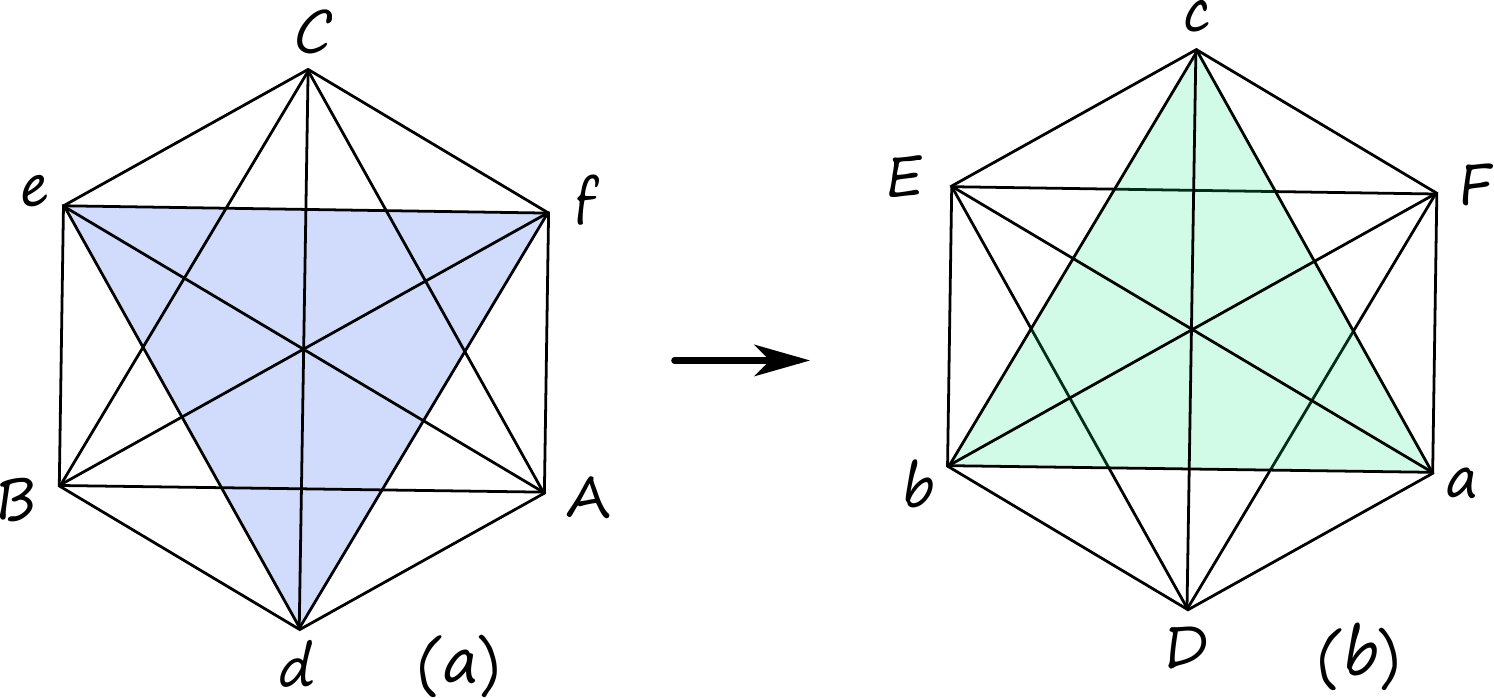}
	\caption{Triangulations for the 3--3 move.}
		\label{fig:33triangulation}
\end{figure}

The corresponding cable diagram is shown in Fig.\ref{fig:33move}.  The various colours of strands in the graph are used to indicate the different positions of triangles.  The blue loop to be integrated out corresponds to the triangle $ABC$. The purple strands in $(a)$ for example  are dual to the triangles $Adf \subset A,\ Bde \subset B,\ Cef \subset C$ and they run from the tetrahedra $Af\to Ad$, $Bd\to Be$, $Ce\to Cf$. 
After performing the 3--3 Pachner move, the same triangles (still indicated by the purple strands) are no longer shared by two tetrahedra within a given 4-simplex. They become    commonly shared by tetrahedra belonging to the three different 4-simplices: $aDF \subset (D\bigcap F ) ,\ bDE \subset (D \bigcap E),\  cEF \subset (E\bigcap F)$.  
The same happens to the black strands, whereas the opposite happens for the red and light blue strands. 
In summary, on one hand,  due to the 3--3 move from $(a)$ to $(b)$, the  red and light blue strands, shared  between different simplices in $(a)$  become unshared strands which  belong to one simplex in $(b)$. On the other hand, the unshared strands (the black and purple strands) in $a$ become the commonly shared ones in $(b)$.   The dark blue loop and the green loop correspond to  faces which are dual to the internal triangles $ABC$ and  $DEF$  respectively.

\begin{figure}[h]
	\centering
		\includegraphics[width=1\textwidth]{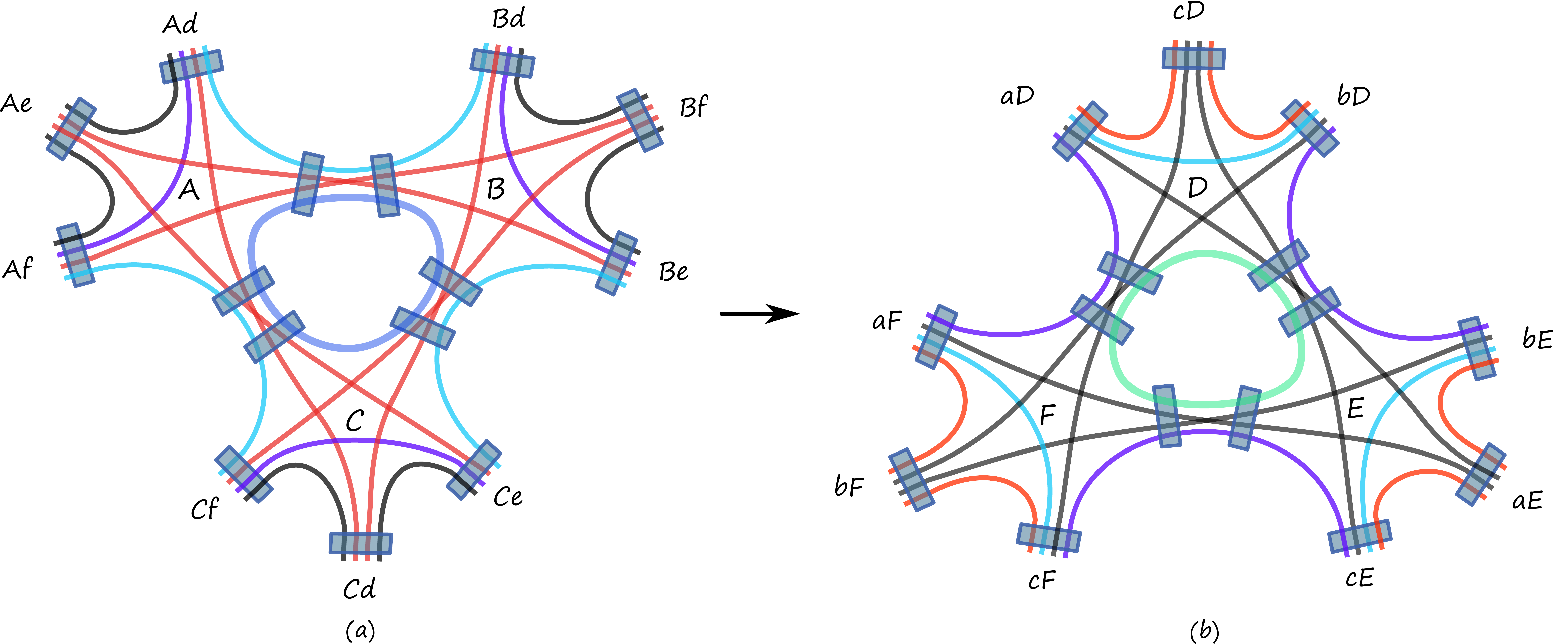}
	\caption{Cable diagram for the 3--3 move $ABCdef \to abcDEF$.  }
		\label{fig:33move}
\end{figure}
To compare the partition function/amplitudes between the configurations (a) and (b), we need to integrate out the shared loop on both sides. Based on the discussion in section \ref{sec:gf}, we can gauge fix two out of three pairs of the constrained propagators around the loop by a choice of a maximal tree in a way that leaves the amplitude invariant. We then need to apply only once  the constrained loop identity which we obtained in the previous section to complete the 3--3 Pachner move. 
In order to do so, it is important to introduce some notation for the spinors.
Let us describe the parametrization of $(a)=(ABCdef)$. For each 4-simplex $\alpha \in \{A,B,C \}$ we need to introduce a collection of spinors associated with each strand within that 4-simplex. Each strand carries a label which corresponds to a pair of tetrahedra $\alpha \beta$ sharing a face. Within $A$ we have two types of tetrahedra: three external ones $Ad, Ae, Af$ and two  internal ones $AB,AC$. 
The strands run either between  two internal tetrahedra or from one internal to one external tetrahedron.
Accordingly, we label the external strands by boundary spinors $z^{\alpha \beta}_\gamma$  where 
$\alpha \in \{A,B,C \}$, \  $\beta\in\{ d,e,f\}$, \ $\gamma\in\{ A,B,C,d,e,f\}$ for $(a)$  in Fig.\ref{fig:33move} , and $\alpha \in \{D,E,F \}$, \  $\beta\in\{a,b,c\}$, \ $,\gamma\in\{ a,b,c,D,E,F\}$ for $(b)$.  $\alpha \beta$ are the indices labeling boundary tetrahedra, and $z^{\alpha \beta}_\gamma$ indicate boundary spinors. The boundary propagators are then labeled as $P_\rho (z^{\alpha \beta}_\gamma; w^{\alpha \beta}_\gamma)$. 

\begin{figure}[h]
	\centering
		\includegraphics[width=0.45\textwidth]{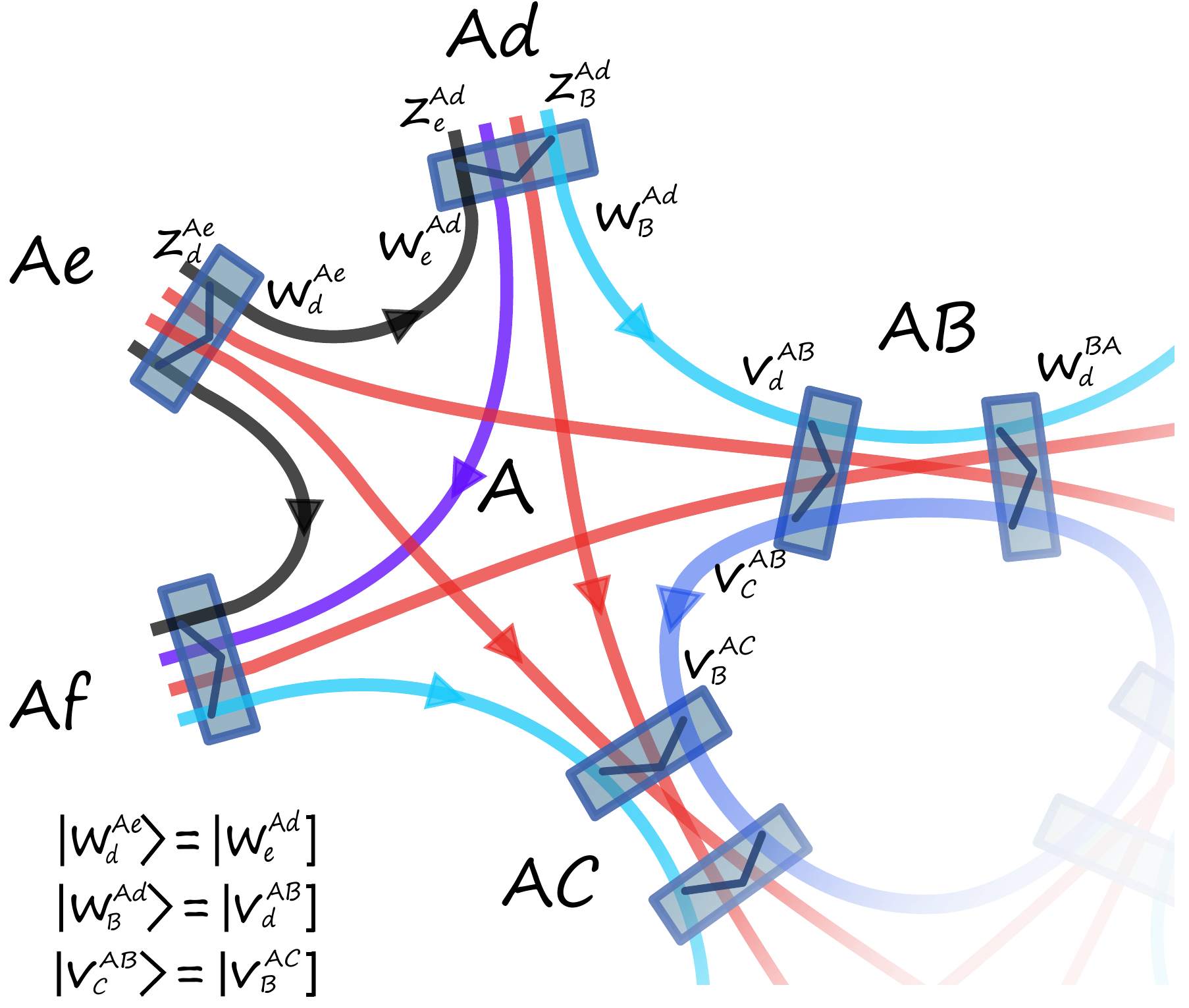}
	\caption{Zoomed in part of the cable diagram for the 3--3 move with some of the labels and contractions of spinors explicitly written down.}
		\label{fig:33zoom}
\end{figure}

Let us label the internal pairs of propagators by $P_\rho\circ P_\rho (v_\gamma^{\alpha\alpha '}, w_\gamma^{\alpha '\alpha})$, where $\alpha,\alpha' \in \{A,B,C \}$ for $(a)$ and  $\alpha,\alpha' \in \{D,E,F \}$ for $(b)$. We need to contract these spinors with the spinors $w_\gamma^{\alpha\beta}$ of the external propagators. An example of this is shown in Fig. \ref{fig:33zoom} with all the labels and orientations written explicitly of a part of $(a)$. The contractions are done according to the orientations of strands, and for example we have $|w^{Ad}_B\ket = |v_d^{AB}]$.  In summary, the amplitude is constructed from $z^{\alpha \beta}_\gamma$ and $w^{\alpha\beta}_\gamma$ for the external propagators and on $ w^{\alpha\alpha'}_\gamma , v^{\alpha\alpha'}_{\gamma'}$ for the internal ones. The amplitude is obtained then after integration over the internal spinors after imposing the contractions, thus becomes a function of $z^{\alpha \beta}_\gamma$ only.

We thus find that the amplitude for three 4-simplices combined as in Fig.\ref{fig:33move} can be written  as
\be
\mathcal{A}_{3} (z^{\alpha \beta}_\gamma) =\int \prod_{all} d\mu_\rho(v)d\mu_\rho(w) \prod_{\alpha \beta} P_\rho (z^{\alpha \beta}_\gamma; w^{\alpha \beta}_\gamma) \cdot \ \  \raisebox{-10mm}{\includegraphics[keepaspectratio = true, scale = 0.4] {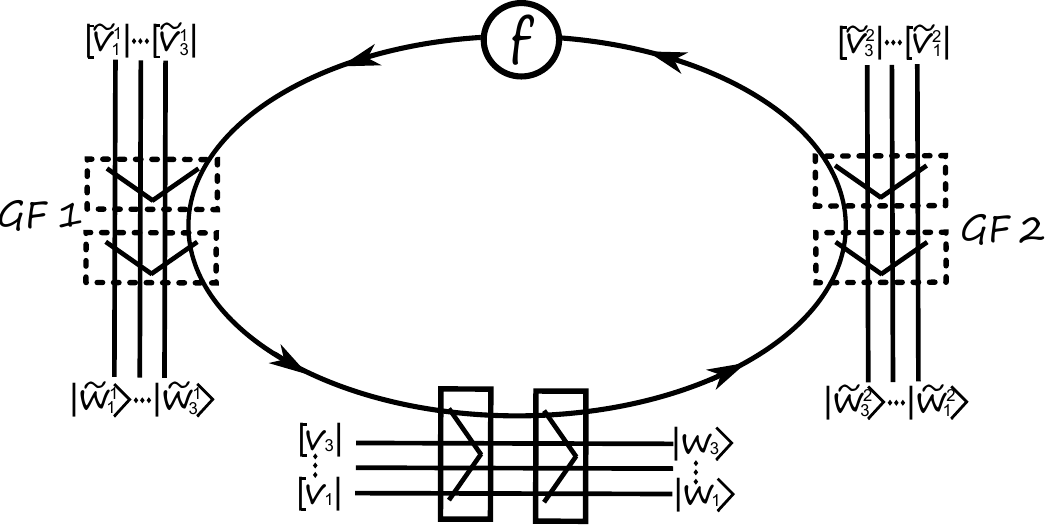}}.
\ee
The spinors of the three internal  propagators which share a loop are labeled by $v$ and $w$ and each of them is contracted with different boundary constrained propagators, with the gluing depending on the orientation of the graph. 

The crucial difference between amplitude $(a)$ and $(b)$ is that  the non-trivial coefficient $N(J,J',A,B,\rho)$ of Eq.(\ref{fullloop}) encodes the spin information of different strands. In $(a)$, the coefficient  $N$  encodes the spin information of the blue and red strands in one configuration, while in (b) it encodes the spin of the black and purple strands. Unless the corresponding boundary spins are chosen to be  the same, the 3--3 move cannot be invariant. 

It is thus very easy to see where the topological invariance of BF theory is broken. Let us come back to BF theory and look at the 3--3 move. The BF loop identity (\ref{eqn_loop_identity})  does not have any factor depending on spins and hence gives a trivial equality, as the diagrams in Fig. \ref{fig:33moveBF} are combinatorially equivalent. Thus for BF theory, the partition function/amplitudes are invariant under 3--3 move.

\begin{figure}[h]
	\centering
		\includegraphics[width=1\textwidth]{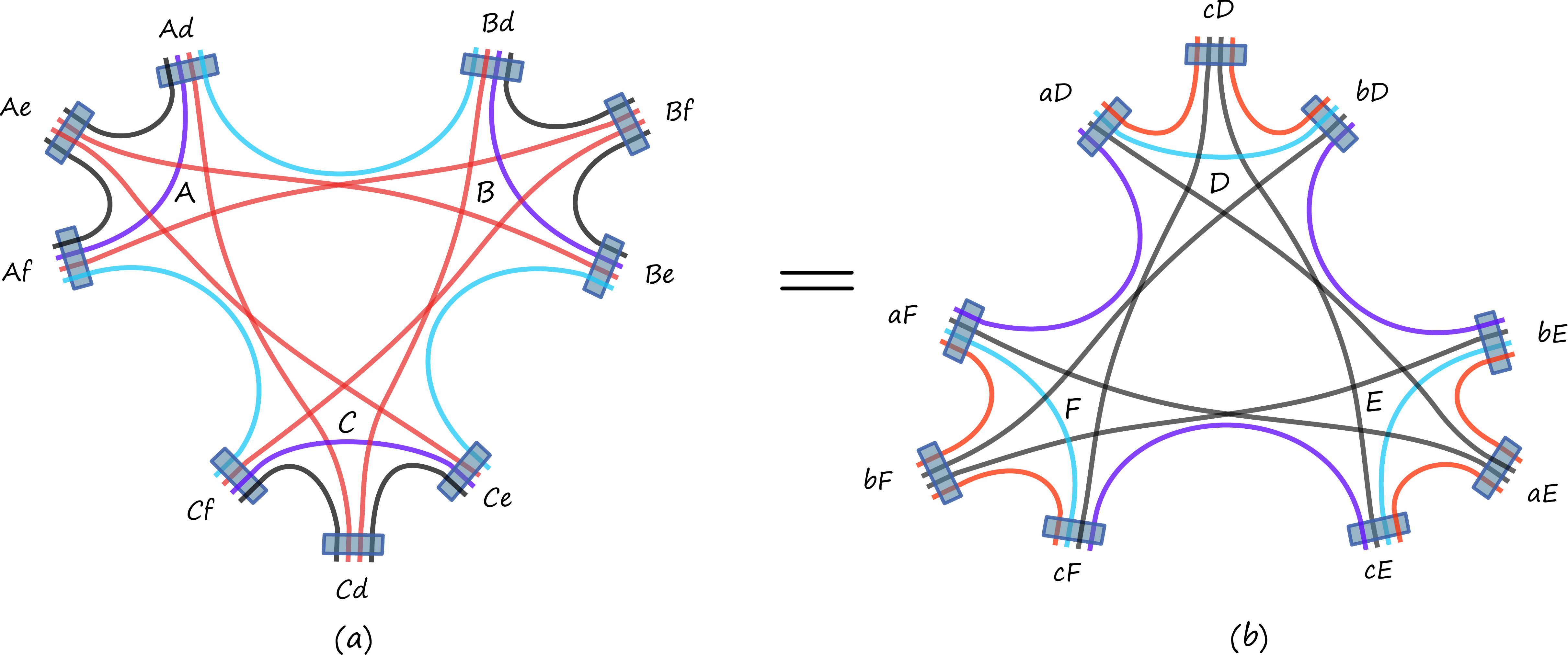}
	\caption{For 4-d BF theory, after integrating out the middle loops in the 3--3 move, the rest of the strands are combinatorially equivalent.}
		\label{fig:33moveBF}
\end{figure}

%%%%%%%%%%%%%%%%%%%%%%%%%%%%%%%%%%%%%%%%%%%%%%%%%%%
\subsubsection{4--2 move}
%%%%%%%%%%%%%%%%%%%%%%%%%%%%%%%%%%%%%%%%%%%%%%%%%%%
The 4--2 move $ ABCDef \mapsto abcdEF$ is shown in Fig.\ref{fig:42movetri}. In $(a)$, four 4-simplices $A,B,C,D$ are sharing 6 tetrahedra. After removing four triangles (or four loops in the dual cable graph) and changing the combinatorial structure, the four 4-simplices are rearranged in two 4-simplices $E,F$ glued by one tetrahedron. The corresponding cable diagram of the four 4-simplices is shown in Fig.\ref{fig:42movecab}. 

\begin{figure}[h]
	\centering
		\includegraphics[width=0.75\textwidth]{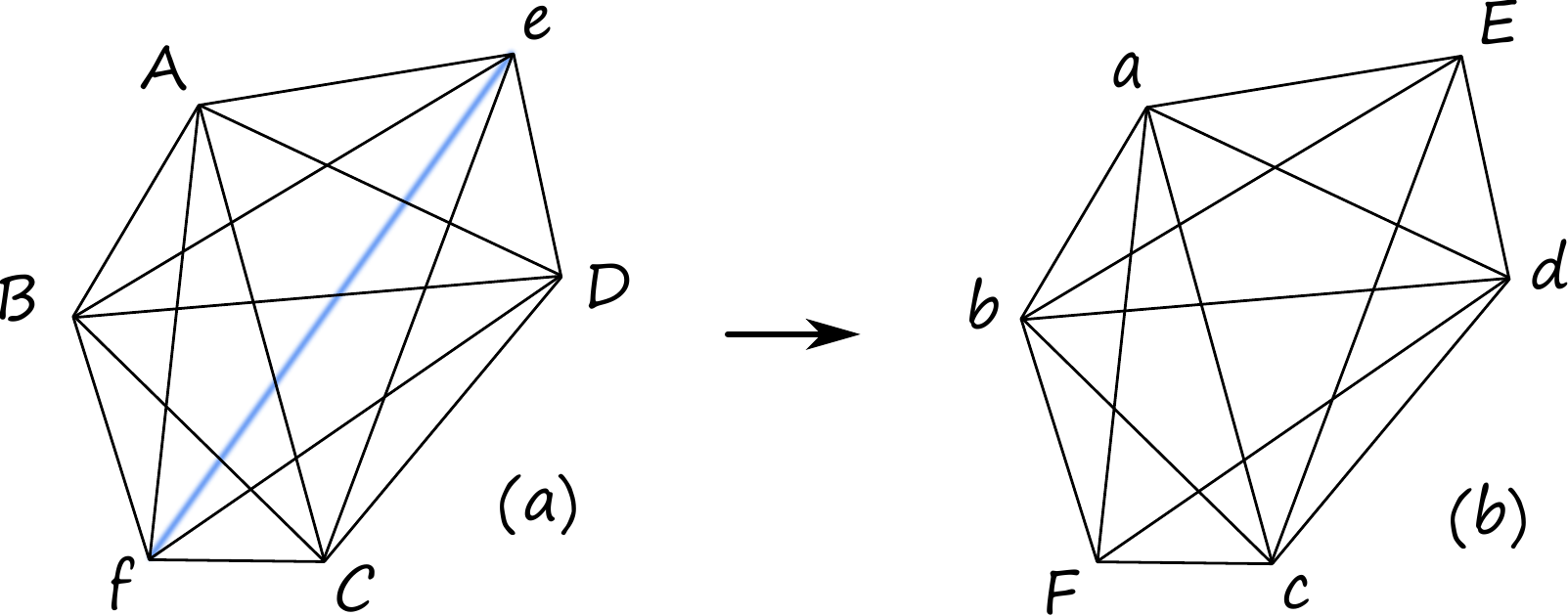}
	\caption{Triangulations for the 4--2 move.}
		\label{fig:42movetri}
\end{figure}

\begin{figure}[h]
	\centering
		\includegraphics[width=0.65 \textwidth]{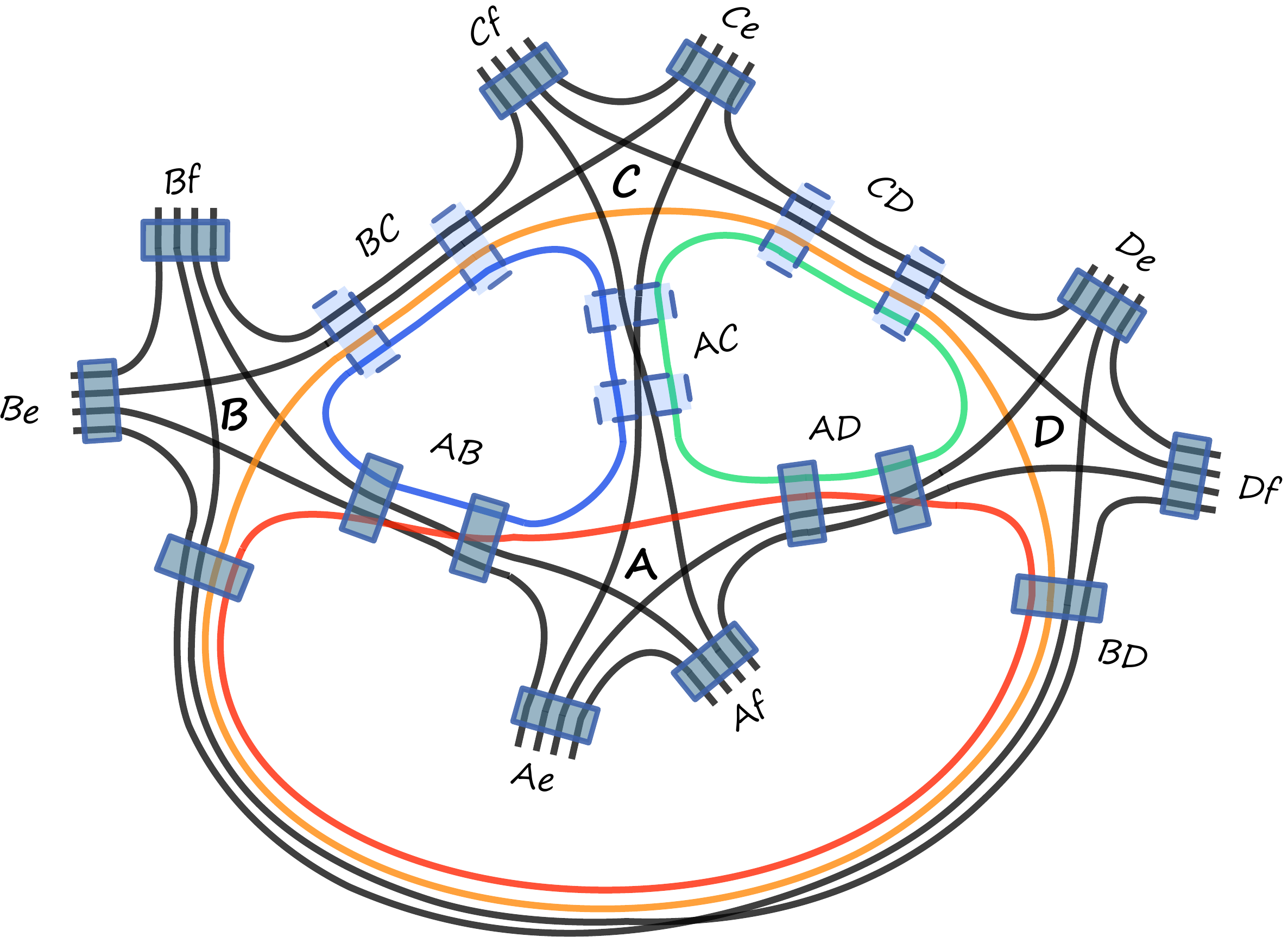}
	\caption{Cable diagram for the 4--2 move with gauge fixing along  $BC, AC, CD$.}
		\label{fig:42movecab}
\end{figure}
We can perform gauge fixing of this graph by choosing vertex $C$ as the root of the maximal tree in such a way that we can gauge fix 3 couples of propagators $BC, AC, CD$. This allows us to apply the constrained loop identities Eq.(\ref{fullloop}) to three of the four loops. More specifically, we can apply the constrained loop identity to the propagators $(AB, BC, CA)$ to drop the blue loop, then apply it to the propagators $(AC,CD,DA)$ to integrate  the green loop and propagators $(BC, CD, DB)$ to remove the big yellow loop. This results in integrating out all couples of constrained propagators, and hence we are left with one last (red) loop, which is mixed with the external strands, as can be seen in Fig.\ref{fig:nonlocalgluing}.

\begin{figure}[h]
	\centering
		\includegraphics[width=0.3\textwidth]{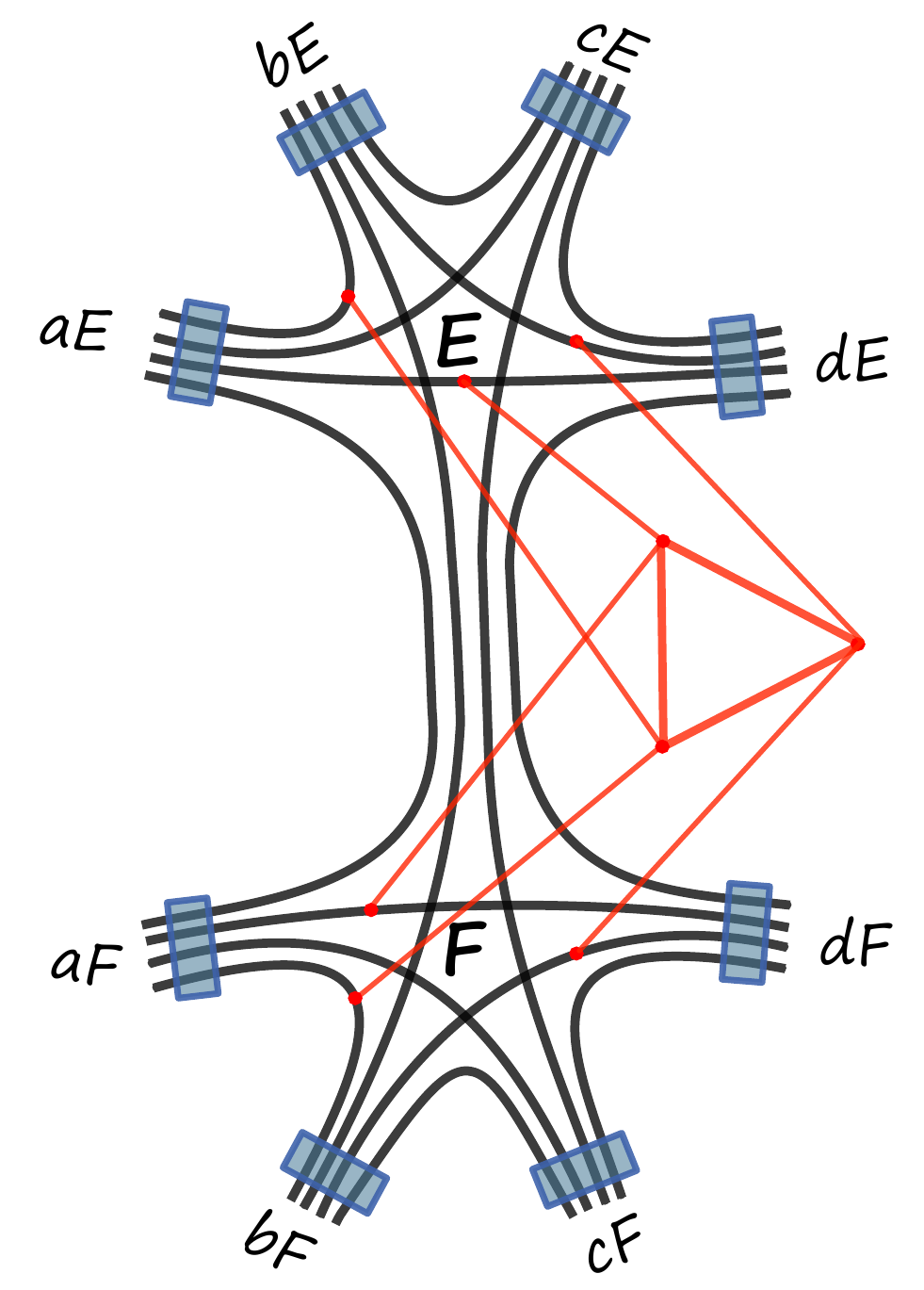}
	\caption{Performing the calculation we get a configuration of two 4-simplices with a nonlocal gluing.}
		\label{fig:nonlocalgluing}
\end{figure}

Note, that we have applied the three loop identities, but the last loop is left without any extra group averaging. 
Similar to the  case of the loop identity, we have to add in a face weight for this last loop. We will do so again by inserting a factor of $\tau'$ on one of the strands of the left-over loop (say the red strand for edge AD), so that we can use homogeneity map $\tau '^{2j}\rightarrow (2j+1)^\eta$. 

%%%%%%%%%%%%%%%%%%%%%%%%%%%%%%%%%%%%%%%%%%%%%%%%%%%%%%%%%%%%%
Similar as in the previous section, we will denote the spinors on the boundary as $z$, and spinors in the bulk as $w$ and $v$ with indices labeling the propagator and the strand they belong to.  Each spinor carries three indices:  $z^{\alpha \beta}_\gamma$ with indices $\alpha$ labeling the 4-simplex, $\alpha \beta$ labelling the tetrahedron they belong to, $\gamma$ labelling which strands they represent. With assuming a specific orientation of the graph as $C \rightarrow A, C \rightarrow B, C \rightarrow D, A\rightarrow B, D\rightarrow A ,D\rightarrow B$ \footnote{When one reverses the orientation of one propagator,  the corresponding $[v|w\ket \rightarrow [w|v\ket = - [v|w\ket  $}, the amplitude in terms of  the exponentiated loop identity Eq.(\ref{eq:exploopid}) is given then by
\be\label{eq:42move}
\begin{split}
\mathcal{A}_{4-2}^\tau (z^{\alpha \beta}_\gamma)& = \int \left\{\prod_{\text{all}} d\mu_\rho(v) d\mu_\rho(w)\right\}  \prod_{{\alpha \beta} }  P_{\rho } (z^{\alpha \beta}_\gamma ; w^{\alpha \beta}_\gamma)   \cdot \exp \left[\small{\sum}_{  \sigma i} {\tilde{\tau}_{ \sigma C}  [\tilde{v}^{ C \sigma }_i|\tilde{w}^{\sigma C}_i \ket}  \right] \\
 &\times \exp \left[\small{\sum}_{ \mu \nu} { (\tau^{ \mu \nu}_N \small{\sum}_j \alpha_j^{\mu\nu}[v^{\mu \nu}_j|w^{\nu \mu}_j \ket +\tau^{ \mu \nu}_M \small{\sum}_{j < k} \alpha_j^{\mu\nu}[w^{\nu \mu}_j|w^{ \nu \mu}_k\ket [v^{\mu \nu}_j|v^{\mu \nu}_k \ket  }  ) \right] .
\end{split}
\ee
For the external propagators $\alpha \in\{A,B,C,D\}$ and $\beta \in \{e,f\}$ label the tetrahedron, while $\gamma \in\{A,B,C,D,e,f\}$ labels the strands in each tetrahedron. For internal gauge fixed propagators, $\sigma \in \{ A, B, D\},\  i \in \{ e,f\}$, and for the non-gauge fixed propagators, $\mu \nu \in \{AB,AD,BD\},\  j,k \in \{e,f,r \}$, where $r$ indicates the red strand of the left-over loop.  We define $\alpha_j^{\mu\nu}$ as
\be
\alpha_j^{\mu\nu} = 1 + \delta^{\mu\nu}_{AD}\delta_j^r \left(\tau'-1\right)
\ee
for keeping track of the homogeneity factor for the face weight of the last loop.

The equation (\ref{eq:42move}) gives a compact and  explicit expression for the amplitude associated with the 4--2. It is obtained by  using the exponentiated loop identity Eq.(\ref{eq:exploopid}), which then can be transformed using the homogeneity map to obtain the full expression after performing all of the contractions of spinors and all the Gaussian integrals. The homogeneity maps we  need to apply to this expression to get the full result were defined in  Eq. (\ref{eq:homtrivial}) for the $\tilde{\tau}$,  in Eq.(\ref{eq:loophommap}) for $\tau_N$ and $\tau_M$ and the homogeneity map for $\tau '$ is  $\tau '^{2j}\rightarrow (2j+1)^\eta$. The calculation can be straightforwardly done, but the resulting expression itself is a complicated,  one with lots  of mixed strands that is  difficult to manipulate.
The integrals also contain potential divergences that have to be taken care of.
 We will delay the discussion of  the resulting expression and the significance of the mixing terms   until the next section, 
 as we first  encounter a similar bahaviour for the 5--1 Pachner 
 move as well.  Here in the expression Eq.(\ref{eq:42move}), we intentionally leave the last red loop unintegrated to pave the way for truncation in section \ref{sec:coarse}.

%%%%%%%%%%%%%%%%%%%%%%%%%%%%%%%%%%%%%%%%%%%%%%%%%%%
\subsubsection{5--1 move}
%%%%%%%%%%%%%%%%%%%%%%%%%%%%%%%%%%%%%%%%%%%%%%%%%%%

We  now  calculate the 5--1 Pachner move. The 5--1 move corresponds to a change of a configuration of five 4-simplices sharing an internal vertex into a single 4-simplex by removing the common vertex, see Fig. \ref{fig:51triangulation}.

\begin{figure}[h]
       \centering
               \includegraphics[width=0.75\textwidth]{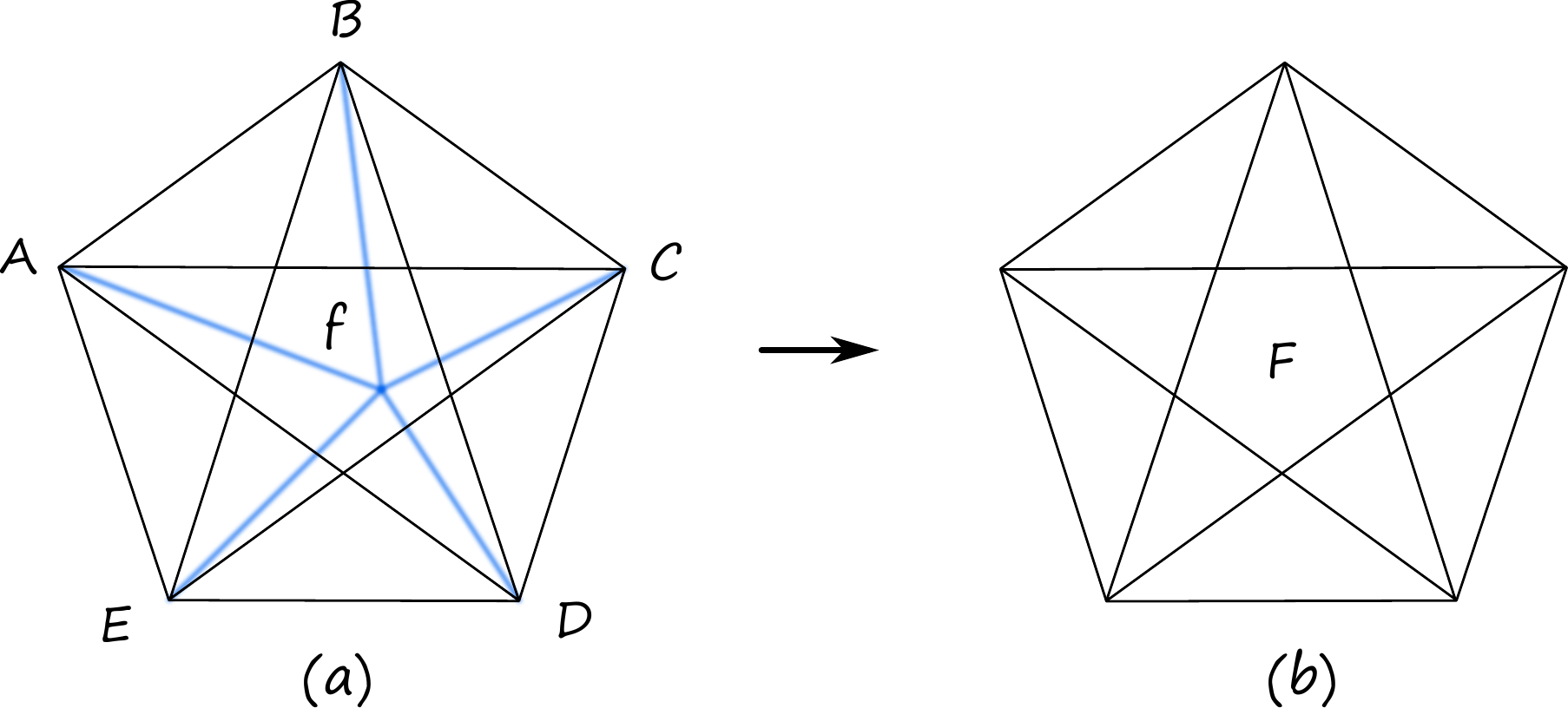}
       \caption{Triangulations for the 5--1 Pachner move.}
               \label{fig:51triangulation}
\end{figure}

  The cable diagram for this move can be seen in Fig. \ref{fig:51move}.
\begin{figure}[h]
	\centering
		\includegraphics[width=0.6\textwidth]{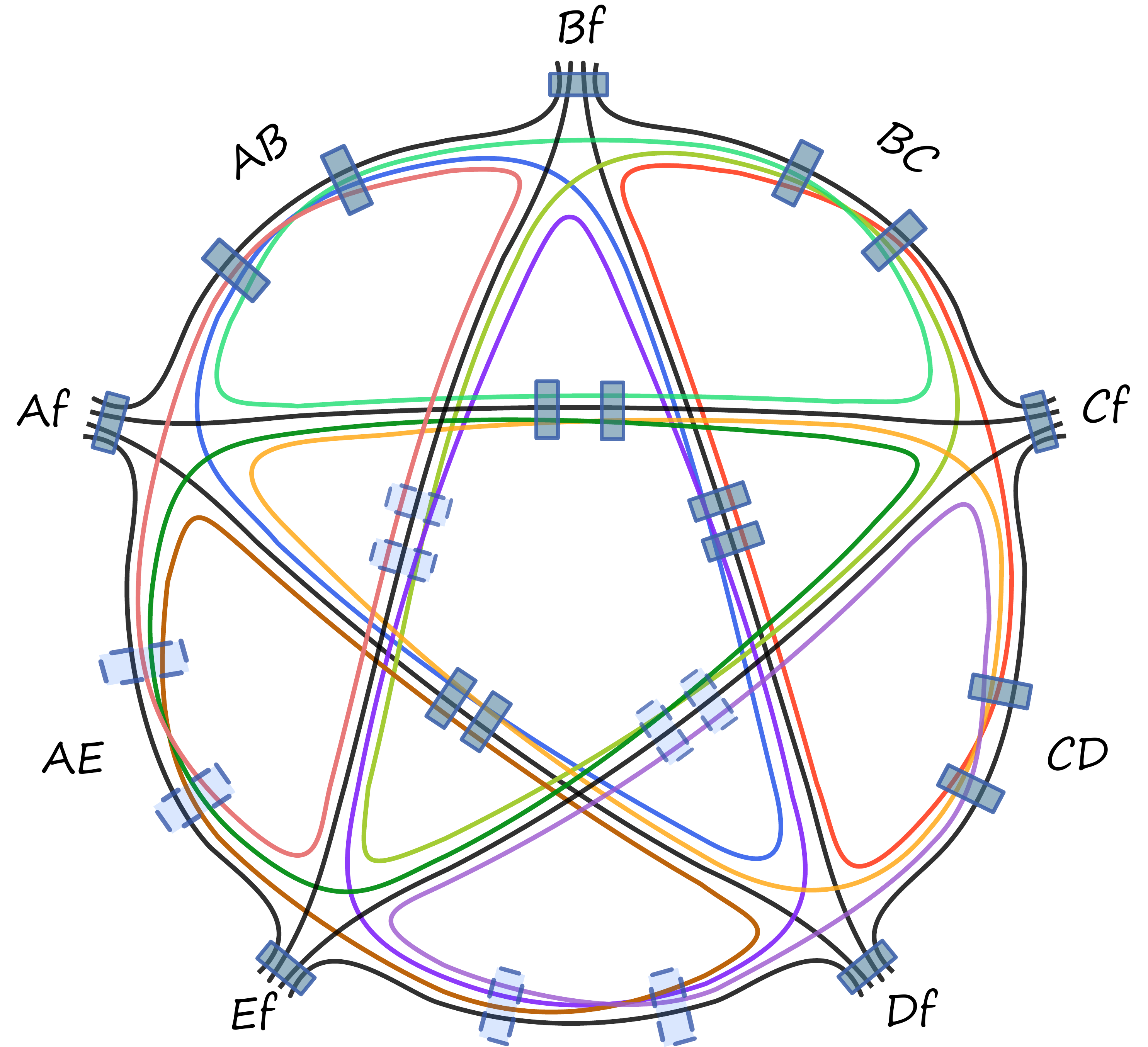}
	\caption{Cable diagram for the 5--1 move. The loops inside are colored.}
		\label{fig:51move}
\end{figure}
We have a total of 10 loops and 10 pairs of constrained propagators inside the bulk of the graph. Even though there is an increase in complexity, compared to the 4--2 move, the calculation will go over in nearly the same way. We start by choosing a maximal tree in the diagram, which allows us to gauge fix 4 of the pairs of propagators. A careful choice of this tree corresponds to a root at one of the 4-simplices and allows us to apply loop identities to 6 of the loops, leaving us with 4, as can be seen in Fig. \ref{fig:51fixed}.

\begin{figure}[h]
	\centering
		\includegraphics[width=0.6\textwidth]{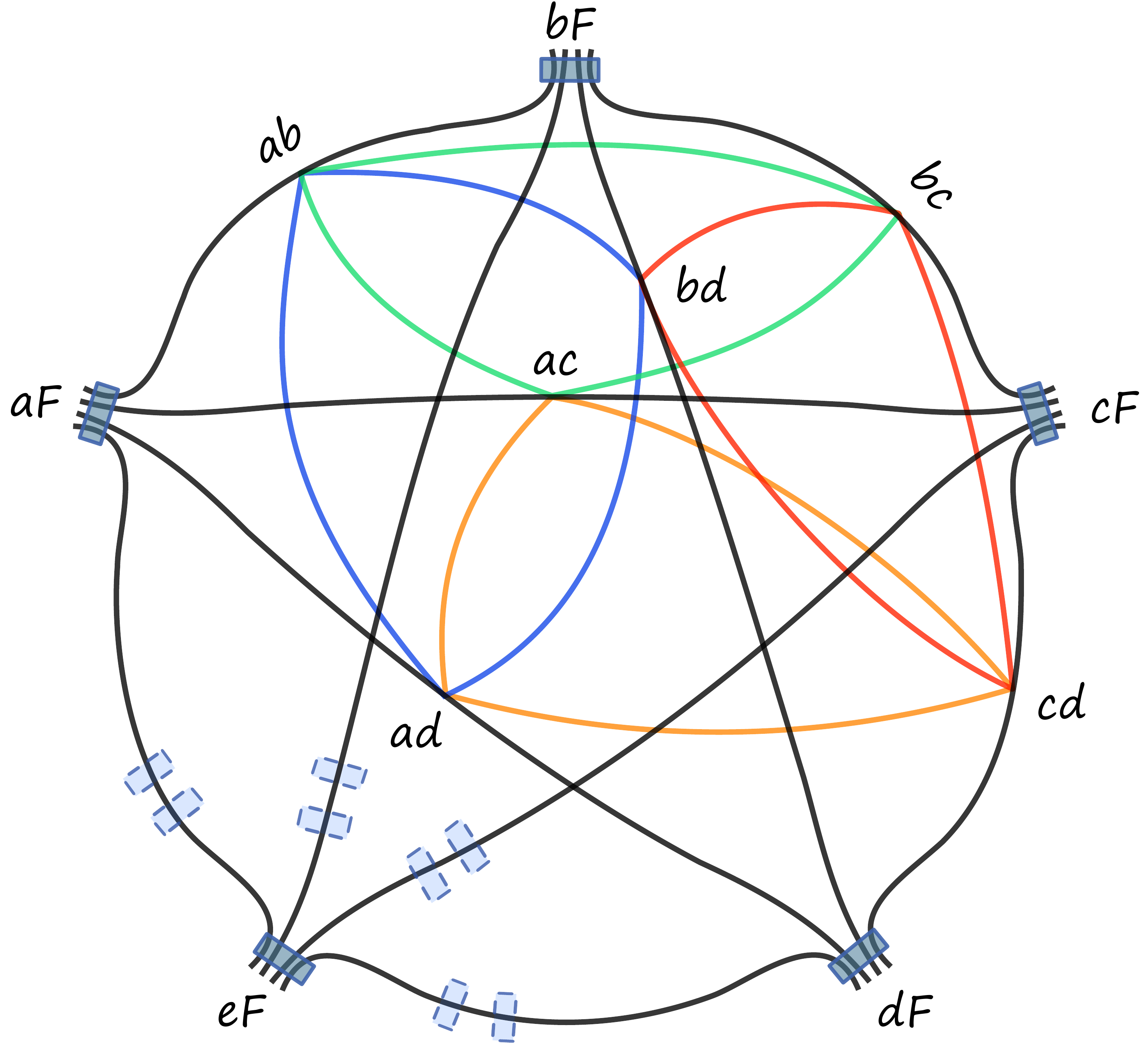}
	\caption{Gauge-fixing 4 strands allows to apply loop identities 6 times, leaving the 4 colored loops.}
		\label{fig:51fixed}
\end{figure}

We can write the amplitude for the 5--1 move using the exponentiated loop identity Eq.(\ref{eq:exploopid}) as in the case of the 4--2 move. We will again have to add the face weights for the last four loops by adding factors of $\tau '$. The expression for the full Pachner move then would be obtained by applying the homogeneity map to the resulting power series. We keep to the notation of inside spinors being $w$ and $v$ labeled by the strands and propagators they belonged to.
With assuming the orientation of the graph as $E \rightarrow A, E \rightarrow B, E \rightarrow C, E \rightarrow D$, the amplitude in terms of boundary spinors $z$ is formally given then as 
\be
\begin{split}\label{eq:51tau}
\mathcal{A}_{5-1}^\tau (z^{\alpha f}_\gamma)& = \int \left\{\prod_{\text{all}} d\mu_\rho(v) d\mu_\rho(w)\right\}   \prod_{{\alpha } }  P_{\rho } (z^{\alpha f}_\gamma ; w^{\alpha f}_\gamma)   \cdot \exp \left[\small{\sum}_{   \beta } {\tilde{\tau}_{E \sigma}  [\tilde{v}^{E \sigma}|\tilde{w}^{\sigma E } \ket}  \right] \\
 &\times \exp \left[\small{\sum}_{ \mu \nu} { (\tau^{ \mu \nu}_N \small{\sum}_i \beta_i^{\mu\nu}[v^{\mu \nu}_i|w^{\nu \mu}_i \ket +\tau^{ \mu \nu}_M \small{\sum}_{i < j} \beta_i^{\mu\nu}[w^{\nu \mu}_i|w^{\nu \mu}_j\ket [v^{\mu \nu}_i|v^{\mu \nu}_j \ket  }  ) \right] ,
\end{split}
\ee
where the the indices run over the following ranges: $ \sigma \in \{ A, B, C,D\},\  \mu \nu \in \{AB, AC, AD, BD, BC, CD\},\  i,j \in \{f,b,r,y,g \}$, where $b,r,y,g $ indicates the blue (ABD), red (BCD), yellow (ACD), green (ABC) strands of the left-over loops respectively, and $f$ indicates the black strands  which compose the simplex F after the move. The external propagators $ P_{\rho } (z^{\alpha f}_\gamma ; w^{\alpha f}_\gamma)$ are defined the same way as in previous sections, namely $\alpha \in\{A,B,C,D,E\}$ labels the simplices in which the boundry tetrahedra belong to, and $\gamma$ labels the strands in each tetrahedra.  The coefficients $\beta_i^{\mu\nu}$ that keep track of homogeneity of the face weights are defined this time as
\be
\beta_i^{\mu\nu}= 1 + \delta^{\mu\nu}_{AD}\delta_i^y \left(\tau_y'-1\right) + \delta^{\mu\nu}_{AC}\delta_i^g \left(\tau_g'-1\right)+\delta^{\mu\nu}_{AB}\delta_i^b \left(\tau_b'-1\right)+\delta^{\mu\nu}_{BC}\delta_i^r \left(\tau_r'-1\right) .
\ee
The formal expression of 5--1 is of similar structure as the 4--2 move, with the difference being the range of the indices due to bigger number of loops and propagators. The expression (\ref{eq:51tau}) is relatively compact for such a complicated calculation and it contains all the information necessary to evaluate the amplitude after the Gaussian integrations are performed. In order to do so we just need to specify is  the homogeneity map 
\be
H_{5-1}[\mathcal{A}_{5-1}^\tau] = \mathcal{A}_{5-1}.
\ee
The 5--1 homogeneity map $H_{5-1}$  is  given by the composition of : 
\be\label{eq:pachnerhom}
\begin{split}
& \tau_N^{\mu\nu J}  \tau^{\mu\nu J'}_M  \!\rightarrow \!\sum_K\!\frac{(-1)^{J'\!-\!K}(J\!+\!J'\!-\!K)!J'!}{K!(J'\!-\!K)!}(J\!+\!2J'\!-\!2K\!+\!1)^\eta\tau_{\mu\nu }^K\left(\frac{ \tilde{\tau}_{E\mu} \tilde{\tau}_{E\nu}\tau_{\mu\nu}}{(1+\rho^2)^3}\right)^{\!\!J\!+\!2J'\!-\!2K} \\
&\qquad\quad\tau_{\mu\nu}^J  \rightarrow \frac{F_\rho(J)^2}{(1+\rho^2)^{2J} (J+1)!},\qquad\tilde{\tau}_{E\sigma}^J \rightarrow  \frac{F_\rho(J/2)^2}{(1+\rho^2)^{J}},\qquad\tau_i '^{2j}\rightarrow (2j+1)^\eta ,
\end{split}
\ee
with $F_\rho(J)$ previously defined as the hypergeometric function $F_\rho(J) = {}_2F_1(-J-1,-J;2;\rho^4)$. The same map can be used to find the full expression for the 4--2 Pachner move as well. The Gaussian integrals for the last four loops can be performed explicitly. Using the results from \cite{Freidel:2012ji}, we can write this as an inverse of a determinant of a large matrix. We leave these integrals undone however to make the trucation procedure in the next section more clear.

Let us now try to understand our result. In BF theory the 5--1 Pachner move would lead to 4 decoupled  loops, each giving a factor of a SU(2) delta function evaluated at identity. This would correspond to setting all the $\tau_M$s to 0 and all the other $\tau$s to 1 in our expression. For the constrained propagator,  as in the previous case of the 4-2 move, the loops inside are coupled to each other and  to the strands of the boundary spinors. This means that as expected  the spin foam model we consider is not  invariant under both the 4--2 and 5--1 Pachner moves. It is natural to conjecture here, that this would be the case for the other spin foam models as well.

The new feature of the model is the mixing between internal loops and external edges that creates a coupling between all the different strands not present in the original form of the vertex amplitude.

Let us try to study this mixing in some more detail.
By splitting the 6-valent vertices in the loops, as in Fig. \ref{fig:vertex with a heart}, it is obvious that we can try to interpret these coupled loops as an insertion of an operator. 
\begin{figure}[h]
	\centering
		\includegraphics[width=0.4\textwidth]{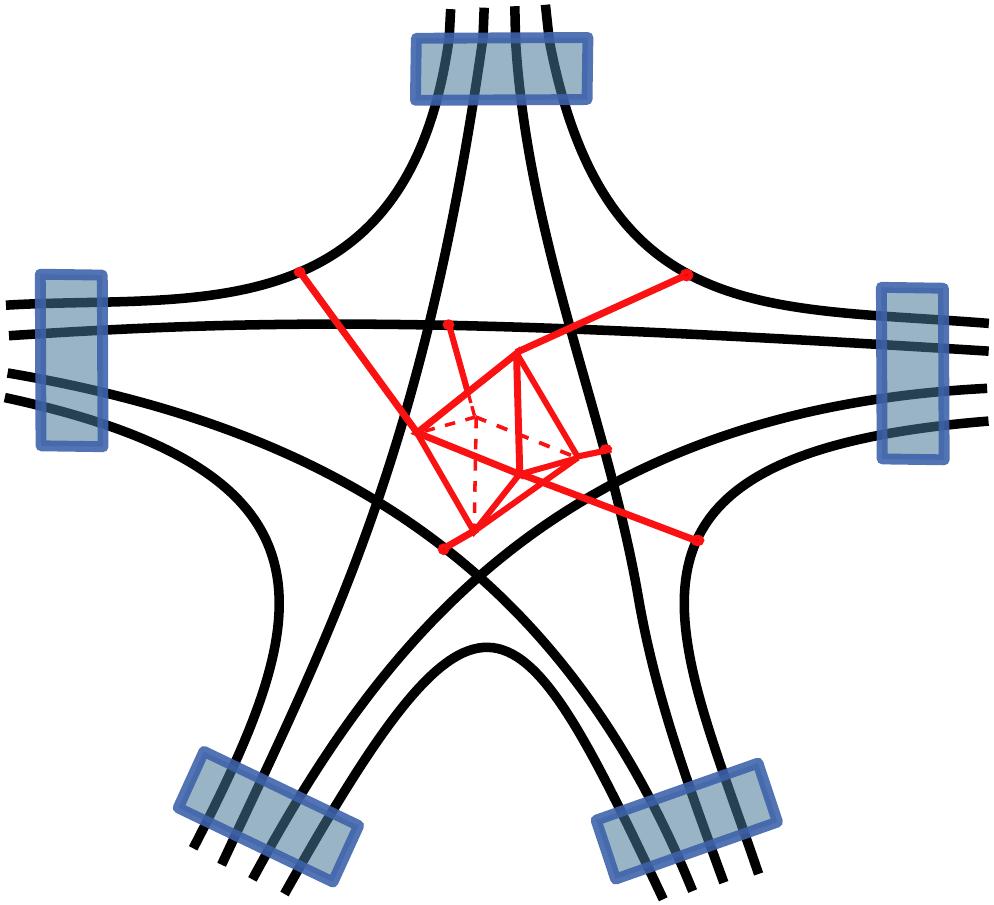}
	\caption{Performing the calculation we get a 4-simplex with an insertion of a nonlocal operator.}
		\label{fig:vertex with a heart}
\end{figure}
The connections between loops and the boundary spinors corresponds to gauge invariant operators inserted inside the 4-simplex amplitude.
It is well known that such operators can be expressed as a sum of grasping operators \cite{FreidelK}.

In the holomorphic context these operators are due to the insertions of the SU(N) operators  \cite{Freidel:2010tt},  from which all geometrical operators are made. The insertion of Wilson loops and the action of  SU(N) operators  are two sides of the same coin \cite{Livine:2013zha} -- they are constructed from the same type of gauge-invariant observables, which in our language are the products $[z|w\ket$ and $\bra z|w\ket$.  The operators we get for the 4--2 and 5--1 moves can be thus thought as an exponentiated combination of SU(N) grasping operators  and Wilson loops. Iteration of 5-1 moves leads to a new kind  of loop expansion, reminiscent of higher order diagrams in perturbative quantum field theory. It might be interesting to flesh out more this correspondence and understand  if this series converges to some interesting object. 
We leave this question for future work since this requires to first  disentangle the divergent part from the part that purely acts as grasping and leads to mixing of strands. We will now try a different approach to understanding these operators.

%%%%%%%%%%%%%%%%%%%%%%%%%%%%%%%%%%%%%%%%%%%%%%%%%%%
\section{Towards coarse graining}\label{sec:coarse}
%%%%%%%%%%%%%%%%%%%%%%%%%%%%%%%%%%%%%%%%%%%%%%%%%%%

In the previous section we have see that the mixing of strands could be understood as the insertion of a SU$(N)$ grasping operator.  In this section we mainly focus on the 5--1 move. This move can be understood as a coarse graining move which maps one choice of the vertex amplitude to another one obtained after coarse graining. All the other course graining moves have to be built out of non-trivial combinations of 3--3, 4--2 and 5--1 moves. As we have shown that 5--1 move generates non-local couplings via the mixing terms, similarly to what happens in Real Space Renormalization Group calculations. Remarkably it turns out that the mixing terms are clearly subdominant.
This motivates a truncation scheme in which we keep only the non-mixing terms in the 5--1 move leading to a specific renormalisation scheme for the vertex amplitude. This is what we analyse in this section.

A truncation scheme is usually associated with a choice of what are  the relevant and  irrelevant couplings. In the usual setting this choice is tied up with the assumptions of  locality but also has to be compatible with the symmetries like Lorentz invariance and eventually should be compatible with unitarity.
These concepts needs to be replaced by others in the case of Spin Foam renormalisation.
The current Spin Foam models, including the holomorphic one we study here, are defined in such a way that they possess the correct leading semi-classical behaviour at the level of a single 4-simplex. In hopes of defining  a continuum theory down the line, the requirement of correct asymptotics should be kept unchanged at each step of truncation in the coarse graining procedure.
Apart from this requirement, the only other one that is obvious is the preservation of gauge symmetries. 
In the next subsection, we will see that a natural truncation scheme does seem to exist for the Pachner moves already at the level of the constrained loop identity and it preserves the above requirements.

In order to successfully coarse grain the non-local operators in the Pachner moves, we need to understand and deal with their  divergences, which we study in section \ref{sec:divergences}. It is important to appreciate that the divergences in the 5--1 move  are welcomed in Spin Foam models, since  ultimately we would like to understand them as coming from a left over of diffeomorphism  invariance. 
More precisely they should represent  a  translation symmetry of the internal vertex, that should be removed by some appropriately defined Fadeev-Popov gauge-fixing procedure, similarly to what has been achieved in 3d \cite{Freidel:2002dw}. 
It is natural in our context, to control  the presence of potential divergences by introducing  parameters like $\eta$  determining the strength of the  face weights and absorb the divergence into them (and perhaps into the other coupling constants already present, like $\rho$, the Newton constant $G_N$ together with a cosmological constant $\Lambda$, or even the $\tau$ parameters that we treated so far as book-keeping parameters).

After having understood the truncation and divergences, we  have to perform the renormalization step, which  entails absorbing the relevant part of the operator into the definition of the constrained projector. This  define for us a flow $P_\rho \rightarrow \tilde{P}_\rho$ in the space of constrained propagators. We  leave the detailed study of this last step to future investigations.

It is interesting to note that besides the truncation we perform, we could also study the effect of the insertion of the mixing terms which are subleading contributions. For the 5--1 Pachner move, we could in principle integrate out the non-local operator. Once the divergence is removed the effect of the mixing terms leaves us with an amplitude that is more involved than a simple  4-simplex graph. It corresponds to a more general structure of all strands being mixed in the middle of the vertex, giving rise to higher-valent intertwiners. This suggests that the additional contributions  would  allow the theory to flow to higher-valent vertex amplitudes. 

%%%%%%%%%%%%%%%%%%%%%%%%%%%%%%%%%%%%%%%%%%%%%%%%%%%
\subsection{Truncation of the loop identity}\label{sec:looptruncation}
%%%%%%%%%%%%%%%%%%%%%%%%%%%%%%%%%%%%%%%%%%%%%%%%%%%
In this section we  introduce a truncation scheme for the Pachner moves, that will ultimately allow us to define the renormalization flow. The expression for the 5--1 Pachner move in Eq.(\ref{eq:51tau}) is very compact, but requires us to perform many extra integrations over spinors, each of which in itself is straightforward, but the resulting answer is rather long.  To simplify the discussion, let us drop the dependence on the external spinors, which corresponds to setting the boundary spins to zero. As we will discuss in the next section, this selects out the most divergent part of the Pachner move. With this simplification, we can use the techniques introduced in \cite{Freidel:2012ji} and perform all of the spinor integrals immediately, with the result being again the inverse of a determinant. The power series expansion is however very large, depending on the order of $\mathcal{O}(150)$ sums over integers. Nonetheless, its structure is  simple -- it is a large summation of a product of six functions $N(J,A,B,J',\rho)$ defined in the constrained loop identity in Eq.(\ref{fullloop}). Thus, instead of trying to truncate the whole 5--1 Pachner move, which is a daunting task, first we can simplify the problem by just studying the properties of a single constrained loop identity -- a much more tractable problem.

Let us then take a look at the constrained loop identity. Recall that in section \ref{sec:constrainedloolidentity} , after we integrated out the loop,  additional mixing terms appeared, which were not there in BF theory and which seem to be non-geometrical. We can analyze Eq.(\ref{fullloop}) to see how much these extra terms contribute to the amplitude. The mixing terms are characterized by their total spin $J'$ in Eq.(\ref{fullloop}). The larger $J'$ is, the higher order polynomials of complicated mixings appear. The mixed strands disappear only when $J'=0$. 

Let us look at the large spin behaviour first. As an illustration, the Fig. \ref{fig:logplots} presents logarithmic plots for the coefficient function $N(J,A,B,J',\rho)$ when $J,A,B$ are universally large (as an example, we set them to 100, but it can be any large enough number), while $J'$ picks small values $J'\in \{0,1,2\}$. 
\begin{figure} [h]
\centering
\begin{minipage}{7cm}
\includegraphics[width=1\textwidth]{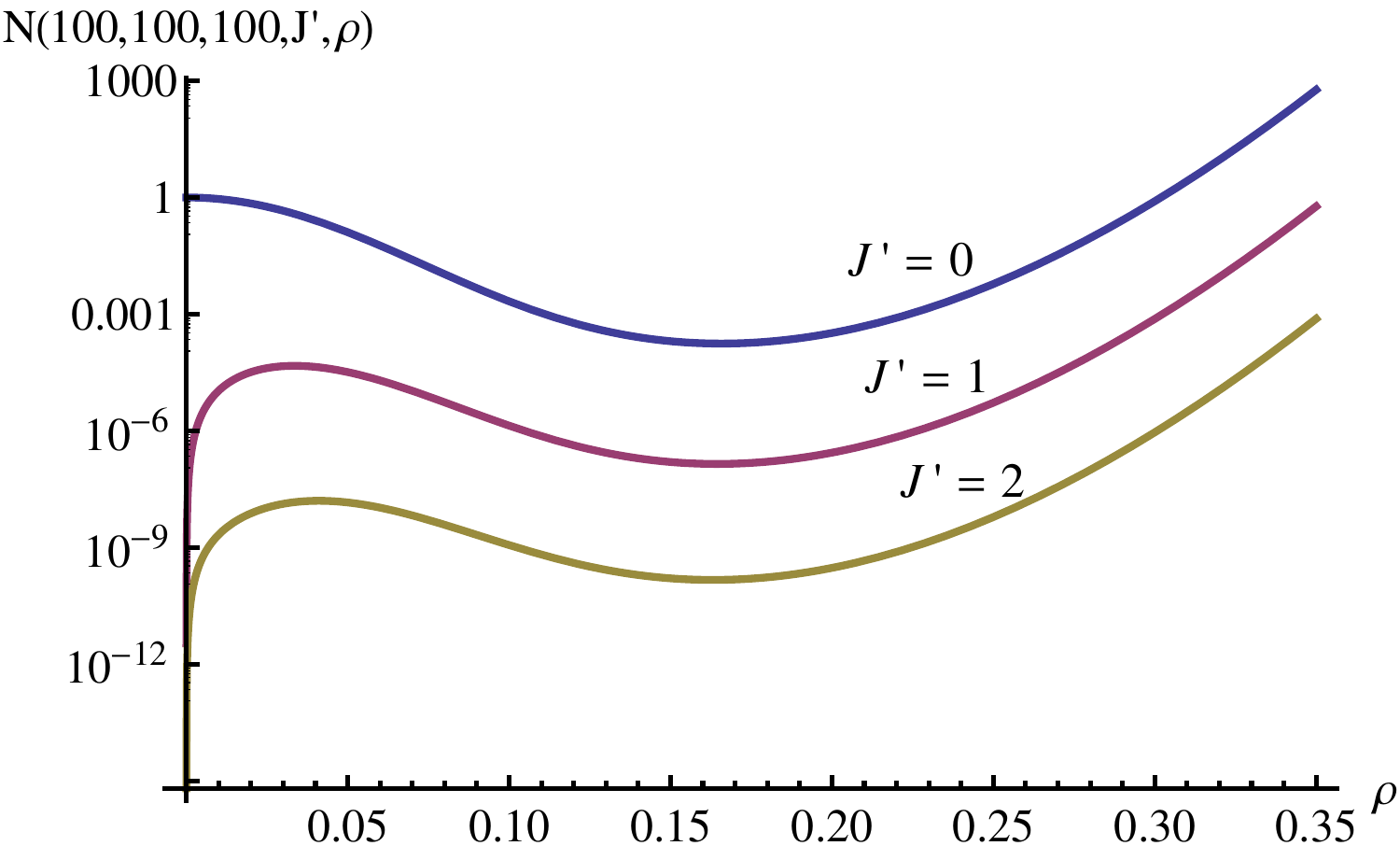}
\end{minipage}
\ \ \ \ \ \ \ \ \ \ 
\begin{minipage}{7cm}
\includegraphics[width=1\textwidth]{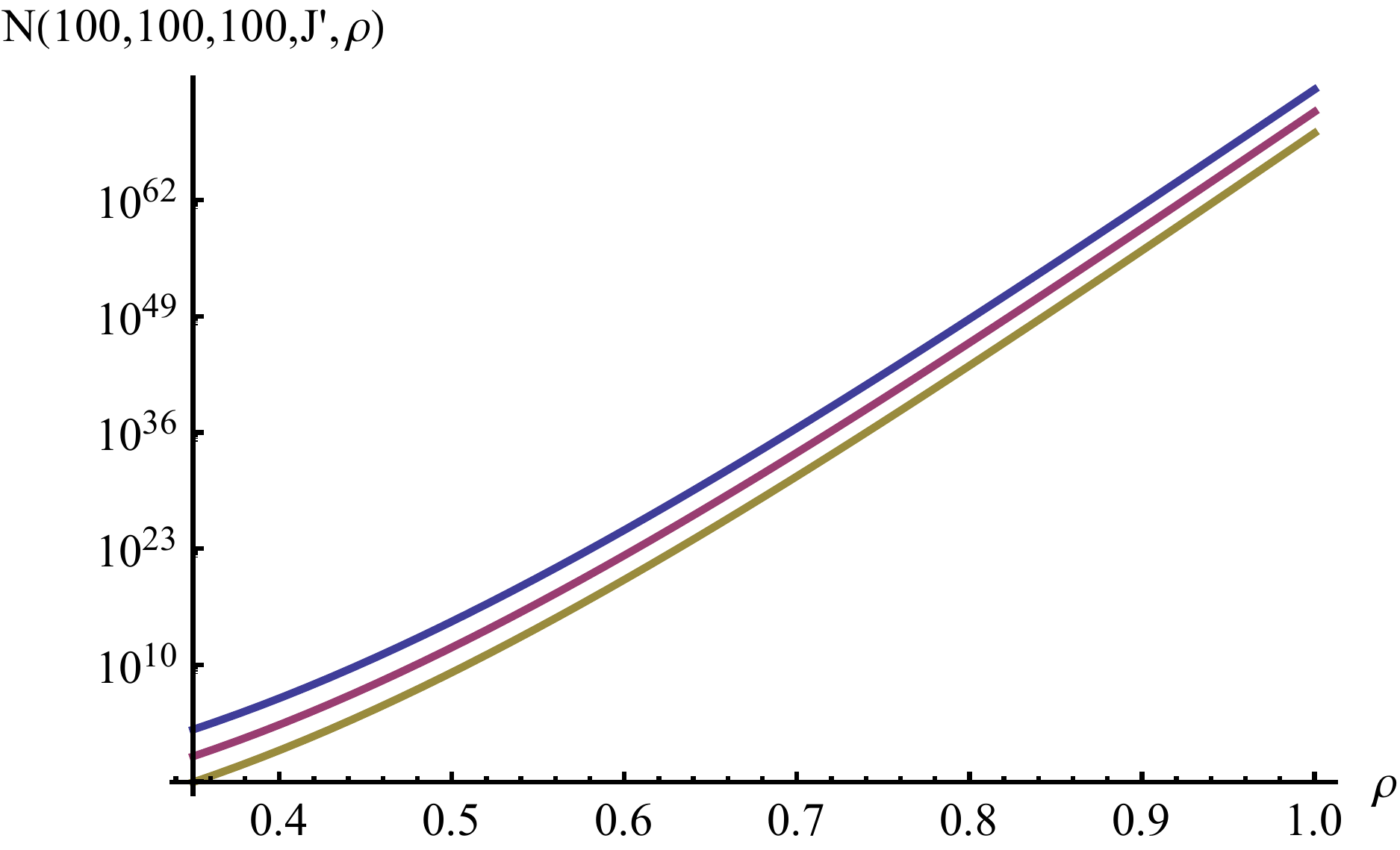}
\end{minipage}
\caption{Logarithmic plots for the coefficient  $N$ when $J=A=B=100$ and face weight scaling is  $\eta=1$. The  blue, red, yellow lines correspond to $J' = 0,1,2$ respectively. \label{fig:logplots}}
\end{figure}
We can observe  that for any $\rho \in [0,1]$, $N(J'=0)$ is at least more than $J$ times larger than the next order $N(J'=1)$, which is also approximately more than $J$ times larger than the next order $N(J'=2)$. Actually, we can plot the ratio between the coefficient of the first term $N(J'=0)$ and the sum of a few subleading coefficients  $\small{\sum}_{J'=1}^{10}  N(J')$ in Fig.\ref{ratio} as a function of $\rho$. When $\rho=0$, the expression converges to the behaviour of BF theory, $N(J,A,B,0,0) =1$ and $N(J,A,B,J',0) =0$ for any $J' \neq 0$, as desired. For $\rho\neq 0$, we get a smooth deformation of the BF result, with a similar behaviour, in the sense that the constrained loop identity is dominated by the $J'=0$ term. The smaller the $\rho$, the more dominating the unmixed term is. The same behavior holds when spins are large but not uniformally large -- the constrained loop identity is always dominated by the terms of $J'=0 $. 

\begin{figure} [h]
\centering
\begin{minipage}{7.5 cm}
\includegraphics[width=1\textwidth] {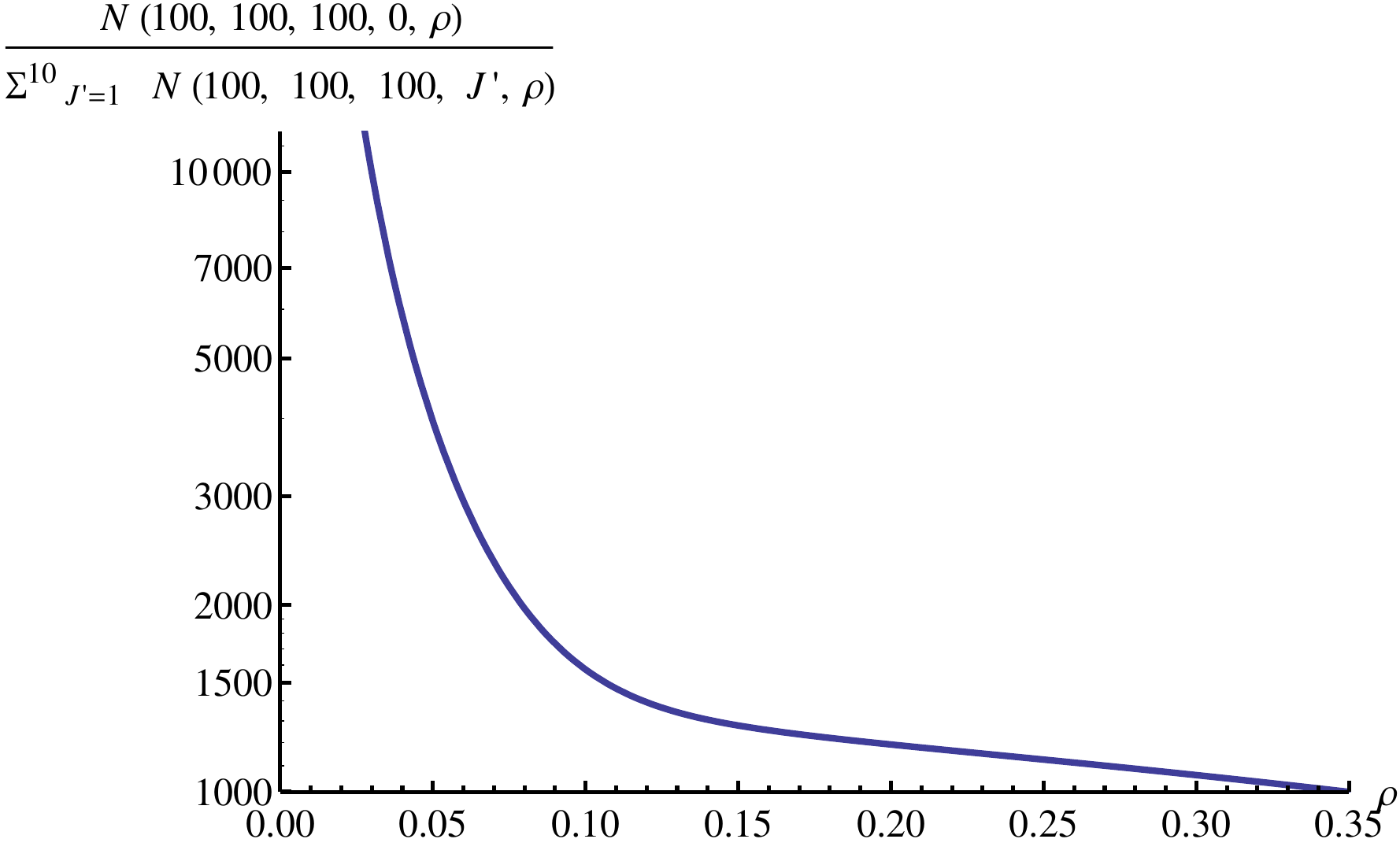}
\end{minipage}
\ \ \ \ \ \ \ \ \ \ 
\begin{minipage}{7.5 cm}
\includegraphics[width=1\textwidth]{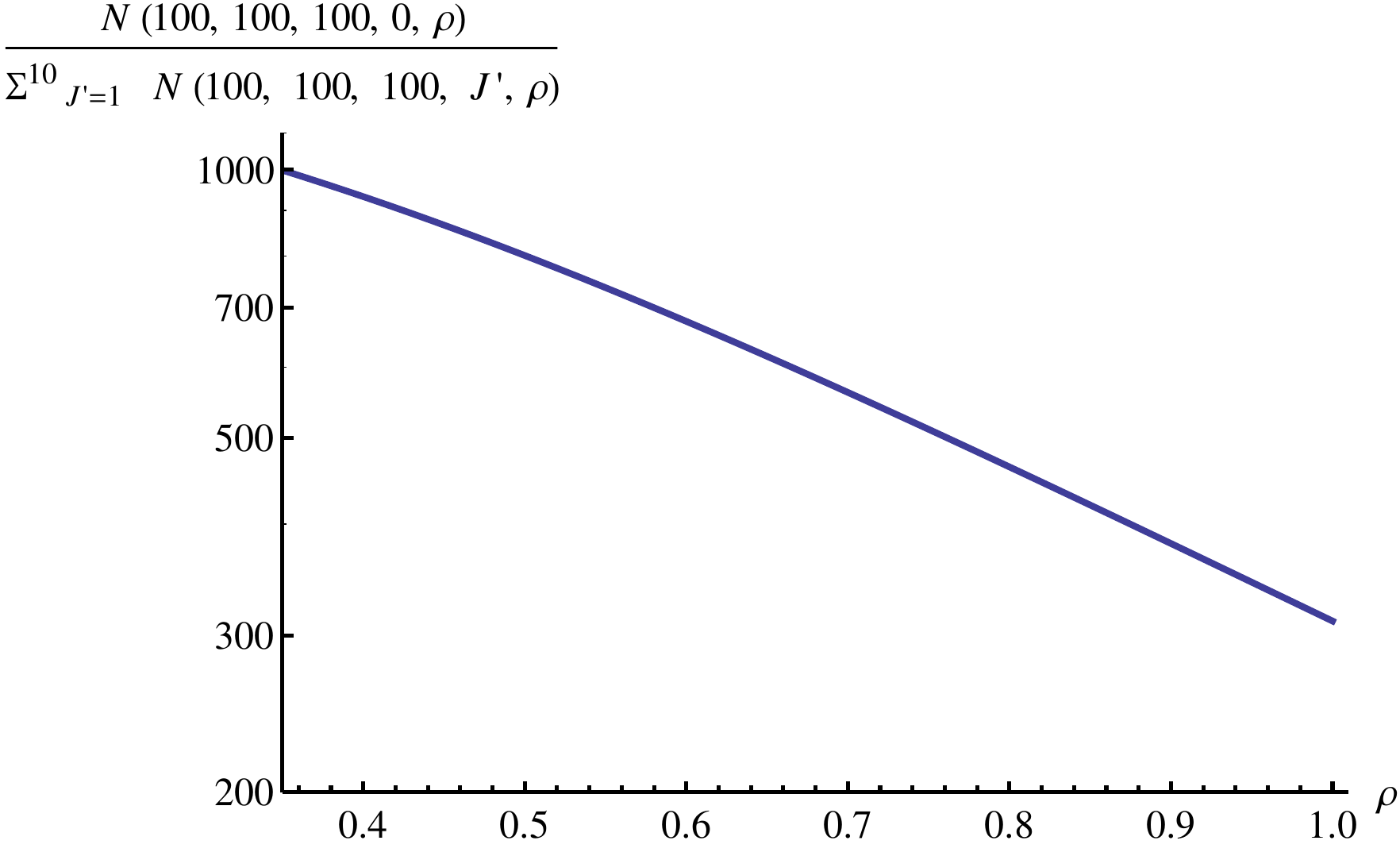}
\end{minipage}
\caption{Plots of the ratio between $N(J'=0)$ and the sum of subleading coefficients  $\small{\sum}_{J'=1}^{10}  N(J')$ with $J=A=B=100$ and face weight scaling $\eta=1$.  \label{ratio} }
\end{figure}

What about the case when the spins are not large? The plots in Fig.\ref{smallspin} illustrate that actually the $J'=0 $ terms are still dominating even when the spins $J, A, B$ are small. This means that the dominance of $J'=0$ terms surprisingly holds not only for large spins, but also for the small ones, even though the suppression is less pronounced compared with large spins cases. For small spins with the value of $\rho \rightarrow 1$, the dominance of $J'=0$ term is the least pronounced but still valid. 

\begin{figure} [h]
\centering
\begin{minipage}{6.8cm}
\includegraphics[width=1\textwidth]{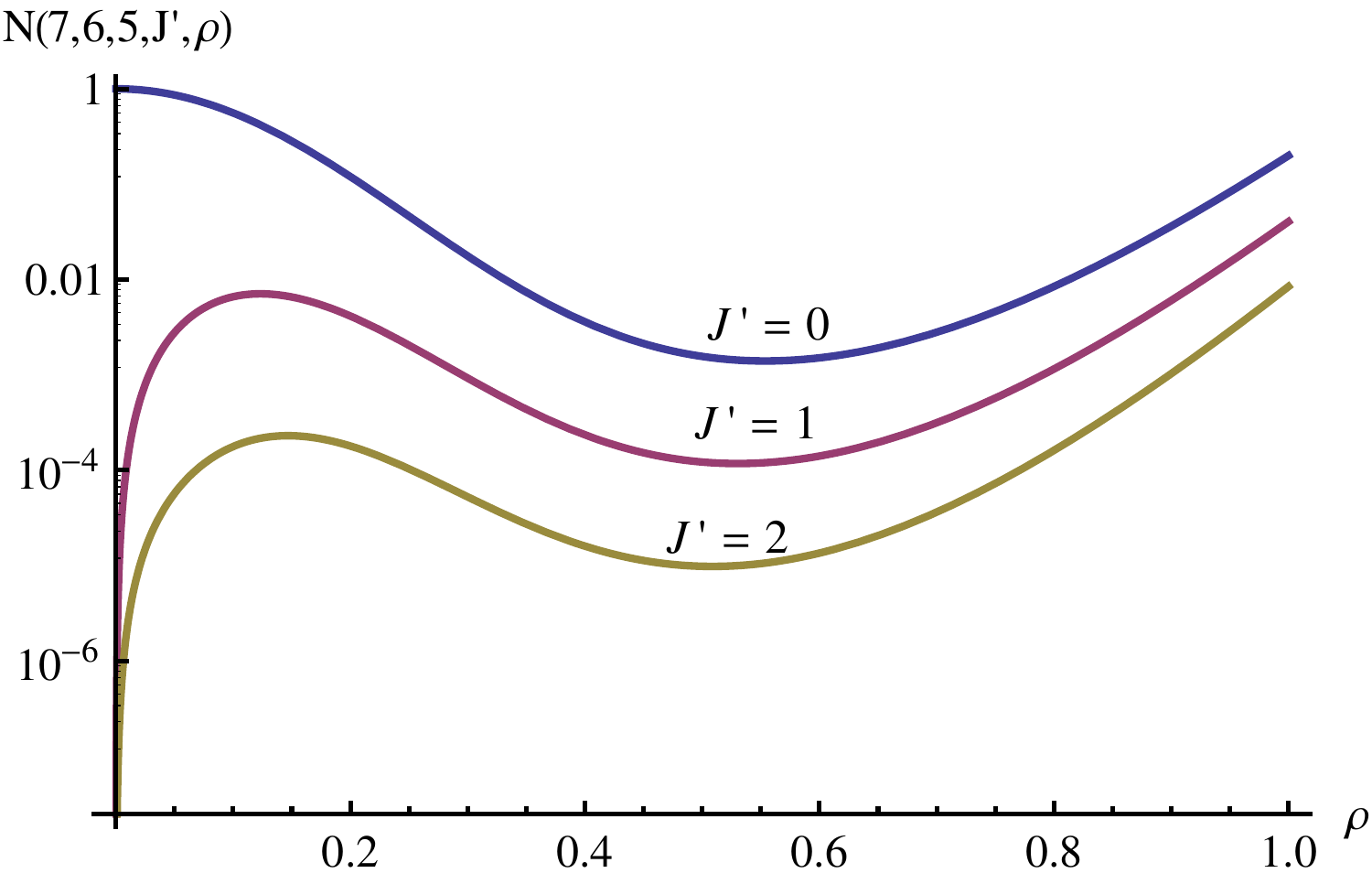}
\end{minipage}
\ \ \ \ \ \ \ \ \ \ 
\begin{minipage}{7cm}
\includegraphics[width=1\textwidth]{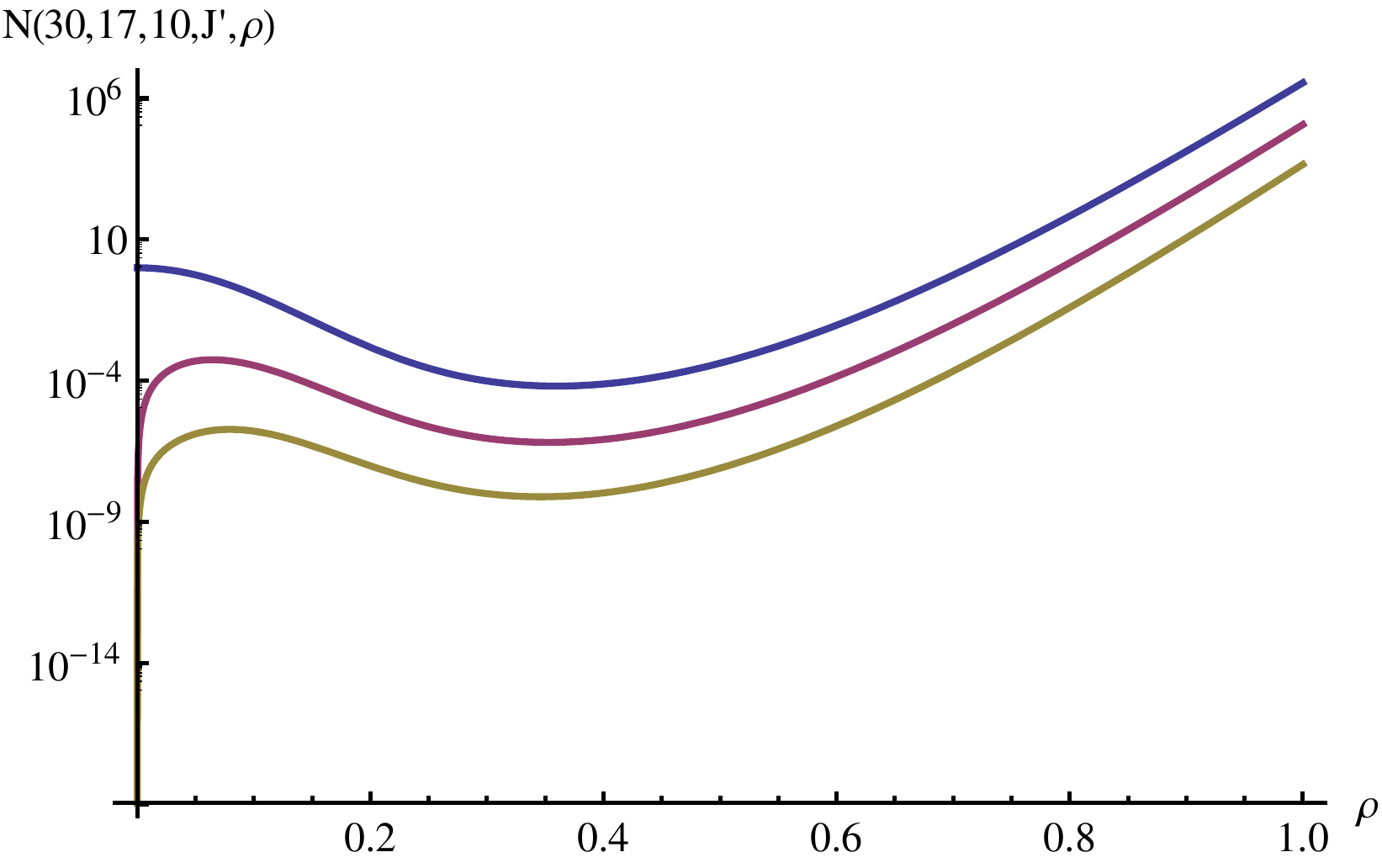}
\end{minipage}
\caption{Logrithmic plots for the coefficient $N$ when  $J=7, A=6, B=5$, and  $J=30, A=17, B=10$. The blue, red, yellow lines correspond to $J' = 0,1,2$ respectively.
\label{smallspin}}
\end{figure}

All of these results so far have been for the choice of face weight corresponding to $\eta = 1$. One could worry that perhaps the dominance of $J'=0$ fails for bigger face weights. We find however that the increasing of the face weight $\eta$ makes the effect stronger, as it is illustrated for small spins in the Fig.\ref{fwratio}. 

\begin{figure} [h]
\centering
\begin{minipage}{7.5 cm}
\includegraphics[width=1\textwidth]{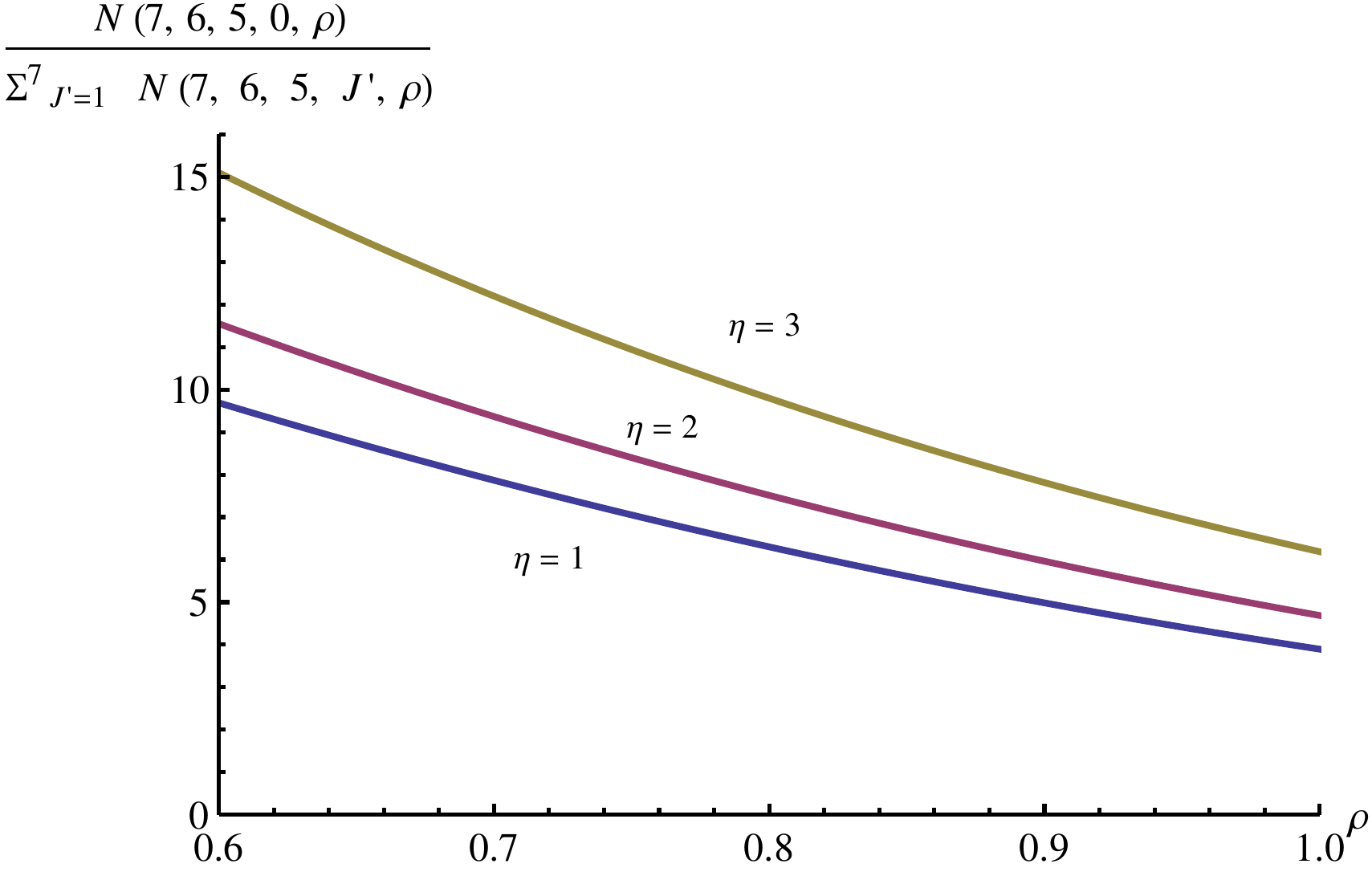}
\end{minipage}
\ \ \ \ \ \ \ \ \ \ 
\begin{minipage}{7.5 cm}
\includegraphics[width=1\textwidth]{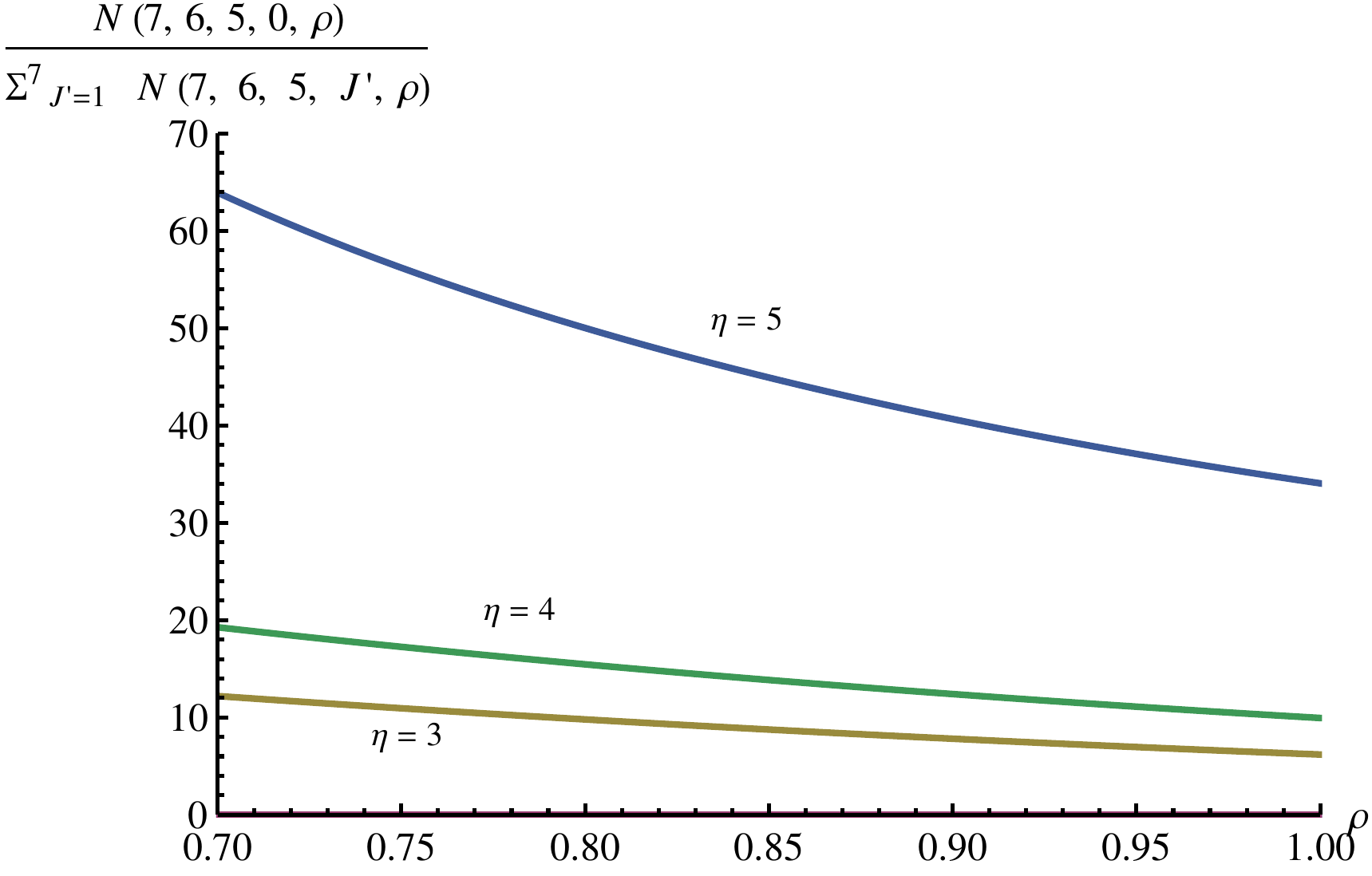}
\end{minipage}
\caption{Plots of the dependence of the ratio between $N(J'=0)$ and the sum of subleading coefficients with face weight $\eta$ when $J=7,A=6,B=5$ and $\rho$ is close to $1$.  \label{fwratio} }
\end{figure}

We thus  propose a natural truncation of keeping just the $J'=0$ terms and throwing away all the mixing terms $J' \neq 0$:
\be
\begin{split}
N\left(J, A,B,  J',\rho\right) &\approx N\left(J, A,B, 0, \rho\right)   \\
&=    \frac{(J\!+\!1)^{\eta-1} }{ (1+\rho^2)^{A+B +7J} }   F^2_\rho\left(J\right)  F^2_\rho\left( \frac{A+J}{2}\right) F^2_\rho\left(\frac{B+J}{2}\right).
\end{split}
\ee
This truncation dramatically simplified the expression of $N $, making all the mixing and non-geometrical terms disapear.  The truncation scheme can be graphically expressed as 
\be
\raisebox{-8mm}{\includegraphics[keepaspectratio = true, scale = 0.33] {loopidentity.pdf}}=\sum_{A,B,J}N(J'=0) \!\!\!\!\raisebox{-8mm}{\includegraphics[keepaspectratio = true, scale = 0.3] {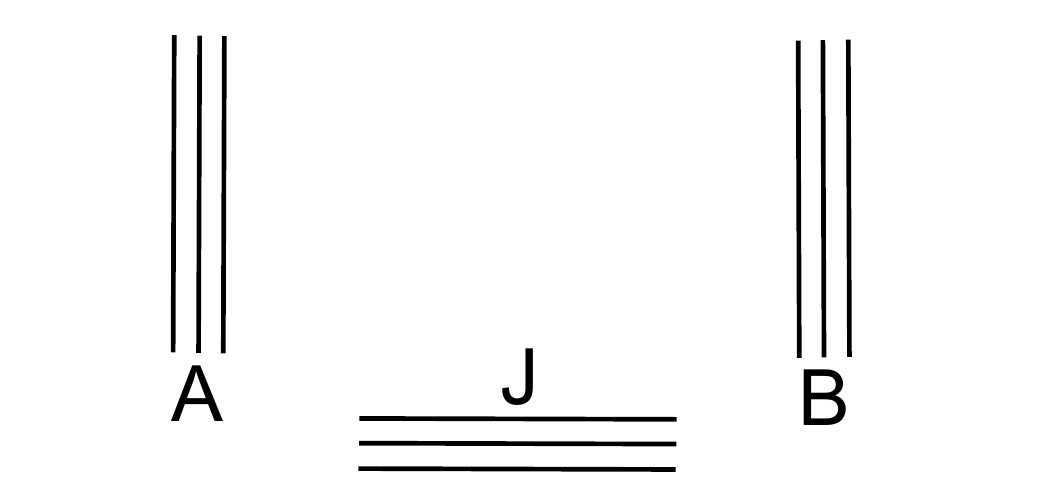}}\!\!\!\!\!\!+ \!\!\sum_{A,B,J,J'}N(J'\neq 0)\!\!\!\!\raisebox{-8mm}{\includegraphics[keepaspectratio = true, scale = 0.3] {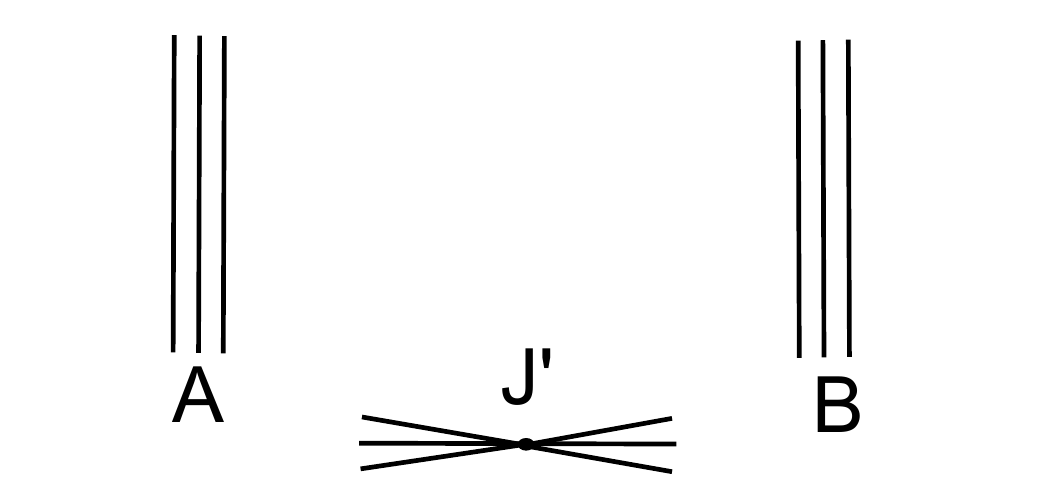}}\!\!\!\!\!\!.
\ee

Note that the left over strands will have to be integrated over in a calculation of a Spin Foam amplitude. And the contractions of these spinors  give additional factors of $1/(1+\rho^2)$, see the Appendix \ref{app_gauss}. These factors will lead to  additional suppression of amplitude for bigger $\rho$, making it more convergent. This however does not spoil the truncation. We will see the effect of this suppression in calculating the degree of divergence of Pachner moves in the coming section.

%%%%%%%%%%%%%%%%%%%%%%%%%%%%%%%%%%%%%%%%%%%%%%%%%%%
\subsection{Counting the degree of divergence}\label{sec:divergences}
%%%%%%%%%%%%%%%%%%%%%%%%%%%%%%%%%%%%%%%%%%%%%%%%%%%
Before we write out the truncated Pachner moves, let us first calculate how divergent the 5--1 move is as a function of face weight. The question of divergence is closely related to the one of symmetries. Indeed it is expected that in a physical model divergences of the partition function should be related to symmetries. This  has been only shown exactly in 3 dimensions \cite{Freidel:2002dw} so far.

 In a model describing 4d gravity we would expect the 5--1 move to be invariant up to a divergence coming from the freedom of translation of the added vertex inside the 4-simplex. Hence we would expect that for gravity the divergence should scale as $(length)^4$.  Of course at this stage this is a naive guess but it would be harder to argue for a diffeomorphism symmetry otherwise. In the case of translational symmetry this divergence is due to the possibility of moving the internal vertex outside the geometrical simplex.
 It can be tamed by incorporating orientation dependent factors as shown in
 3 dimensions  \cite{Christodoulou:2012af}.
 The Spin Foam models at our disposal  do not yet incorporate orientation dependence so it is unlikely that this phenomena can be used in our context.

The easiest way to count the degree of divergence is to set the external spins to 0, so that only the internal loops contribute. The calculation for the mixed 4 loops in the 5--1 move is rather involved, but thanks to the natural truncation  discussed in the previous section we can do the calculation. Let us however first try to estimate the degree of divergence arising from a single loop in the 4--2 move. It is important to stress here that in this case we do not need to do the truncation, as  setting the external spins to 0 corresponds to dropping all the products of spinors that contain the external ones in Eq. (\ref{eq:42move}), and hence naturally makes all the mixing terms drop out\footnote{This is another reason for seeing that the mixing terms might not be important.}.  This allows us to write the amplitude for the single loop in 4--2 move as
\be\label{eq:420}
\mathcal{A}^{4-2}_{\tau}(0)\!
=\!\!\int\!\!  \rd \mu_\rho(w) e^{\frac{\tau^{AB}_N\tau^{AD}_N\tau^{BD}_N\tau '}{(1+\rho^2)^2}\bra w|w\ket} \!=\! \frac{1}{\left(1-\frac{\tau^{AB}_N\tau^{AD}_N\tau^{BD}_N\tau '}{(1+\rho^2)^3}\right)^2}\!=\!\sum_j (2j+1)\!\left(\frac{\tau^{AB}_N\tau^{AD}_N\tau^{BD}_N\tau '}{(1+\rho^2)^3}\right)^{\!\!2j},
\ee
where, recall we have labeled the three loops, on which we applied the loop identity, by $\{AB, AD, BD\}$.    Using the homogeneity map defined in Eq. (\ref{eq:pachnerhom}), we can reintroduce the factors of face weight and the functions of $\rho$ from loop identities. Regularizing the expression by putting a cut-off of $\Lambda$ on spins, we get that a single loop in the 4--2 move is given by
\be
D_{4-2}(\Lambda, \rho, \eta)=\sum_{j=0}^\Lambda \frac{(2j+1)^{4\eta - 2}}{(1+\rho^2)^{24\times 2j}}  [{}_2F_1(-2j-1,-2j;2;\rho^4)]^{12}.
\ee
It is easy to see that, since ${}_2F_1(-2j-1,-2j;2;0)=1$, for $\rho=0$ and $\eta=1$ we recover the SU(2) BF theory's divergence of a delta function $\delta_{SU(2)}(\id)$. It may seem surprising that the exact result is this simple. For the purpose of analysing the divergence let us write $D_{4-2}(\Lambda, \rho, \eta)=\sum_{j=0}^\Lambda X_{4-2}(2j, \rho, \eta)$.
\begin{figure} [h]
\centering
\parbox{6cm}{
\includegraphics[width=0.4\textwidth]{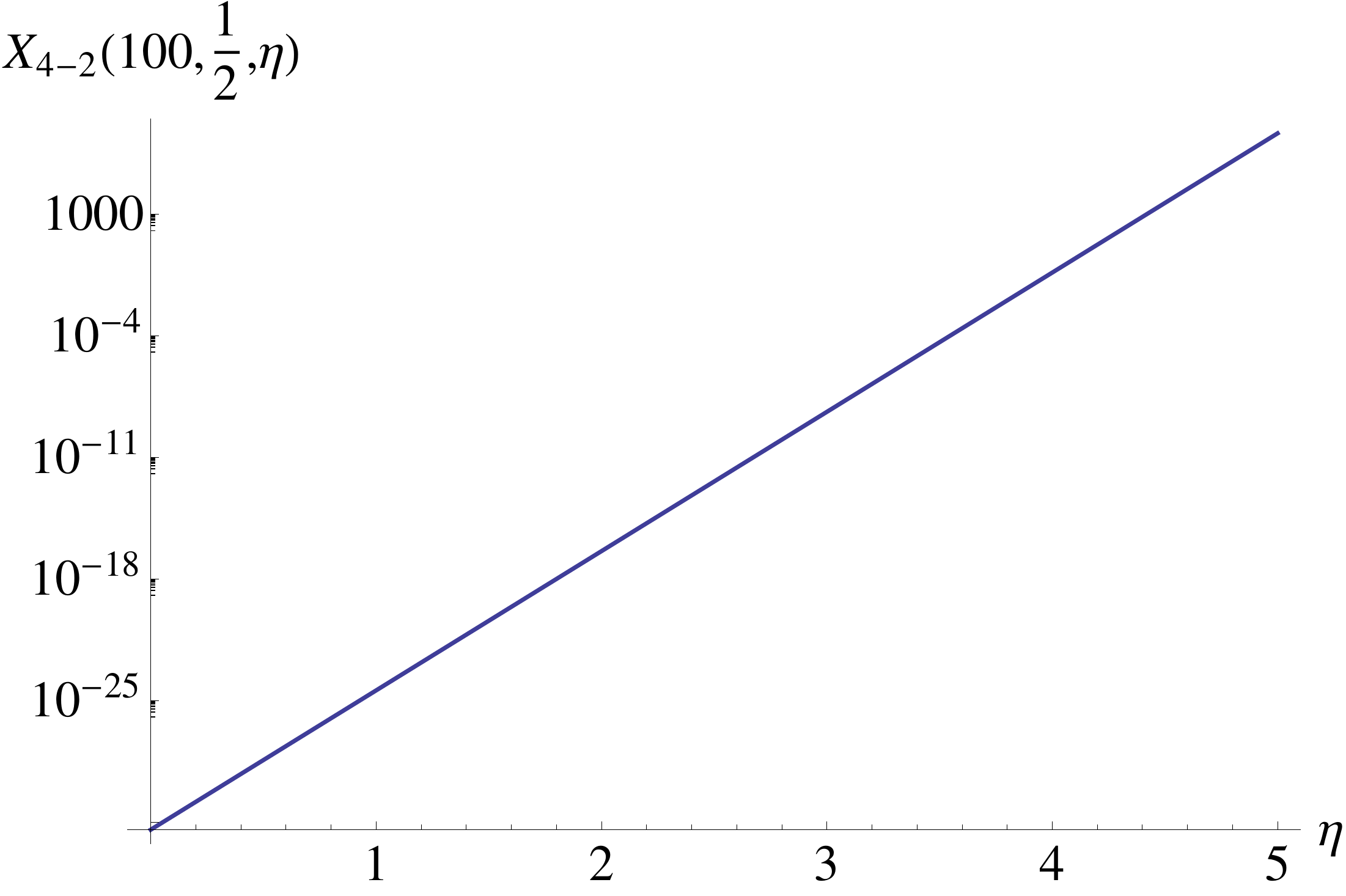}
}
\qquad
\begin{minipage}{6cm}
\includegraphics[width=1.1\textwidth]{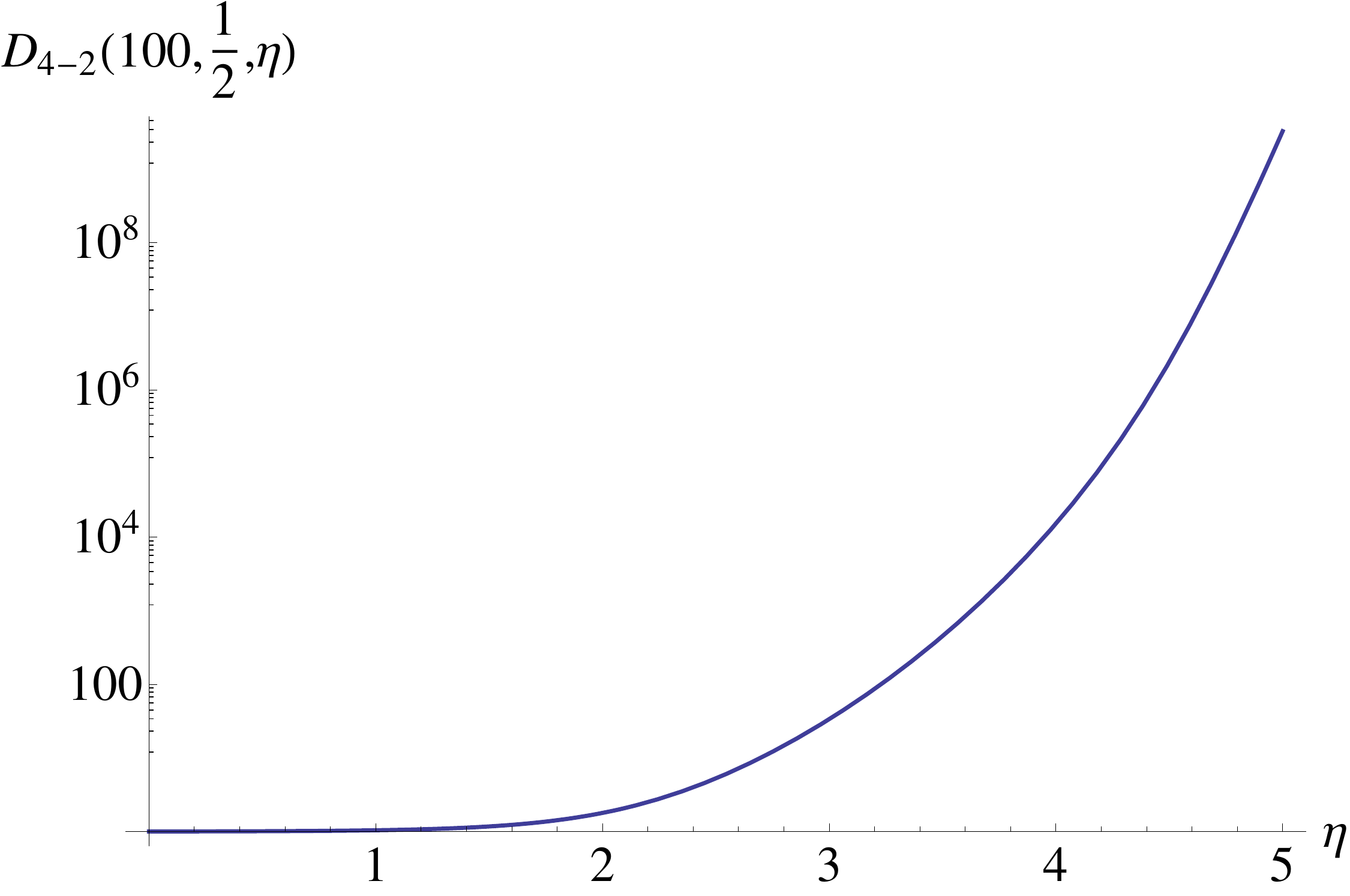}
\end{minipage}
\caption{$X_{4-2}$ has an obvious dependence on $\eta$ on a logarithmic plot. The sum diverges a lot faster with increasing $\eta$.\label{fig:X42eta}}
\end{figure}

Let us start with analysing the behaviour of $X_{4-2}$ and $D_{4-2}$ as a function of $\eta$. This is shown in Fig. \ref{fig:X42eta}. Quite obviously, at fixed spin, both $X_{4-2}$ and $D_{4-2}$ are diverging with increasing $\eta$. We get the opposite behaviour for increasing $\rho$ -- both  $X_{4-2}$ and $D_{4-2}$ are heavily suppressed for increasing $\rho$, as can be seen in Fig. \ref{fig:X42rho1}. This is the effect of the additional suppression by factors of $1/(1+\rho^2)$ that we mentioned in the previous section. We can thus expect intresting competition between $\rho$ and $\eta$ in concerning divergences.

\begin{figure} [h]
\centering
\parbox{6cm}{
\includegraphics[width=0.4\textwidth]{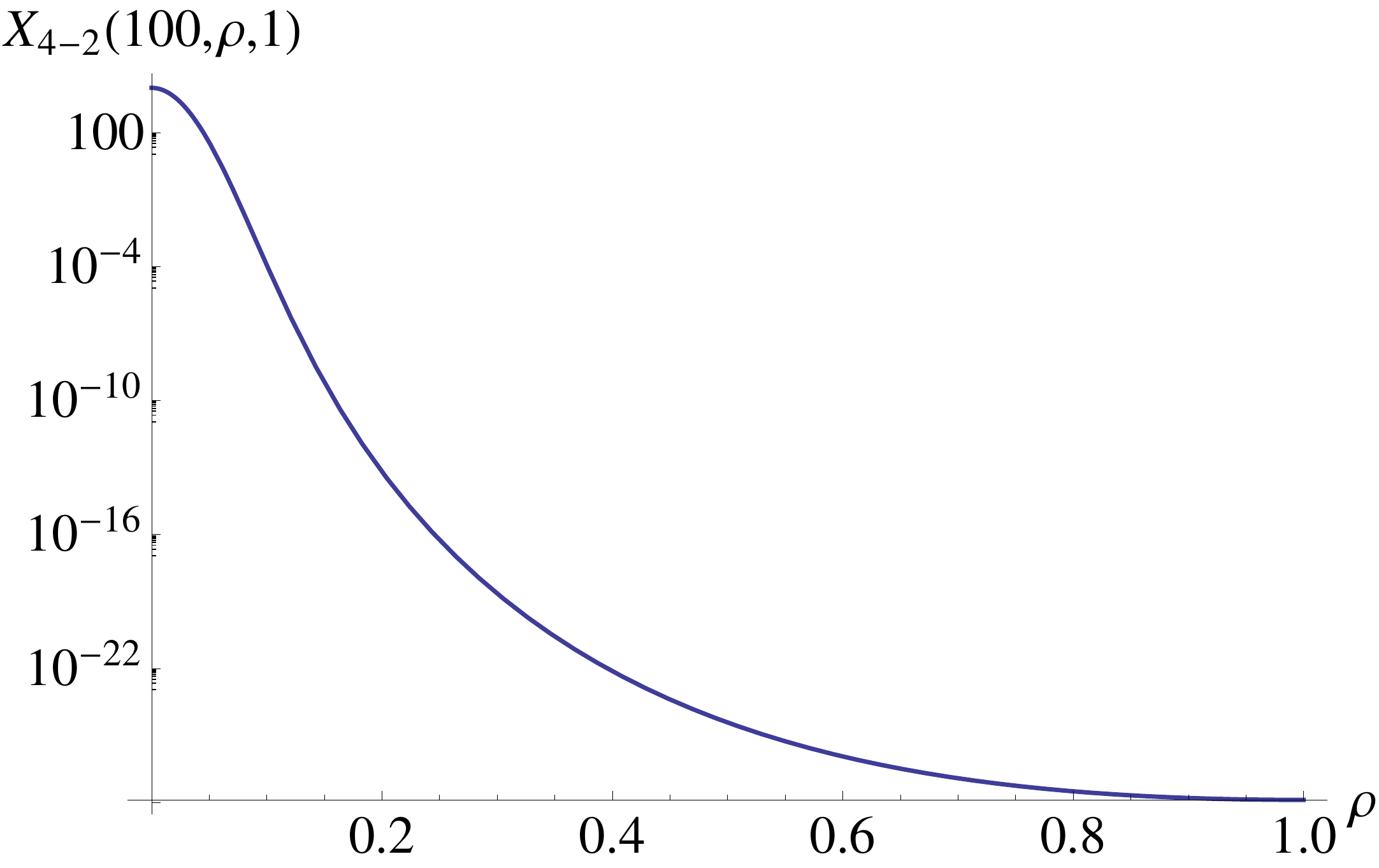}
}
\qquad
\begin{minipage}{6cm}
\includegraphics[width=1.1\textwidth]{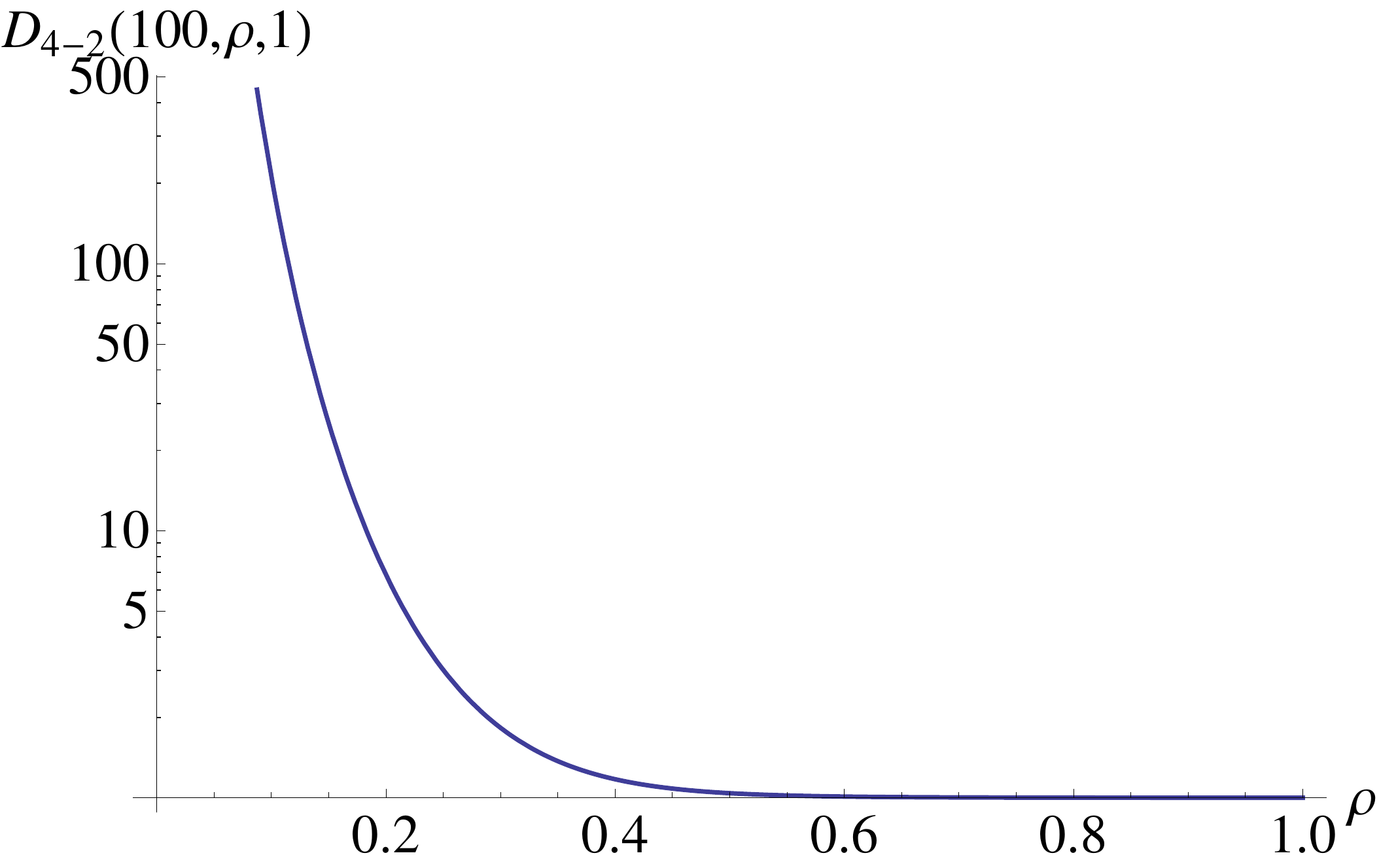}
\end{minipage}
\caption{$X_{4-2}$ gets suppressed with increasing $\rho$, with $\rho=0$ being the limit of SU(2) BF divergence. The sum is even more suppressed with increasing $\rho$.\label{fig:X42rho1}}
\end{figure}

Fixing $\rho$ to a specific value, we can analyze now the divergence of $X_{4-2}$ for different values of $\eta$, as a function of spin. We numerically find that $X_{4-2}$ is suppressed with increasing spin, but at $\eta = 5$ there is a transition to divergence, see Fig. \ref{fig:X42j}. This seems to be independent of the value of $\rho$, though the exact degree of divergence depends heavily on its value. 

\begin{figure} [h]
\centering
\parbox{5cm}{\vspace{-0.4cm}
\includegraphics[width=0.3\textwidth]{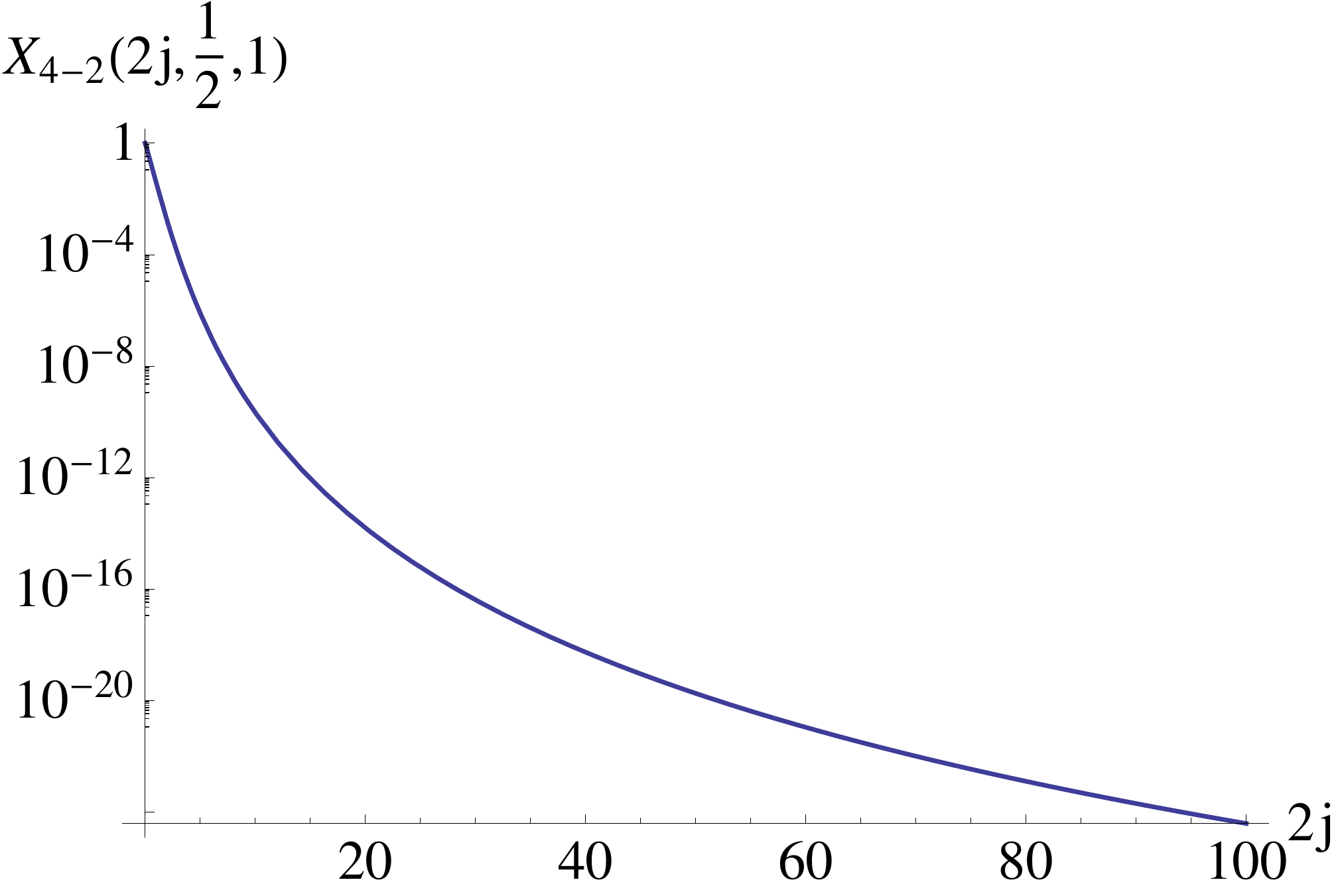}
}
\quad 
\begin{minipage}{5cm}
\includegraphics[width=1\textwidth]{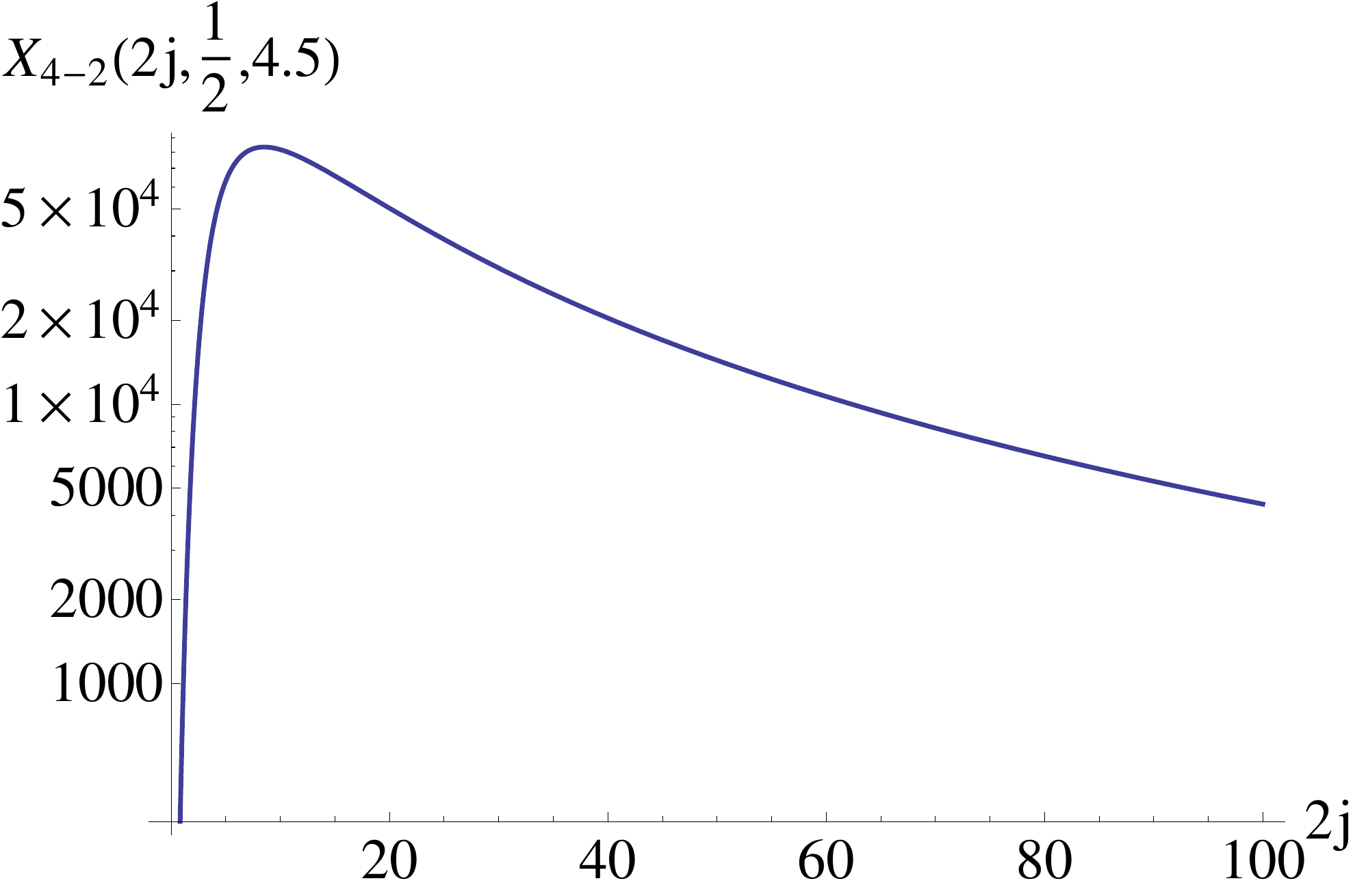}
\label{fig:X42j45}
\end{minipage}
\quad
\begin{minipage}{5cm}
\includegraphics[width=1\textwidth]{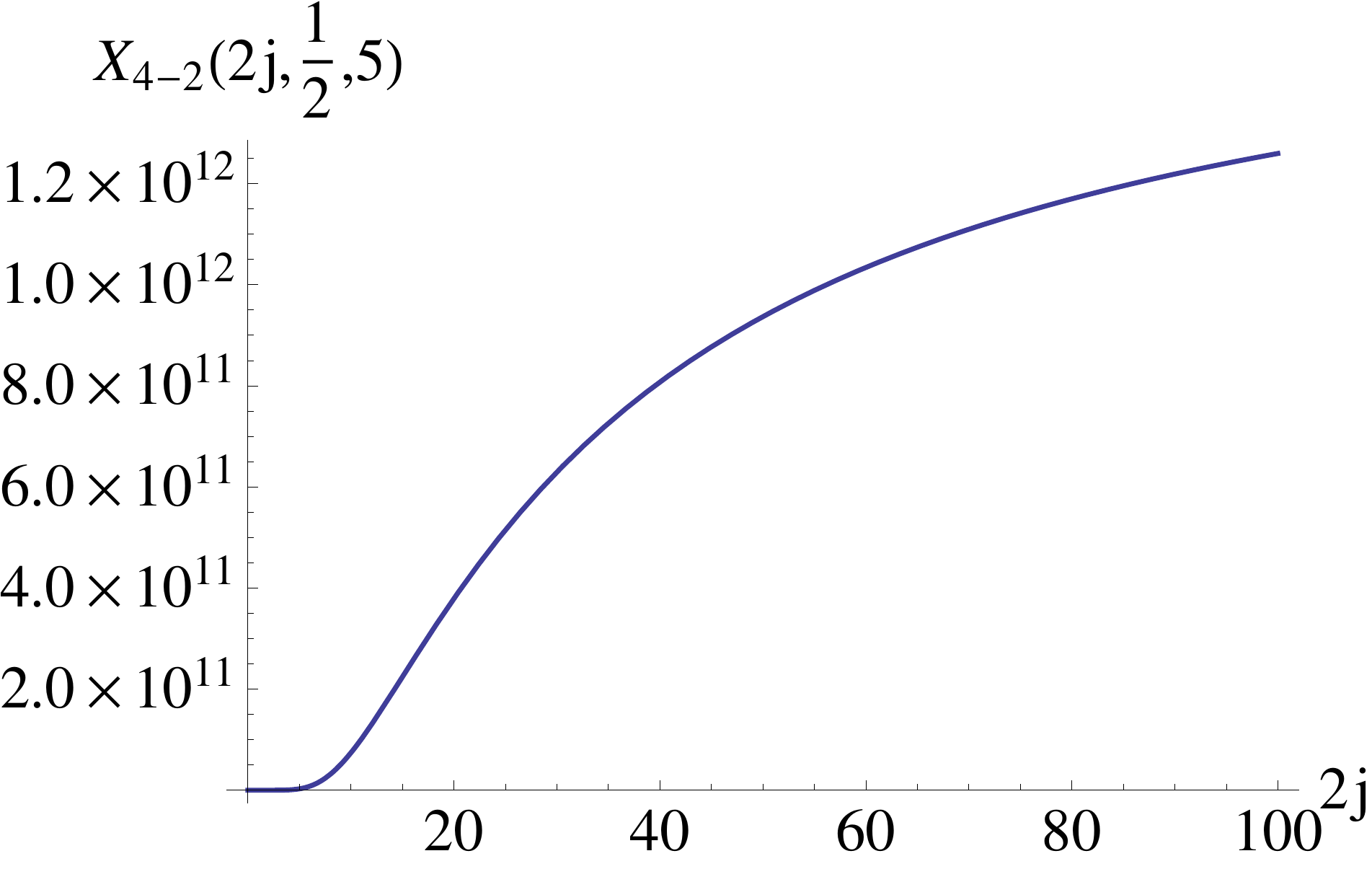}
\label{fig:X42j5}
\end{minipage}
\caption{For small values of $\eta$, $X_{4-2}$ gets suppressed with increasing spin. There seems to be a transition in the behaviour for $\eta=5$.\label{fig:X42j}}
\end{figure}

Let us now move onto the calculation of the degree of divergence for the 5--1 Pachner move. 
Truncating the loop identities in the 5--1 move allows us to perform the Gaussian integrals easily and write the four remaining loops as
\begin{equation}\label{eq:51trunc0}
\!\!\!\!\!\!\mathcal{A}^{5-1}_{\tau \ \text{truncated}}(0)\!
=\!\frac{1}{\left(1\!+\!\frac{\tau^{AC}_N\tau^{AD}_N\tau^{CD}_N\tau_y'}{(1+\rho^2)^3} \right)^{\!\!2}\!\left(1\!+\!\frac{\tau^{AB}_N\tau^{AD}_N\tau^{BD}_N\tau_b'}{(1+\rho^2)^3} \right)^{\!\!2}\!\left(1\!+\!\frac{\tau^{AB}_N\tau^{AC}_N\tau^{BC}_N\tau_g'}{(1+\rho^2)^3} \right)^{\!\!2}\!\left(1\!+\!\frac{\tau^{BC}_N\tau^{BD}_N\tau^{CD}_N\tau_r'}{(1+\rho^2)^3} \right)^{\!\!2}},
\end{equation}
where, similarly as in the 4--2 move, the six loops that we have integrated out were labeled by the set $\{AB, AC, AD, BC, BD, CD\}$ and the left over four loops are labeled by $\{y, g, b, r\}$.
Comparing this to the 4--2 move expression (\ref{eq:420}), we see that clearly we have 4 loops, that are not connected by any strands, but which are nonetheless coupled by sharing the $\tau$s, and hence functions of spin and $\rho$. We can now expand this in a power series for the fours spins $j_y, j_g, j_b, j_r$ and reintroduce the factors of the hypergeometric functions and face weights by using the homogeneity map from  Eq. (\ref{eq:pachnerhom}). 
Letting $a,b,c\in\{y, g, b, r\}$ we can write the full expression for the degree of divergence as
\begin{equation}
\begin{split}
\!\!\!\!\!\!\!\!D_{5-1}\!=\!\!\!\!\!\!\sum_{j_y, j_p, j_b, j_r}\! \frac{\prod_{a}(2j_a\!+\!1)^{\eta+1}}{(1\!+\!\rho^2)^{24\sum_a \!2j_a}} \! \left(\prod_{a<b} F_\rho(2j_a\!+\!2j_b)^2(2j_a\!+\!2j_b+1)^{\eta-1}  \right) \!\left(\prod_{a<b<c}F_\rho(2j_a\!+\!2j_b\!+\!2j_c)^2\right)\!,
\end{split}
\end{equation}
where, recall, we have previously defined $F_\rho(J) = {}_2F_1(-J-1,-J;2;\rho^4)$ for simplification.  Let us define $D_{5-1}=\sum_{\{j\}}X_{5-1}(j)$.

This general expression is rather long when expanded, but numerically it turns out that it is peaked around all the spins being equal.  Hence for all spins equal to $j$, we have a nice simplification
\begin{equation}
X_{5-1}\left(\{j\}\right)=\frac{(2j+1)^{4(\eta+1)}(4j+1)^{6(\eta-1)}}{(1+\rho^2)^{96\times 2j}}F_\rho(4j)^{12}F_\rho(6j)^{8}.
\end{equation}
Again, it is easy to see that for $\rho = 0$ and $\eta=1$ we recover the result of $\delta_{SU(2)}(\id)^4$ for the SU(2) BF theory. We can now analyze the behaviour of $X_{5-1}$ as a function of $\rho$. The results are qualitatively similar to those of 4--2 move, in the sense that the expression is suppressed for increasing $\rho$, see Fig. \ref{fig:X51rho}.

\begin{figure}[h]
	\centering
		\includegraphics[width=0.4\textwidth]{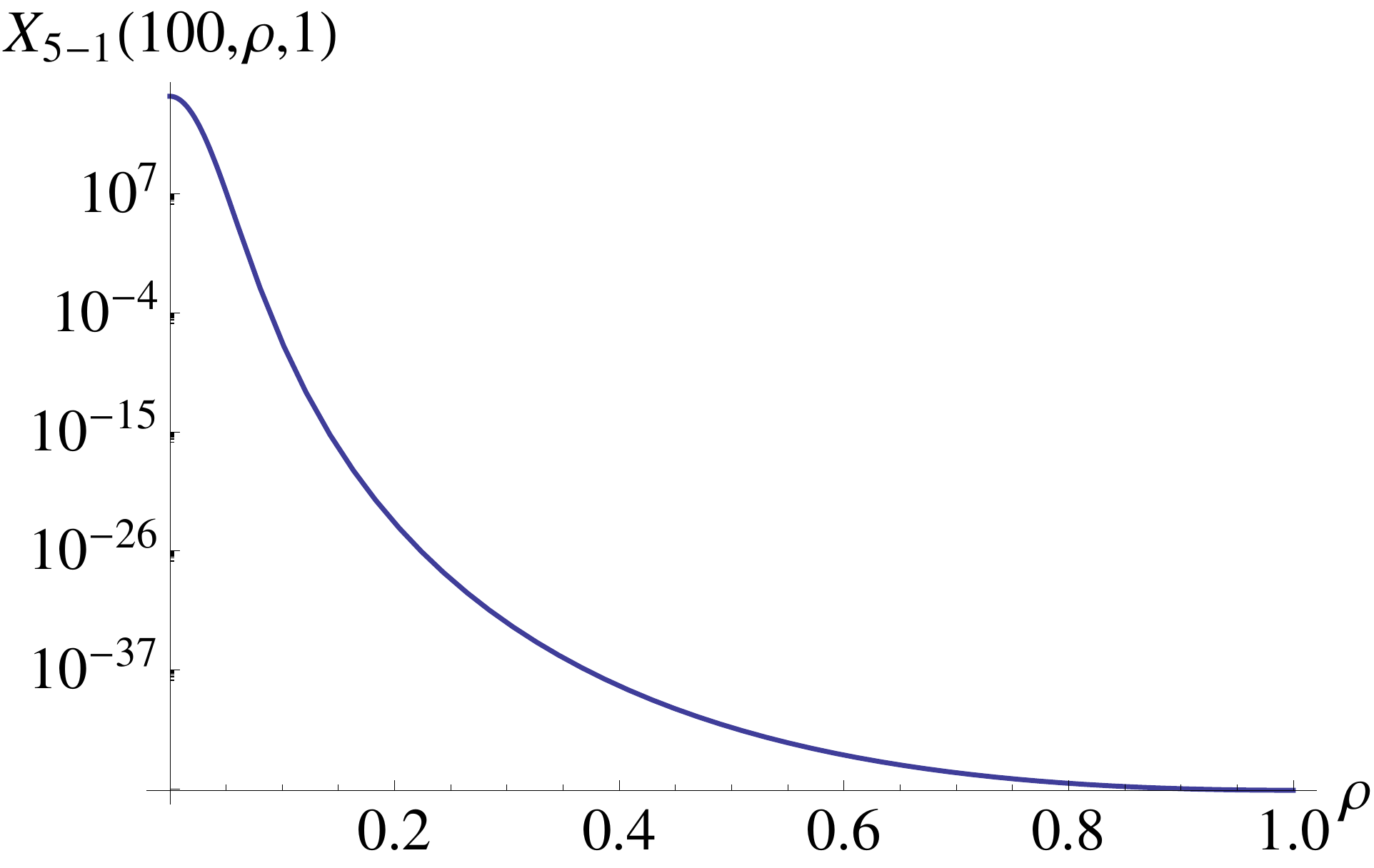}
	\caption{$X_{5-1}$ is suppressed with increasing $\rho$ for all values of $\eta$ and spin. This plot is evaluated at $2j=100$ and $\eta =1$.}
		\label{fig:X51rho}
\end{figure}

Quite obviously $X_{5-1}$ has similar behaviour to $X_{4-2}$ as a function of $\eta$, so we will not present plots for this. The interesting difference is in the transition from convergence to divergence of each $X_{5-1}$ term in the summation. The point of transition numerically seems to be around $\eta = 3.2$, see Fig. \ref{fig:X51j}.
\begin{figure} [h]
\centering
\parbox{5cm}{\vspace{-0.4cm}
\includegraphics[width=0.31\textwidth]{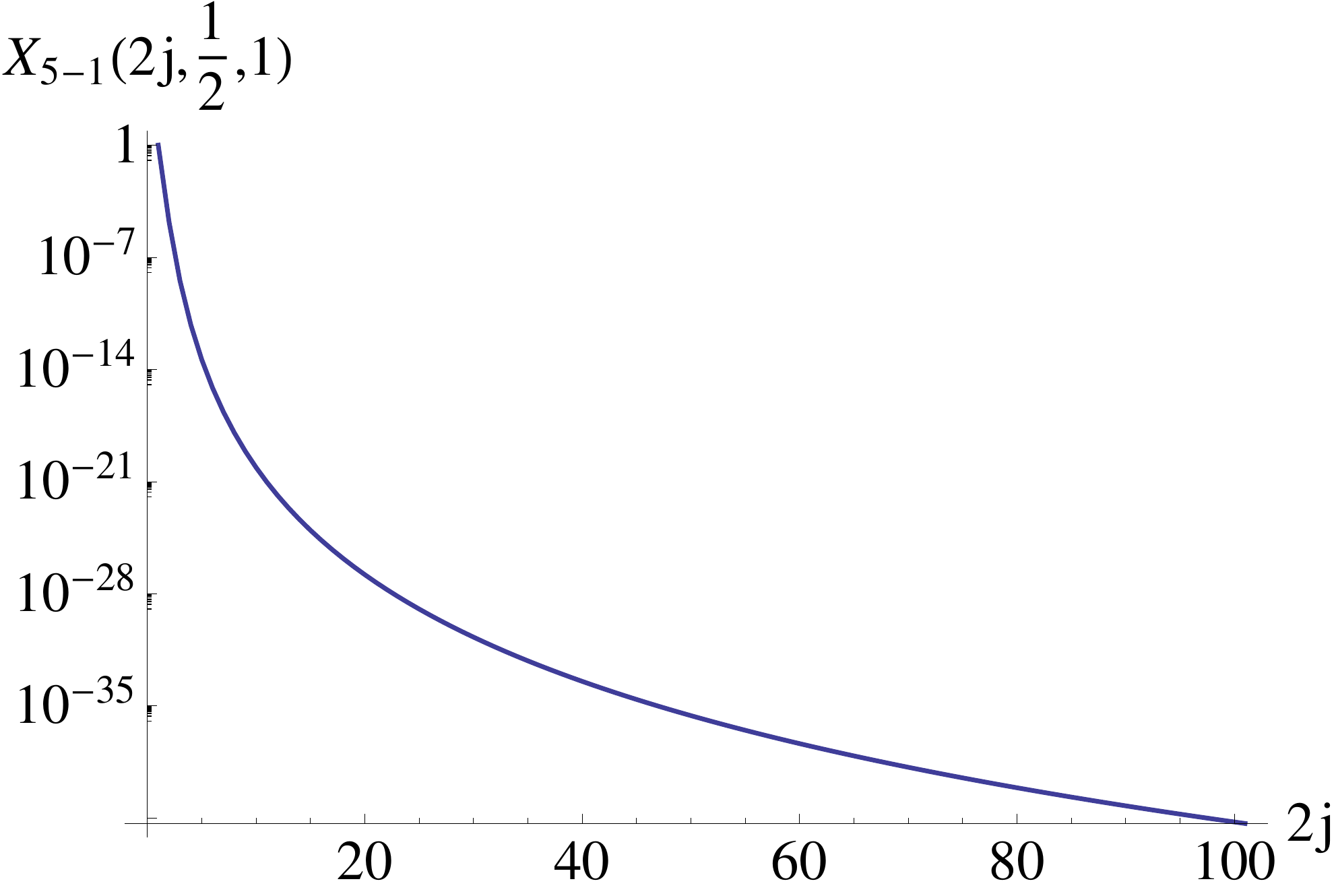}
}
\quad 
\begin{minipage}{5cm}
\includegraphics[width=1\textwidth]{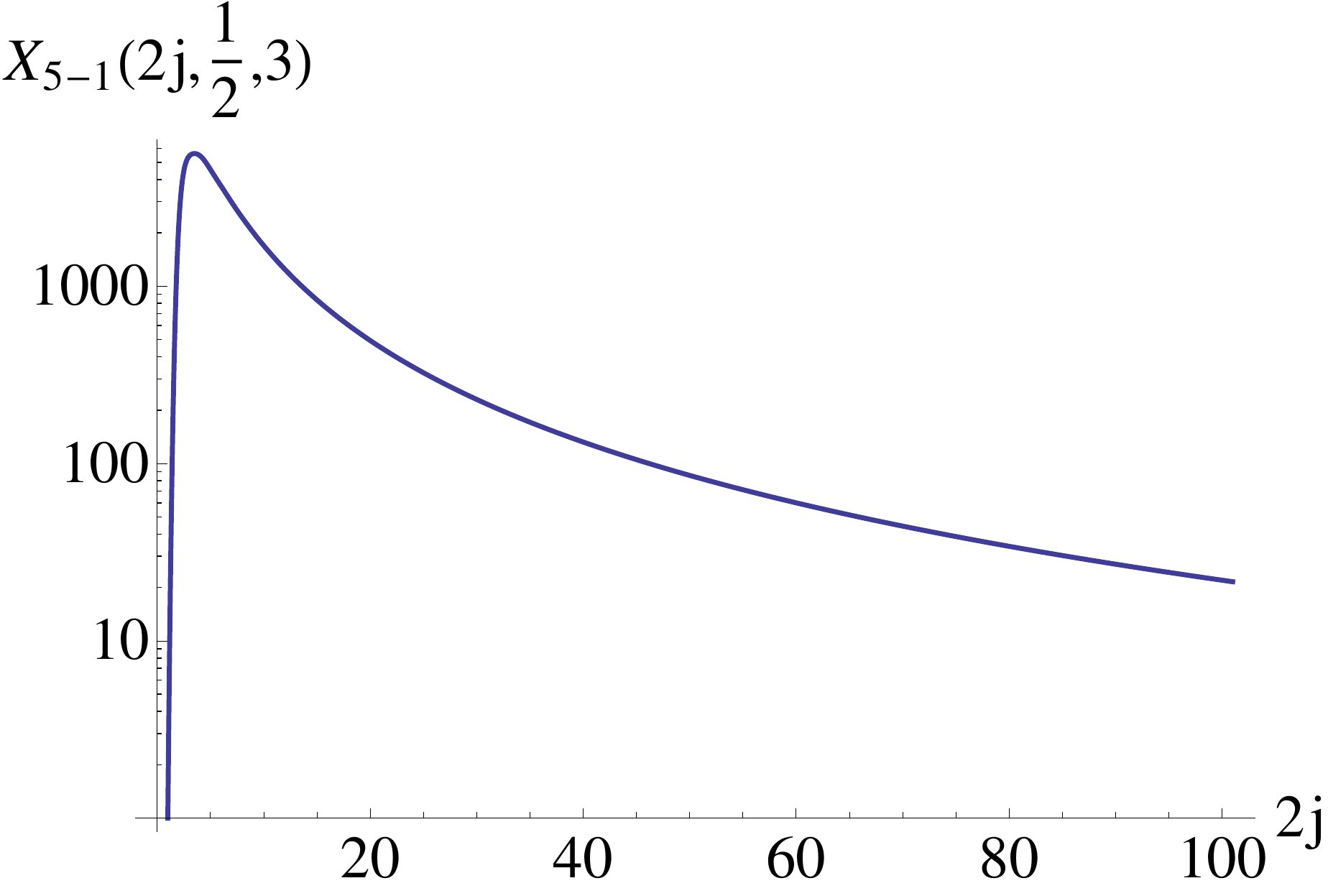}
\label{fig:X42j45}
\end{minipage}
\quad
\begin{minipage}{5cm}
\includegraphics[width=1\textwidth]{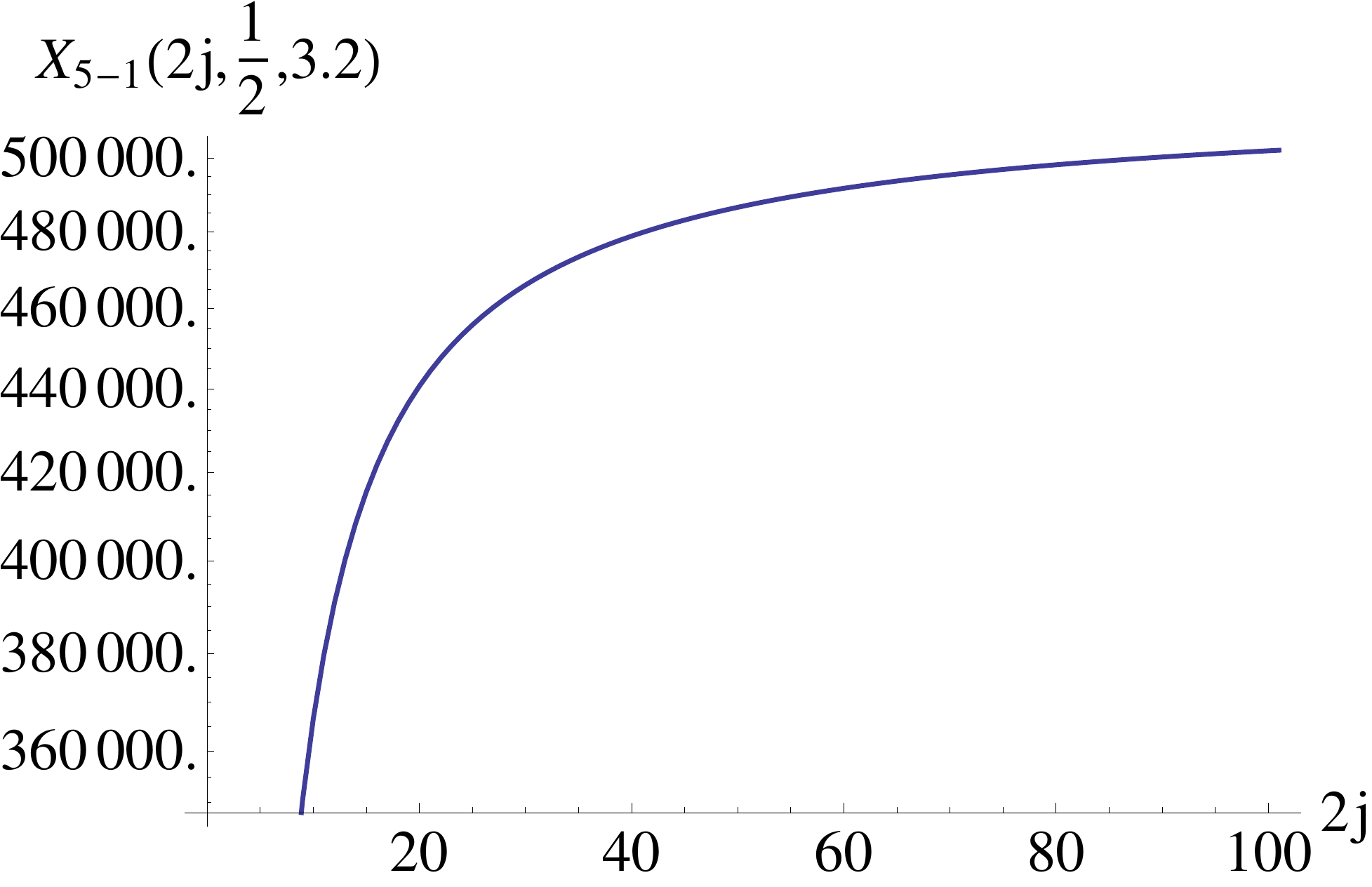}
\label{fig:X42j5}
\end{minipage}
\caption{For small values of $\eta$, $X_{5-1}$ gets suppressed with increasing spin. There seems to be a transition in the behaviour for $\eta=3.2$.\label{fig:X51j}}
\end{figure}
Note, that the expression of $D_{5-1}$ includes four summations, so it could become divergent even before $\eta=3.2$. We hence find that there is a range of the parameters $(\rho,\eta)$ for which 4--2 Pachner move is finite and 5--1 move is divergent.

%%%%%%%%%%%%%%%%%%%%%%%%%%%%%%%%%%%%%%%%%%%%%%%%%%%
\subsection{Truncated Pachner moves}
%%%%%%%%%%%%%%%%%%%%%%%%%%%%%%%%%%%%%%%%%%%%%%%%%%%
Now that we have already studied their divergence properties, we can write down the full expression for the 4-dimensional Pachner moves after truncation of the loop identities. As we will see, even though the loops in the moves are no longer mixed, there is still non-local coupling by spins.

Let us start with the simple observation, that the truncation does not change the non-invariance of the 3--3 Pachner move. The loop inside does decouple, but the truncation of the constrained loop identity does not change the fact that the hypergeometric functions of $\rho$ depend on different boundary spins in the two configurations. Thus even after the truncation, the 3--3 Pachner move is obviously not invariant, unless one considers very fine-tuned boundary spins.

Since the amplitude for the 4--2 and 5--1 moves look formally very similar, let us focus on the most interesting case of the 5--1 Pachner move. After truncation, the amplitude in Eq. (\ref{eq:51tau}) becomes

\be\label{eq:51truncated}
\begin{split}
\!\! & \mathcal{A}^{5-1}_{\tau \ \text{truncated}} (z^{\alpha f}_\gamma) = \int \prod_{\text{all}} d\mu_\rho(v,w)  \prod_{{\alpha } }  P_{\rho } (z^{\alpha f}_\gamma ; w^{\alpha f}_\gamma)  \cdot e^{  \sum_{  \sigma } {\tilde{\tau}_{E \sigma}  [\tilde{v}^{E \sigma}|\tilde{w}^{E \sigma} \ket}  +\sum_{ \mu \kappa i}\tau^{ \mu \nu}_N  \beta_i^{\mu\nu}[v^{\mu \nu}_i|w^{\nu \mu}_i \ket}\\
&=\int \prod_{\text{left over}} d\mu_\rho(v,w) \prod_{{\alpha } }  P_{\rho } (z^{\alpha f}_\gamma ; w^{\alpha f}_\gamma)  \cdot e^{  \sum_{  \sigma } {\tilde{\tau}_{E \sigma}  [\tilde{v}^{E \sigma}|\tilde{w}^{E \sigma} \ket} + \sum_{ \mu \nu }\tau^{ \mu \nu}_N  [v^{\mu \nu}_f|w^{\nu \mu}_f \ket} \mathcal{A}^{5-1}_{\tau \ \text{truncated}} (0),
\end{split}
\ee
where recall that indices run over $ \sigma \in \{ A, B, C, D\},\  \mu \nu \in \{AB, AC, AD, BD, BC, CD\}, i\in \{f,b,r,y,g \}, \ \alpha, \gamma\in\{A,B,C,D,E\}$. We have also defined the amplitude with boundary spins set to zero, $ \mathcal{A}^{5-1}_{\tau \ \text{truncated}} (0)$, in the previous section in Eq. (\ref{eq:51trunc0}) to be given by
\be\nonumber
\mathcal{A}^{5-1}_{\tau \ \text{truncated}}(0)\!
=\!\frac{1}{\left(1\!+\!\frac{\tau^{AC}_N\tau^{AD}_N\tau^{CD}_N\tau_y'}{(1+\rho^2)^3} \right)^{\!\!2}\!\left(1\!+\!\frac{\tau^{AB}_N\tau^{AD}_N\tau^{BD}_N\tau_b'}{(1+\rho^2)^3} \right)^{\!\!2}\!\left(1\!+\!\frac{\tau^{AB}_N\tau^{AC}_N\tau^{BC}_N\tau_g'}{(1+\rho^2)^3} \right)^{\!\!2}\!\left(1\!+\!\frac{\tau^{BC}_N\tau^{BD}_N\tau^{CD}_N\tau_r'}{(1+\rho^2)^3} \right)^{\!\!2}}.
\ee
It is  imperative now to notice that this does not trivially factorize, as we still have to apply the homogeneity map to obtain the final expression. The map defined in Eq. (\ref{eq:pachnerhom}) tells us that the $\tau_N$s are actually functions of the $\tilde{\tau}$s from the partially gauge-fixed propagators. The homogeneity map for the truncated 5--1 Pachner move is 
$ H_{5-1}[\mathcal{A}_{\tau \ \text{truncated}}^{5-1}] = \mathcal{A}^{5-1}_{\text{truncated}}$ and is given by 
\be\label{eq:pachnerhomtrunc}
\begin{split}
&H_{5-1}: \tau_N^{\mu\nu J}   \!\rightarrow \!\!F_ \rho (J)^2(J\!+\!1)^{\eta-1}\left(\frac{ \tilde{\tau}_{E\mu} \tilde{\tau}_{E\nu}}{(1+\rho^2)^5}\right)^{\!\!J},\qquad\tilde{\tau}_{E\sigma}^J \rightarrow  \frac{F_ \rho(J/2)^2}{(1+\rho^2)^{J}},\qquad\tau_i '^{2j}\rightarrow (2j+1)^\eta .
\end{split}
\ee
Before applying this homogeneity map, we need to first
 integrate out the  extra spinors on the internal strands -- because of the previously inserted propagators, each strand now has two spinors, instead of one. This is a simple Gaussian integration that we have performed many times before. This however requires us to contract the boundary propagators $P_{\rho }$ with functions of $\tau_N$ and $\tilde{\tau}$. 
 We  have to be careful now to perform these absorptions in a symmetric manner, which  allow us after applying the homogeneity map (\ref{eq:pachnerhomtrunc}) to define new boundary propagators $\tilde{P_{\rho } }$. 
 The amplitude (\ref{eq:51truncated}) is  then expected to become 
\be
 \mathcal{A}^{5-1}_{\text{truncated}}  (z^{\alpha f}_\gamma)=\tilde{D}_{5-1}\int \prod_{\gamma} d\mu_\rho(w_\gamma)  \prod_{{\alpha } }\tilde{P_{\rho }} (z^{\alpha f}_\gamma ; w^{\alpha f}_\gamma) .
\ee
This is the form of an amplitude for a 4-simplex with the modified propagators $\tilde{P_{\rho }}$ weighted by an overall, possibly divergent, factor $\tilde{D}_{5-1}$ which 
we expect, has  the same degree of divergence as the function $D_{5-1}$
 we studied in the previous section.
The exchange $P_{\rho }  \rightarrow \tilde{P_{\rho } }$ now is a proposal for a renormalization flow in the space of propagators (and perhaps the coefficient $\rho$ and face weight as well).  We leave the study of this flow and the behaviour of the other Pachner moves under it to future research.

%%%%%%%%%%%%%%%%%%%%%%%%%%%%%%%%%%%%%%%%%%%%%%%%%%%
\section{Discussion}
%%%%%%%%%%%%%%%%%%%%%%%%%%%%%%%%%%%%%%%%%%%%%%%%%%%
In this paper we have introduced a new Riemannian holomorphic Spin Foam model, with an alternative way of imposing holomorphic simplicity constraints. Instead of constraining the boundary spin network function, we impose the constraints on BF projectors. This model allows for more general graphs than the usual models built from vertex amplitudes. Surprisingly, it turns out to have the same asymptotics as the Dupuis-Livine model \cite{spinfoamfactory}, and hence the same semi-classical limit as the seminal EPRL-FK model \cite{Conrady:2008mk, Barrett:2009gg, Han:2011rf}. The holomorphic representation allows us to recast difficult integrals of SU(2) Wigner D-matrices into much simpler spinor Gaussian integrals. This allows for exact evaluation of spin network functions in BF theory \cite{Freidel:2012ji}. 
In our view the calculations done here realize the previously stated  hope that the spinor formalism  should allow the evaluation of expressions that 
 could not be done in the standard language of group integrals.
  In this work we have defined a \emph{homogeneity map}, which allowed us to extend these results to constrained models in the holomorphic representation, thus allowing us to evaluate all the Pachner moves in 3-- and 4--dimensions explicitly. 

In 3 dimensions, the results have been long known: the 3--2 move is invariant and 4--1 move is invariant up to a factor of an SU(2) delta function, which results from not fixing the gauge translation symmetry. It is however the first time, that Pachner moves have been calculated explicitly in a simplicity-constrained Spin Foam model of 4--dimensional Quantum Gravity. The calculation of all the Pachner moves followed from a simplicity-constrained version of so-called \emph{loop identity}. The crucial tool that allowed to find the constrained loop identity was the change from integrals over SU(2) group elements to integrals over spinors and the use of the homogeneity map. 
We found that 4d gravity Spin Foam models are not invariant under  3--3 move  unless very specific and symmetric boundary configurations are chosen. This is expected of a model for 4 dimensional gravity. A naive expectation, at least at the level of the classical action, is that the model should be invariant under the 4--2 and 5--1 Pachner moves. We found however this to be not the case for the exact evaluation. For both the 4--2 and 5--1 moves, there is an insertion of a  non-local combination of SU$(N)$ grasping  operators in the final coarse grained simplices, with a  mixing of strands leading to non-geometrical and non-local configuration. 

From the viewpoint of real-space renormalization group however, such non-local operators are expected to appear in each step of coarse graining, and have to be truncated to local ones in a controlled manner. Indeed, we have found that there exists a very accurate, natural and simple truncation scenario, which allows for a dramatic simplification of all of the 4--dimensional Pachner moves. Most crucially, the proposed truncation scheme removes the mixing of strands in the coarse-grained simplices, thus allowing them to remain geometrical, and hence making the 5--1 Pachner move structure preserving. After the truncation, the 4--2 and 5--1 Pachner moves are invariant up to a weight depending on the boundary spins. We should not expect an exact invariance, until a properly gauge-fixed model at a fixed point of renormalization flow corresponding to the continuum limit is found.

After introducing the truncation scheme, we have studied the divergence properties of the 4--2 and 5--1 moves. We find that the degree of divergence is crucially related to the two free  parameters of the model -- $\rho$ (related to the Immirzi parameter) and the power of the face weight $\eta$. More precisely, we find that whether the moves are convergent or divergent  depends on the choice of face weight $\eta$, but the exact scaling depends on $\rho$. The 4--2 Pachner moves transitions from convergence to divergence for $\eta \geq 5$. The 5--1 move becomes divergent much faster -- the transition is numerically evaluated  to be around $\eta\geq 3.2$. As such, there is a range of parameters, for which 4--2 move is convergent, while 5--1 is divergent. The popular naive choice of $(2j+1)^2$ is then not in that range. However, the exact value of the parameters, at which 5--1 move is divergent depends probably on the exact normalization of the model considered, and as such might change under renormalization flow.

An important direction for future is to check whether the truncation scheme we have proposed is robust. As such, we should study other models and find whether removing the mixing of strands in the loop identity is a good approximation. This will be studied for the DL model in \cite{DLPachner}.

The next crucial step that we leave to future study is of course to analyze the renormalization of this constrained propagator model. It would be natural to renormalize both the coupling constants $\rho$ and $\eta$, as well as the propagators $P_\rho(z_i;w_i)$. The evidence from renormalization of spin net models \cite{Dittrich:2013uqe, Dittrich:2013aia, Dittrich:2013voa} seems to point to a very rich structure of fixed points in the renormalization flow for constrained spin foam models. Finding a fixed point invariant under 4--2 and 5--1  Pachner moves would give us a model invariant under discretized diffeomorphisms, and give hope to finding a continuum limit of the theory.

\acknowledgments

We would like to thank Bianca Dittrich, Aldo Riello and Lee Smolin for helpful discussions and comments on this work. Research at Perimeter Institute is supported
by the Government of Canada through Industry Canada and by the Province of Ontario
through the Ministry of Research and Innovation. This work is also part of the research program of the Foundation for Fundamental Research on Matter (FOM), all which is part of the Netherlands Organization for Scientific Research (NWO).

\appendix

%%%%%%%%%%%%%%%%%%%%%%%%%%%%%%%%%%%%%%%%%%%%%%%%%%%
\section{Gaussian Integration Techniques}
%%%%%%%%%%%%%%%%%%%%%%%%%%%%%%%%%%%%%%%%%%%%%%%%%%%
\label{app_gauss}
In this appendix we compile a list of useful Gaussian spinor integrals. Consider first a standard Gaussian integral over the complex line $\C$
\be
\int_\C \frac{\rd^2\alpha}{\pi^2}e^{-|\alpha|^2 + \bar{x}\alpha + y\bar{\alpha}}=e^{\bar{x}y}.
\ee
This easily generalizes to a Gaussian integration over spinors on $\C^2$ with the Bargmann measure $\rd\mu (z) =  \pi^{-2}e^{-\bra z|z\ket}\rd^4z$ giving us the integral that allows us to contract strands on cable graphs
\be
\int_{\C^2}\rd\mu (z) e^{\bra x|z\ket + \bra z|y\ket} = e^{\bra x|y\ket}.
\ee
It is interesting to note that this contraction also works with anti-holomorphic spinors $|z]$, since $[x|y] =\bra y|x\ket$. We have thus
\be
\int_{\C^2}\rd\mu (z) e^{\bra x|z] + [z|y\ket} = e^{\bra x|y\ket}.
\ee

As with usual Gaussian integrations, we can calculate Gaussian spinor integrals of arbitrary polynomials. The special case worth mentioning is of course how delta function acts on holomorphic functions
\be
\int_{\C^2}\rd\mu (z) f(z) e^{\bra z|w\ket} = f(w).
\ee
Let us now consider the integrals that are crucial to the computations in the paper -- integrals with a matrix $A$. First consider the more familiar case of integrals over vectors of $n$ complex numbers
\be
\int_{\C^n}\prod_{i=1}^n \frac{\rd^2\alpha_i}{\pi^2} e^{-\sum_{i,j} \bar{\alpha}_i . A_{ij}\alpha_j}=\frac{1}{\det (A)}
\ee
This again trivially extends to the integrals over spinors. The expression useful for our paper is
\be
\int_{\C^{2n}}\prod_{i=1}^n\rd\mu (z_i)e^{\sum_{i,j}\bra z_i|A_{ij}|z_j\ket} = \frac{1}{\det (\one-A)}.
\ee
Recall that for the constrained model we had to change the measure of integration over spinors to $\rd\mu_\rho (z)=(1+\rho^2)^2\pi^{-2}e^{-(1+\rho^2)\bra z|z\ket}\rd^4 z$. It is easy to check that this is normalized properly as
\be
\int_{\C^2}\frac{(1+\rho^2)^2\rd^4 z}{\pi^{2}}e^{-(1+\rho^2)\bra z|z\ket}=1.
\ee
This change of measure leads to very simple changes to the above integrals. In particular, for a contraction we have
\be
\int_{\C^2}\rd\mu_\rho (z) e^{\bra x|z\ket + \bra z|y\ket} = e^{(1+\rho^2)^{-1}\bra x|y\ket}.
\ee
Hence for every contraction of spinors we pick up a factor of $1/(1+\rho^2)$. Thus for a loop on which we have three spinors we get the factor of $(1+\rho^2)^{-3}$ -- this appears all the time in loop identity and Pachner moves calculations.

%%%%%%%%%%%%%%%%%%%%%%%%%%%%%%%%%%%%%%%%%%%%%%%%%%%
\section{Mapping SU(2) to spinors }
%%%%%%%%%%%%%%%%%%%%%%%%%%%%%%%%%%%%%%%%%%%%%%%%%%%
\label{app_A}

\begin{lemma} \label{eqn_SU2_lemma}
Let $f \in L^2(SU(2))$ be homogeneous of degree $2J$, i.e. $f(\lambda g) = \lambda^{2J} f(g)$.  Given a spinor by $|z\ket$ define $g(z) = (|0\ket\bra 0| + |0][0|)g(z) = |0\ket\bra z| + |0][z|$ where $|0\ket = (1,0)^t$.  Then
\be
  \int_{\C^2} \rd\mu(z) f(g(z)) = \Gamma(J+2) \int_{\text{SU(2)}} \rd g \, f(g).
\ee
\end{lemma}
\begin{proof}
We can relate the inner product (\ref{barg_in_prod}) to the standard $L^2(\text{SU(2)})$ inner product by parameterizing the spinor as
\be
  |z\ket = \bpm r \cos\theta e^{i\phi} \\ r \sin\theta e^{i\psi} \epm ,
\ee
where $r \in (0,\infty)$, $\theta \in [0,\pi/2)$, $\phi \in [0,2\pi)$, $\psi \in [0,2\pi)$.  The Lebesgue measure in these coordinates is $\rd^4 z = r^3 \sin \theta \cos \theta \rd r \, \rd\phi \, \rd\theta \, \rd\psi$.  Now using the homogeneity property $f(g(z)) = r^{2J} f(\widetilde{g}(z))$ we have
\be
  \int_{\C^2} \rd\mu(z) f(g(z)) = \int_{0}^{\infty} \rd r \, r^{3+2j} e^{-r^2} \int_{0}^{\pi/2} \rd\theta \sin\theta \cos\theta \int_{0}^{2\pi} \rd\phi \int_{0}^{2\pi} \rd\psi f(\widetilde{g}(z)),
\ee
where $\widetilde{g}(z) \in \text{SU(2)}$.  Performing the intgral over $r$ we get
\be
  \int \rd r \, r^{3+2J} e^{-r^2} = \frac{1}{2} \Gamma(J+2)
\ee
and so
\be
  \int_{\C^2} \rd\mu(z) f(g(z)) = \Gamma(J+2) \int_{\text{SU(2)}} \rd g \, f(g),
\ee
where $\rd g$ is the normalized Haar measure on SU(2).  In our case $J$ is an integer so $\Gamma(J+2) = (J+1)!$.
\end{proof}

%%%%%%%%%%%%%%%%%%%%%%%%%%%%%%%%%%%%%%%%%%%%%%%%%%%
\section{Group averaging the SU(2) projector}
%%%%%%%%%%%%%%%%%%%%%%%%%%%%%%%%%%%%%%%%%%%%%%%%%%%
\label{app_thm_proj}
In this appendix we recall the calculation in  \cite{Freidel:2012ji} which shows that we can perform the integration over $g$ explicitly for the BF projector (\ref{proj}), which we prove in the following theorem.
\begin{theorem} \label{thm_proj_k}
The projector (\ref{proj}) can be expressed as a power series in the holomorphic spinor invariants as
\be  
  P(z_i;w_i) = \sum_{[k]} \frac{1}{(J+1)!} \prod_{i<j} \frac{([z_i|z_j\ket[w_i|w_j\ket)^{k_{ij}}}{k_{ij}!} .
\ee
where the sum is over sets of $n(n-1)/2$ non-negative integers $k_{ij}$ with $1\leq i<j\leq n$.
\end{theorem}

\begin{proof} 
Expanding (\ref{proj}) in a power series
\be
  \int_{\text{SU(2)}} \rd g e^{[z_i|g|w_i\ket} = \sum_{j_i} \int \rd g \prod_i \frac{[z_i|g|w_i\ket^{2j_i}}{(2j_i)!},
\ee
we see that each term in the sum is homogeneous of degree $2J = \sum_i (2j_i)$.  This fact allows us to use Lemma \ref{eqn_SU2_lemma} detailed in Appendix \ref{app_A} which says that we can replace the integral over SU(2) with a Gaussian integral paying a factor of $1/(J+1)!$ as in
\be
  (J+1)! \int \rd g \prod_i \frac{[z_i|g|w_i\ket^{2j_i}}{(2j_i)!} = \int \rd\mu(\alpha) \prod_i \frac{\left([z_i|0\ket\bra \alpha|w_i\ket + [z_i|0][\alpha|w_i\ket \right)^{2j_i}}{(2j_i)!}.
\ee
Now resum over $j_i$ to get
\be
  \sum_{j_i} (J+1)! \int \rd g \prod_i \frac{[z_i|g|w_i\ket^{2j_i}}{(2j_i)!} 
  = \int \rd\mu(\alpha) e^{\sum_i \left([z_i|0\ket\bra \alpha|w_i\ket + [z_i|0][\alpha|w_i\ket \right)} 
  = e^{\sum_{i,j} [z_i|0][0|z_j \ket[w_i|w_j\ket },
\ee
where we've performed the Gaussian integration in the second equality.  Using the antisymmetry $[w_i|w_j\ket = -[w_j|w_i\ket$ and recognizing the identity $1 = |0\ket\bra 0| + |0][0|$ in
\be
  \sum_{i,j} [z_i|0][ 0|z_j \ket[w_i|w_j\ket = \sum_{i<j} [z_i\left( |0\ket\bra 0| + |0][0| \right)|z_j \ket[w_i|w_j\ket =  \sum_{i<j} [z_j|z_i\ket[w_i|w_j\ket .
\ee
Finally we have
\be
  \sum_{j_i} (J+1)! \int \rd g \prod_i \frac{[z_i|g|w_i\ket^{2j_i}}{(2j_i)!} = e^{\sum_{i<j} [z_j|z_i\ket[w_i|w_j\ket} = \sum_{[k]} \prod_{i<j} \frac{([z_i|z_j\ket[w_i|w_j\ket)^{k_{ij}}}{k_{ij}!}
\ee
and since $J = \sum_{i<j} k_{ij}$ is just the total homogeneity of each term we can move the $(J+1)!$ to the RHS and complete the proof.
\end{proof}

%%%%%%%%%%%%%%%%%%%%%%%%%%%%%%%%%%%%%%%%%%%%%%%%%%%%
\section{Proof of Lemma (\ref{eqn_det_lemma})}\label{app_det_lemma}
%%%%%%%%%%%%%%%%%%%%%%%%%%%%%%%%%%%%%%%%%%%%%%%%%%%%
\begin{proof}
For a $2\times 2$ matrix $2\det M = \tr(M)^2-\tr(M^2)$.
If one consider   $M = \one - \sum_i C_i |A_i\ket[B_i|$,  we have
$$
  \tr(M^2) = 2 - 2\sum_i C_i [B_i|A_i\ket + \sum_{i,j} C_i C_j [B_i|A_j\ket[B_j|A_i\ket
$$
and
$$
  \tr(M)^2 = 4 - 4\sum_i C_i[B_i|A_i\ket + \sum_{i,j} C_i C_j [B_i|A_i\ket[B_j|A_j\ket ,
$$
therefore
$$
  2\det(M) = 2 - 2\sum_i C_i [B_i|A_i\ket + \sum_{i,j}  C_i C_j  \left([B_i|A_i\ket[B_j|A_j\ket - [B_i|A_j\ket[B_j|A_i\ket \right)
$$
and using $[A_i|B_i\ket[B_j|A_j\ket - [A_i|B_j\ket[B_i|A_j\ket = [A_i|A_j\ket[B_j|B_i\ket$ gives the result.
\end{proof}

%%%%%%%%%%%%%%%%%%%%%%%%%%%%%%%%%%%%%%%%%%%%%%%%%%%%
\section{Explicit calculation of the Constrained Loop Identity}\label{app_loopid}
%%%%%%%%%%%%%%%%%%%%%%%%%%%%%%%%%%%%%%%%%%%%%%%%%%%%
In this appendix we explicitly show how to calculate the constrained loop identity (\ref{fullloop}). Let us consider the loop composed of two pairs of partially gauge-fixed propagators $\one_\rho\circ\one_\rho$ and one pair of propagators $P_\rho\circ P_\rho$. To calculate this loop, let us use the homogenized propagators $\one_{\tilde{\tau}}\circ\one_{\tilde{\tau}}$ and $G_\tau\circ G_\tau$ instead and at the end of the calculation use the homogeneity maps (\ref{eq:homtrivial}) and (\ref{eq:homprop}), which we recall are given by 
\be\nonumber
\one_{\tilde{\tau}}\circ \one_{\tilde{\tau}} = e^{\tilde{\tau} \sum_i [\tilde{z}_i|\tilde{w}_i\ket} \quad \text{with} \quad \tilde{\tau}^J \rightarrow  \frac{F_\rho(J/2)^2}{(1+\rho^2)^{J}} \quad\text{for}\quad \one_{\tilde{\tau}}\circ \one_{\tilde{\tau}}\rightarrow \one_\rho\circ \one_\rho
\ee
for a pair of gauge-fixed propagators and by
\be\nonumber
G_\tau\circ G_\tau = e^{\tau \sum_{i<j}  [z_i|z_j\ket [w_i|w_j\ket} \quad \text{with} \quad \tau^J  \rightarrow \frac{F_\rho(J)^2}{(1+\rho^2)^{2J} (J+1)!} \quad \text{for} \quad G_\tau\circ G_\tau \rightarrow P_\rho\circ P_\rho.
\ee
for the pair of propagators $P_\rho$. We will also insert a face weight by tracking the homogeneity of spin in the loop by a factor of $\tau '$. The contractions of the spinors around the loop are as follows: $|w_4\ket = |\tilde{w}_4^2]$, $|\tilde{z}_4^2\ket = |\tilde{z}_4^1]$ and $|\tilde{w}_4^1\ket = |z_4]$. The cable diagram with all the labels is shown in Fig. \ref{fig:loopidentity2}.
\begin{figure}[h]
	\centering
		\includegraphics[width=0.8\textwidth]{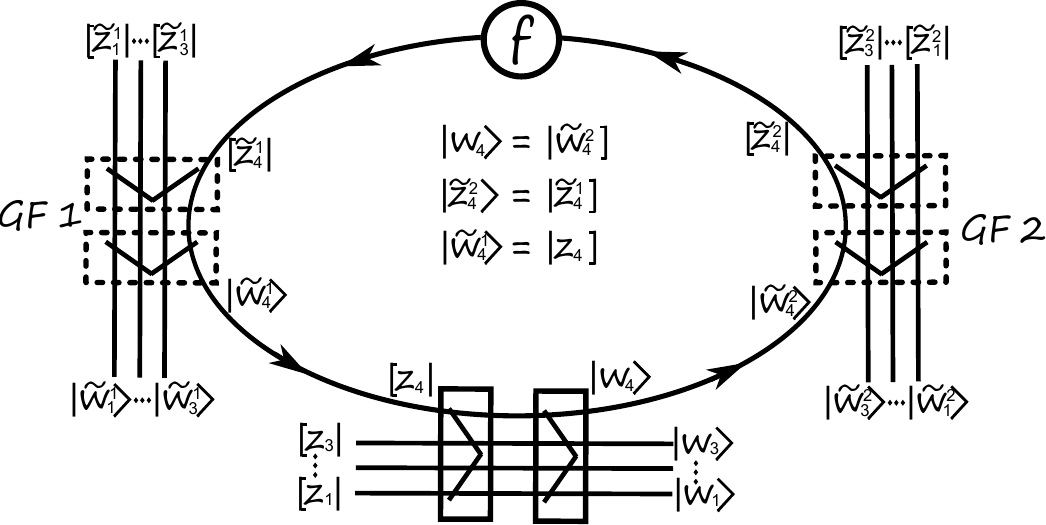}
	\caption{Cable diagram with all the labels for the constrained loop identity.}
		\label{fig:loopidentity2}
\end{figure}

 We can thus finally calculate the loop identity:
\be
\begin{split}
 &\int \!\!   \rd\mu_\rho (z_4, w_4, \tilde{z}^1_4) G_\tau^2(z_1,\ldots , \tau ' z_4; w_1,\ldots , \check{\tilde{w}}^2_4)  \one_{\tilde{\tau}_1}^2 (\tilde{z}^1_1,\ldots ,  \tilde{z}^1_4; \tilde{w}^1_1,\ldots , \check{z}_4) \one_{\tilde{\tau}_2}^2 (\tilde{z}^2_1,\ldots , \check{\tilde{z}}^1_4; \tilde{w}^2_1,\ldots , \tilde{w}^2_4)  \\ 
&= \frac{e^{\tau \sum_{i<j<4}[z_i|z_j\ket [w_i|w_j\ket+\sum_{i<4} \tilde{\tau}_1 [\tilde{z}^1_i|\tilde{w}^1_i\ket +  \tilde{\tau}_2 [\tilde{z}^2_i|\tilde{w}^2_i\ket} }{1 - \frac{\tau\tilde{\tau}_1\tilde{\tau}_2\tau '}{(1+\rho^2)^3} \sum_{i<4} [z_i|w_i\ket  + \left(\frac{\tau\tilde{\tau}_1\tilde{\tau}_2\tau '}{(1+\rho^2)^3}\right)^2 \sum_{i<j<4}  [z_i|z_j\ket [w_i|w_j\ket} \\ 
&= \exp\left(\tau \sum_{i<j<4}[z_i|z_j\ket [w_i|w_j\ket +\sum_{i<4} \tilde{\tau}_1 [\tilde{z}^1_i|\tilde{w}^1_i\ket +  \tilde{\tau}_2 [\tilde{z}^2_i|\tilde{w}^2_i\ket\right)\times \\
& \quad \times \sum_{N,M} \frac{(N+M)!}{N!M!}\left(\!\frac{\tau\tilde{\tau}_1\tilde{\tau}_2\tau '}{(1+\rho^2)^3} \right)^{ N+2M} \left( \sum_{i<4} [z_i|w_i\ket\right)^{N} \left(-\sum_{i<j<4}[z_i|z_j\ket [w_i|w_j\ket\right)^{M}.\nonumber
\end{split}
\ee
The factor of $1/(1+\rho^2)^3$ arises from the three spinor integrations. Compared to the toy loop, the result is thus an exchange of $\frac{\tau\tau '}{1+\rho^2} \rightarrow \frac{\tau\tilde{\tau}_1\tilde{\tau}_2\tau '}{(1+\rho^2)^3}$ and the addition of the trivial propagation of the gauge-fixed strands.
Before we can use the homogeneity maps  we have to expand the exponentials in a power series. Doing this we arrive at
\be
\begin{split}
& \sum_{A,B,C,M,N}\frac{(-1)^M(N+M)!}{A!B!C!M!N!}\left(\frac{\tau\tilde{\tau}_1\tilde{\tau}_2\tau '}{(1+\rho^2)^3} \!\right)^{ N+2M}\tilde{\tau}_1^A\tilde{\tau}_2^B\tau^C \times \\
& \times\left(\sum_{i<4} [\tilde{z}^1_i|\tilde{w}^1_i\ket \right)^{ A}\left(\sum_{i<4} [\tilde{z}^2_i|\tilde{w}^2_i\ket \right)^{ B} \left( \sum_{i<4} [z_i|w_i\ket\right)^{N} \left(\sum_{i<j<4}[z_i|z_j\ket [w_i|w_j\ket\right)^{M+C}.\nonumber
\end{split}
\ee
Relabeling $N\rightarrow J$ and $M+C\rightarrow J'$ and using the above homogeneity maps for $\tau , \tilde{\tau}_1, \tilde{\tau}_2$ and $\tau '^{2j}\rightarrow (2j+1)^\eta$,  we recover the result for the constrained loop identity  (\ref{fullloop}).

%%%%%%%%%%%%%%%%%%%%%%%%%%%%%%%%%%%%%%%%%%%%%%%%%%%%%%%%%%

\end{document}